\documentclass[12pt,letterpaper]{article}
\usepackage[latin9]{inputenc}
\usepackage{xcolor}
\usepackage{verbatim,booktabs}
\usepackage{bm}
\usepackage{multirow}
\usepackage{amsmath}
\usepackage{amsthm}
\usepackage{amssymb}
\usepackage{graphicx}
\usepackage{titling}
\usepackage[authoryear,round]{natbib}
\setcitestyle{aysep={}}
\makeatletter

\providecommand{\tabularnewline}{\\}

\usepackage{amsthm}\usepackage{lscape}
\usepackage{bm}
\usepackage{tabularx}
\usepackage{longtable}
\usepackage{graphicx}
\usepackage{setspace}
\usepackage{multirow}
\usepackage[font=normalsize,labelfont=bf]{caption}

\theoremstyle{plain}
\newtheorem{theorem}{Theorem}[section]
\newtheorem{lemma}{Lemma}[section]

\newtheorem{proposition}{Proposition}[section]
\theoremstyle{definition}
\newtheorem{remark}{Remark}[section]

\DeclareMathOperator{\Var}{Var}
\DeclareMathOperator{\Cov}{Cov}

\newcommand{\R}{{\mathbb R}}

\DeclareMathOperator{\E}{E}
\DeclareMathOperator{\Q}{Q}

\DeclareMathOperator*{\argmin}{arg\,min}

\providecommand{\normbig}[1]{\left\lVert#1\right\rVert}

\allowdisplaybreaks
   
\makeatother

\usepackage[letterpaper,margin=1in]{geometry}
\usepackage[hidelinks]{hyperref}
\providecommand{\blind}{0}
\date{}
\hypersetup{pdfauthor={},pdftitle={GMM for Panel Data Quantile Regression}}
\begin{document}
\title{GMM for Panel Data Quantile Regression}
\ifnum\blind=0
\author{Antonio F. Galvao\thanks{Department of Economics, Michigan State University, East Lansing,
MI 48824, USA. E-mail: \texttt{agalvao@msu.edu}.} \and Ulrich Hounyo\thanks{Department of Economics, University at Albany - State University
of New York, Albany, NY 12222, USA. E-mail: \texttt{khounyo@albany.edu}.} \and Jiahao Lin\thanks{Department of Economics, University at Albany - State University
of New York, Albany, NY 12222, USA. E-mail: \texttt{jlin28@albany.edu}.} }

\else
\author{}
\fi
\maketitle

\vspace*{-2cm} 
\begin{abstract}
\begin{singlespace}
This paper develops a smoothed generalized method of moments (GMM) framework for panel quantile models with individual-specific fixed-effects and endogenous regressors. By smoothing the indicator in the quantile moment conditions and profiling out the individual effects through smoothed individual quantile restrictions, the proposed estimator can be implemented using standard GMM algorithms.
We establish consistency, asymptotic normality, and feasible inference under weak time dependence as both panel dimensions diverge. This dependence structure accommodates dynamic panel specifications, serially dependent disturbances, lagged instruments, and within-individual heteroskedasticity and autocorrelation consistent inference. Monte Carlo simulations show that the estimator is nearly unbiased under endogeneity and yields generally accurate inference while remaining computationally tractable. An application to OECD inflation dynamics suggests asymmetric inflation propagation: under the structural IVQR interpretation, the response of current inflation to lagged inflation is substantially larger at the lower structural quantile than at the upper structural quantile.


\end{singlespace}
\end{abstract}
\vfill
\thispagestyle{empty}
\newpage
\section{Introduction}

Since the seminal work of \citet{KoenkerBassett78}, quantile regression
(QR) has become a fundamental tool in econometrics for analyzing heterogeneous effects across the distribution of an outcome variable. Unlike mean regression models, QR allows the marginal effect of covariates to vary across conditional quantiles, thereby providing a flexible framework for studying distributional heterogeneity under relatively weak assumptions.\footnote{QR is also robust to outliers and heavy-tailed distributions, which
makes it particularly attractive in many empirical applications.} Beyond its statistical advantages, QR also admits structural economic interpretations. In particular, conditional quantile restrictions may arise as equilibrium conditions derived from dynamic optimization problems under quantile-based preferences; see, for example, \citet{deCastroGalvao19}.

Motivated by the increasing availability of longitudinal data, a growing
literature has extended QR methods to panel data
environments. Panel quantile models allow researchers to control for
unobserved individual heterogeneity while simultaneously exploring distributional effects. A pioneering contribution in this area is \citet{Koenker04}, who introduced a fixed-effects quantile regression (FE-QR) estimator treating individual effects as location shifts common across quantiles. This approach provides a practical method for controlling for time-invariant heterogeneity in panel datasets while allowing heterogeneous covariate effects across the conditional distribution.

Subsequent work has substantially expanded the theoretical and methodological foundations of panel QR. Important contributions
include, among others, \citet{WeiHe06}, \citet{WangFygenson09}, \citet{Galvao11},
\citet{Canay11}, \citet{KimYang11}, \citet{KatoGalvaoMontes-Rojas12}, \citet{ChetverikovLarsenPalmer13}, \citet{GalvaoLamarcheLima13},
\citet{GalvaoWang15}, \citet{ArellanoBonhomme15}, \citet{GalvaoKato16}, \citet{MachadoSantosSilva19}, \citet{GuVolgushev19}, \citet{melly2025minimum}. These papers study estimation and asymptotic properties of panel QR estimators under various assumptions and asymptotic regimes. For comprehensive overviews of this literature, see \citet{GalvaoKato17} and \citet{GalvaoLamarche26}.

Despite these advances, an important challenge remains: the presence
of endogenous regressors. Endogeneity is pervasive in empirical applications due to omitted variables, measurement error, simultaneity, self-selection, and equilibrium behavior. Ignoring endogeneity may lead to biased estimates, incorrect inference, and misleading conclusions regarding heterogeneous effects.

To address this issue, in a standard cross-section context, \citet{ChernozhukovHansen05,ChernozhukovHansen06,ChernozhukovHansen08} develop the instrumental variables quantile regression (IVQR) framework,
which provides identification and estimation methods for models with
endogenous regressors. The IVQR approach has become a central tool
for estimating structural quantile treatment effects. Relatedly, \citet{van2026iv} develop an IV Fr{\'e}chet regression framework for distribution-valued outcomes, where group-level quantile functions are treated as outcomes in Wasserstein space. However, the practical implementation of IVQR remains challenging. See \citet{ChernozhukovHansenWuthrich17} for a review on the subject.

For estimation of IVQR models, the literature still lacks a computationally feasible procedure that simultaneously accommodates individual-specific fixed-effects, endogenous within-unit regressors, and weak serial dependence. Existing asymptotic results are largely confined to cross-sectional settings under independent and identically distributed (iid) sampling assumptions; see \citet{ChernozhukovHong03}, \citet{ChernozhukovHansen06}, and \citet{KaplanSun17}. Although the Markov Chain Monte Carlo estimator of \citet{ChernozhukovHong03} is broadly applicable, it is computationally demanding.\footnote{\Citet{ChernozhukovHansenWuthrich17} note that this approach ``bypasses the need to optimize a non-convex and non-smooth criterion at the cost of designing a sampler that adequately explores the quasi-posterior in a reasonable amount of computation time.''} The estimators of \citet{ChernozhukovHansen06} and \citet{KaplanSun17} are computationally simpler but are limited to cross-sectional data, and their computational advantages deteriorate rapidly as the number of endogenous regressors increases \citep{ChernozhukovHansen06}. Beyond the iid framework, \citet{SuYang11} consider spatial autoregressive models under a correctly specified linear model with nonstochastic regressors and instruments.\footnote{Other proposed estimators have similar limitations. \citet{ChenLee17} develop an IVQR estimator based on mixed-integer quadratic programming, but it can be computationally demanding. \citet{Wuthrich17a} considers binary treatments under iid sampling, while the Bayesian approach of \citet{LancasterJun10} is also restricted to iid data and relies on computationally intensive procedures such as grid search or Markov chain Monte Carlo methods.}

Further extensions to panel data are provided by
\citet{HardingLamarche09}, who adapt inverse IVQR to models with endogenous regressors and individual effects, and by \citet{Galvao11}, who develop estimation and inference for dynamic panel quantile models. These contributions demonstrate how IVQR can
accommodate individual-specific heterogeneity, but retain the nonsmooth inverse-quantile-regression formulation, which generally requires repeated QR estimation over candidate values of the endogenous coefficients. These inverse QR methods,  based on inverting a mapping of coefficients in the second step of the procedure rather than solving the structural QR directly, are computationally difficult to apply in practice under multiple endogenous regressors.
Relatedly, \citet{melly2025minimum} propose a two-stage
minimum-distance approach based on unit-specific quantile regressions followed by linear Generalized Method of Moments (GMM); their baseline first stage, however, relies on exogenous
within-unit variation. These considerations motivate computationally simpler IVQR procedures that are applicable to panel data settings with individual fixed-effects and serial dependence.

This paper makes three main contributions. First, we develop a framework for instrumental variables QR in panel data models with individual-specific fixed-effects using smoothing techniques, thereby allowing researchers to address endogeneity in quantile panel models while maintaining computational tractability. Our approach integrates three strands of the econometric literature: panel QR with fixed-effects, instrumental variables QR, and smoothed moment estimation methods.
Specifically, we extend the smoothed QR framework
for panel data developed by \citet{GalvaoKato16} to settings with
endogenous regressors. At the same time, we build on the smoothed
IVQR approach proposed by \citet{KaplanSun17} and the related developments in \citet{deCastroGalvaoKaplanLiu19} and \citet{FirpoGalvaoPintoPoirierSanroman22}. 

We propose a panel smoothed profiled GMM estimator that can be implemented using standard procedures, avoiding the computational difficulties associated with non-smooth optimization. The estimator incorporates the efficiency gains from a two-step GMM framework. To handle the individual-specific effects, for each candidate value of slopes, the individual-specific coefficients are determined from their smoothed quantile restrictions, and the resulting profiles substituted into the common unconditional moments. This profiling step reduces the outer GMM optimization to the finite-dimensional slope parameters and, importantly, ensures in the over-identified setting that the individual quantile restrictions are imposed exactly before GMM weighting is applied to the common instrument moments. 
The proposed estimator combines the advantages of smoothed moment conditions with the ability to control for individual-specific unobserved heterogeneity. We relax the iid sampling assumption imposed in much of the existing IVQR literature by allowing for weakly dependent panel
data. At the same time, the use of smoothing preserves the computational
simplicity, speed, and scalability of the smoothed IVQR methods proposed
by \citet{KaplanSun17}.

We highlight that, although the conditional quantile restriction motivates
the model and its structural interpretation, the estimator and the
asymptotic theory are based on the unconditional moment restrictions. Consequently, even if
the conditional quantile restriction is misspecified, the estimator
converges to the parameter identified by these unconditional moments,
provided that they remain correctly specified. Relative to the
misspecified conditional quantile model, this parameter may be viewed as
a pseudo-true parameter, complementing the interpretation results of
\citet{AngristChernozhukovFernandez-Val06}.

Second, we establish the large-sample properties of
the proposed estimator under weak dependence. This dependence structure is substantively important: it accommodates dynamic panel specifications, serially dependent disturbances, lagged instruments, and within-individual heteroskedasticity and autocorrelation consistent (HAC) inference. As part of the analysis, we  establish existence and uniform consistency of the sample profiles and show that they are locally unique and differentiable around the true parameter, which permits the profiled GMM moment to be linearized. We then establish weak consistency, a linear representation for the GMM estimator, and derive its asymptotic normal limiting distribution under the standard condition in the QR panel literature, for large cross-section, $n$, and time-series, $T$, with $n/T\to0$, which eliminates the first-order incidental-parameter bias (see, e.g., \citet{GalvaoGuVolgushev20}).

Third, we propose methods for practical inference procedures. Using the asymptotic normality of the GMM estimator, we suggest a sample analogue to estimate the limiting variance-covariance matrix. We then show that this is a consistent estimator; consequently, standard inference procedures are valid.

We provide Monte Carlo simulations evaluating the finite-sample properties of the proposed estimator and inference procedure. The results show that the proposed estimator is nearly unbiased under both exogenous and endogenous designs, and delivers coverage probabilities generally close to the nominal level, with some finite-sample undercoverage and overcoverage. In the presence of endogeneity, it substantially outperforms conventional fixed-effects quantile estimators, which remain severely biased.

Finally, we illustrate the empirical usefulness of the proposed methods by studying inflation dynamics across OECD countries. Using monthly data for 27 countries, we estimate a dynamic panel IV quantile model that allows for endogenous regressors and individual fixed-effects. Under the structural IVQR interpretation, the results suggest asymmetry in the propagation of past inflation: the response of current inflation to lagged inflation is substantially larger at the 25th structural quantile than at the 75th structural quantile. Economically, this pattern points to stronger backward-looking persistence in the lower part of the inflation distribution, whereas inflation at the upper structural quantile appears less tied to its recent history. The application illustrates how the proposed framework can uncover heterogeneous dynamic relationships while accommodating endogeneity and serial dependence.

The remainder of the paper is organized as follows. Section \ref{sec:model} presents the panel data QR model, and Section \ref{sec:est} develops the smoothed GMM estimator. Section \ref{sec:limiting} establishes the asymptotic properties of the estimator, and Section \ref{sec:inference} discusses inference procedures. Section \ref{sec:sim} reports Monte Carlo simulation results. Section \ref{sec:application} presents an empirical application. Section \ref{sec:conclusion} concludes. All proofs and additional simulation results are collected
in the Supplementary Appendix.

We conclude with some remarks on notation.  All vectors are treated
as column vectors. The notation $\bm{1}\{\cdot\}$ denotes the indicator function, $\E[\cdot]$ denotes the expectation, $\Q_{\tau}[\cdot]$ denotes the $\tau$-quantile operator, and $P(\cdot)$ denotes the probability. The distribution $N({\mu},{\sigma^{2}})$ denotes a normal distribution with mean $\mu$ and variance $\sigma^{2}$. For vectors, $\normbig{\cdot}$ denotes the Euclidean norm. Throughout the paper, the asymptotics should be understood as taking $n\to\infty$ and $T\to\infty$. Acronyms
used throughout the paper include central limit theorem (CLT), continuous mapping theorem (CMT), generalized method of moments (GMM), uniform law of large numbers (ULLN), and weak law of large numbers (WLLN).

\section{Model and identification}\label{sec:model}

We consider a panel data instrumental variables quantile model with individual-specific effects. For a fixed quantile index \(\tau\in(0,1)\), let
\(y_{it}\in\mathcal Y\subseteq\R\) be the outcome variable,
\(\bm x_{it}\in\mathcal X\subseteq\R^{d_{\beta}}\) be the vector of
explanatory variables, and \(\bm z_{it}\in\mathcal Z\subseteq\R^{d_z}\)
be the vector of instruments. The process $\{(y_{it},\bm{x}_{it},\bm{z}_{it}):t\in\mathbb{Z}\}$
is strictly stationary. The explanatory variables, \(\bm x_{it}\),
are allowed to be endogenous. The individual effect is denoted by
\(\alpha_{i0\tau}\in\mathcal A\), where \(\mathcal A\) is a compact
subset of \(\R\), and
\(\bm\alpha_{0\tau}=(\alpha_{1,0\tau},\ldots,\alpha_{n,0\tau})'
\in\mathcal A^n\). The finite-dimensional parameter of interest is
\(\bm\beta_{0\tau}\in\mathcal B\subseteq\R^{d_{\beta}}\), where
\(\mathcal B\) is compact. The compactness of
\(\mathcal A^n\times\mathcal B\) is imposed to ensure the existence of
the estimators defined below.\footnote{In our model, the individual
effects include the intercept term and may depend on the quantile index
\(\tau\). By contrast, \citet{Koenker04} restricts the individual-specific
intercepts to be common across quantiles.}

The structural residual is
\(
u_{it}(\alpha_i,\bm\beta)
=
y_{it}-\alpha_i-\bm x_{it}'\bm\beta .
\)
The population parameter \((\bm\alpha_{0\tau},\bm\beta_{0\tau})\) is
defined by the conditional quantile restriction
$\Q_{\tau}\!\left[
u_{it}(\alpha_{i0\tau},\bm\beta_{0\tau})
\mid
\bm z_{it},\alpha_{i0\tau}
\right]
=
0,
\  i=1,\ldots,n,\  t=1,\ldots,T.$
Equivalently,
\begin{equation}
\E\!\left[
1\left\{
y_{it}\leq
\alpha_{i0\tau}+\bm x_{it}'\bm\beta_{0\tau}
\right\}
-\tau
\mid
\bm z_{it},\alpha_{i0\tau}
\right]
=0,
\quad i=1,\ldots,n,\  t=1,\ldots,T.
\label{eq:qrgmmmodel}
\end{equation}
The conditional restriction in \eqref{eq:qrgmmmodel} implies two sets of
unconditional restrictions. First, 
\begin{equation}
\E\!\left[
1\left\{
y_{it}\leq
\alpha_{i0\tau}+\bm x_{it}'\bm\beta_{0\tau}
\right\}
-\tau
\right]
=0,\quad  t=1,\ldots,T,
\label{eqn:alpha-moments}
\end{equation}
for each individual \(i\), which pins down the individual effect for a given value of the common
parameter. Second, the common parameter satisfies the instrument moment
restriction
\begin{equation}
\E\!\left[
\bm z_{it}
\left(
1\left\{
y_{it}\leq
\alpha_{i0\tau}+\bm x_{it}'\bm\beta_{0\tau}
\right\}
-\tau
\right)
\right]=\bm 0,
\quad i=1,\ldots,n,\  t=1,\ldots,T.
\label{eqn:moments}
\end{equation}
The scalar restrictions in \eqref{eqn:alpha-moments} are used to profile
out the individual effects, while the vector restrictions in
\eqref{eqn:moments} identify and estimate the common parameter
\(\bm\beta_{0\tau}\).

The vector \(\bm z_{it}\) is taken as given. It may contain exogenous
components of \(\bm x_{it}\), excluded instruments, transformations, and
lagged variables. We do not include the individual dummy vector in
\(\bm z_{it}\). The individual dummies are part of the fixed-effect
parameterization and are handled through the scalar restrictions in
\eqref{eqn:alpha-moments}. Equivalently, one could write a large stacked
design matrix for the fixed-effects, but doing so is unnecessary here
and can obscure the distinction between individual effects and
instruments.

\citet{KaplanSun17} study a cross-sectional special case of
\eqref{eqn:moments} without fixed-effects and discuss optimal
instruments in \S2.1, pp.\ 108-110. Our estimator uses the unconditional
moments \eqref{eqn:alpha-moments}--\eqref{eqn:moments}.
Under \eqref{eq:qrgmmmodel} and the usual structural IVQR conditions,
including a latent-rank representation monotone in rank,
\(\bm\beta_{0\tau}\) is a structural quantile response parameter
\citep{ChernozhukovHansen05}. Otherwise, correctly specified
unconditional moments define a quantile moment parameter that need
not have a causal interpretation.

The number of individuals is denoted by \(n\), and the number of time
periods is denoted by \(T=T_n\), which may depend on \(n\). In what
follows, we omit the subscript \(n\) of \(T_n\). Although the residual
function \(u_{it}(\alpha_i,\bm\beta)\) is linear in
\((\alpha_i,\bm\beta)\), the quantile restriction is nonlinear in the
distribution of the data. Therefore, unlike in linear mean models, the
individual effects cannot generally be eliminated by differencing.

A useful application of the proposed framework is a dynamic panel quantile model.
For example, let \(\bm\beta_{0\tau}
=
(\beta_{0\tau,1},\bm\beta_{0\tau,2}')'\) and 
\(
\bm x_{it}=(y_{i,t-1},\bm w_{it}')'\), where $\bm w_{it}$ is a vector of additional covariates.
The structural equation is 
\begin{equation*}
y_{it} = 
\alpha_{i0\tau} + \beta_{0\tau,1} y_{i,t-1}
+ \bm w_{it}'\bm\beta_{0\tau,2} + u_{it},
\end{equation*}
where the residual is
\(
u_{it}(\alpha_i,\bm\beta)
=
y_{it}
-\alpha_i
-\beta_{1} y_{i,t-1}
-\bm w_{it}'\bm\beta_{2} .
\)
The model imposes the IV quantile restriction
\(
\Q_{\tau}\!\left[
u_{it}(\alpha_{i0\tau},\bm\beta_{0\tau})
\mid
\bm z_{it},\alpha_{i0\tau}
\right]
=0.
\)
Dynamic specifications of this form are common in empirical panel
applications.  In that case, lagged variables dated sufficiently far in the
past can serve as instruments under a sequential quantile restriction.
For example, if
\(
\Q_{\tau}\!\left[
u_{it}(\alpha_{i0\tau},\bm\beta_{0\tau})
\mid
\mathcal F_{i,t-2},\alpha_{i0\tau}
\right]
=0,
\)
where \(\mathcal F_{i,t-2}\) denotes the information available up to
period \(t-2\), then any measurable function of
\(\mathcal F_{i,t-2}\) can be used as an instrument. Thus one may take,
for instance,
\(
\bm z_{it}
=
(y_{i,t-2},y_{i,t-3},\bm w_{i,t-2}',\bm w_{i,t-3}')',
\)
provided that the corresponding relevance and rank conditions hold. This example illustrates why the fixed-effect IV quantile framework is
useful in dynamic settings. Unlike linear mean models, individual effects
cannot generally be eliminated by differencing in quantile models, while
lagged variables provide natural instruments when dynamic regressors are
potentially endogenous.

Given equations \eqref{eqn:alpha-moments}--\eqref{eqn:moments} and Assumption (A5) stated below, identification of
\(\bm\beta_{0\tau}\) can be understood after accounting for the individual
effects. For each candidate \(\bm\beta\), the individual effect
\(\alpha_i\) is allowed to adjust so that the individual quantile
restriction \eqref{eqn:alpha-moments} is satisfied. The common parameter
is then identified from the remaining instrument moment restrictions in
\eqref{eqn:moments}.  Establishing
global identification in a general IVQR setting is considerably more
challenging; nevertheless, \citet[Thm.\ 2 and App.\ C]{ChernozhukovHansen05}
provide several results for the IVQR framework.

\textcolor{black}{In general, each $\alpha_{i0\tau}$ and $\bm{\beta}_{0\tau}$ can depend on $\tau$, but we assume $\tau$ to be fixed throughout the paper, and hence, suppress such a dependence for notational simplicity.}

\section{Estimation}\label{sec:est}

This section introduces the panel smoothed profiled GMM estimator based on the moment conditions in equations \eqref{eqn:alpha-moments} and \eqref{eqn:moments}. The main object of interest
is the GMM estimator for the over-identified case. The exactly identified
method-of-moments estimator is presented first because it provides the
simplest version of the estimating equation and is also useful as an
initial value for the GMM optimization.

For each individual $i$, define the unsmoothed moment function
\begin{align}
H_{ni}^{u}(\alpha_i,\bm{\beta},\tau)
&\equiv
\E\left[g_{it}^{u}(\alpha_i,\bm{\beta},\tau)\right],
\label{eqn:def-M}\\
g_{it}^{u}(\alpha_i,\bm{\beta},\tau)
&\equiv
g^{u}(y_{it},\bm{x}_{it},\bm{z}_{it},\alpha_i,\bm{\beta},\tau) \equiv
\bm{z}_{it}
\left[
1\left\{
y_{it}-\alpha_i-\bm{x}_{it}'\bm{\beta}\leq0
\right\}
-\tau
\right],
\label{eqn:def-gu}
\end{align}
where the superscript ``$u$'' denotes the unsmoothed moment. The
population moment restriction in \eqref{eqn:moments} can then be written
as
$\bm{0}
=
H_{ni}^{u}(\alpha_{i0\tau},\bm{\beta}_{0\tau},\tau),
\  i=1,\ldots,n.$
The corresponding unsmoothed sample moment replaces $\E[\cdot]$
by the time-series average,
\begin{equation}
\widehat{H}_{ni}^{u}(\alpha_i,\bm{\beta},\tau)
\equiv
\frac{1}{T}\sum_{t=1}^{T}
g_{it}^{u}(\alpha_i,\bm{\beta},\tau),
\quad
\widehat{H}_{n}^{u}(\bm{\alpha},\bm{\beta},\tau)
\equiv
\frac{1}{n}\sum_{i=1}^{n}
\widehat{H}_{ni}^{u}(\alpha_i,\bm{\beta},\tau).
\label{eqn:def-Mu-hat}
\end{equation}
A direct GMM criterion based on
$\widehat{H}_{n}^{u}(\bm{\alpha},\bm{\beta},\tau)$ is generally
difficult to minimize because the indicator function makes the sample
moments discontinuous in $(\bm{\alpha},\bm{\beta})$; see, for example,
\citet[Fig.\ 1(a) and Ex.\ 3, p.\ 297]{ChernozhukovHong03}. We therefore replace the indicator by a smooth approximation.

Let $K(\cdot)$ be a kernel function satisfying
\(
\int_{-\infty}^{\infty}K(u)du=1,\) \(
G(u):=\int_{-\infty}^{u}K(v)dv,
\)
and write $G_{h_n}(u)=G(u/h_n)$, where $h_n\to0$ is a bandwidth
sequence. We do not require $K(\cdot)$ to be nonnegative; the precise
regularity conditions on $K(\cdot)$ are stated in Section \ref{subsec:ass}. Since
\(
G_{h_n}\left(-\left(y_{it}-\alpha_i-\bm{x}_{it}'\bm{\beta}\right)\right)
=
G\left(\frac{\alpha_i+\bm{x}_{it}'\bm{\beta}-y_{it}}{h_n}\right),
\)
this term is a smooth counterpart of
$1\{y_{it}\leq\alpha_i+\bm{x}_{it}'\bm{\beta}\}$. The smoothed moment
function is therefore
\begin{align}
\begin{split}
g_{it}(\alpha_i,\bm{\beta},\tau)
&\equiv
g(y_{it},\bm{x}_{it},\bm{z}_{it},\alpha_i,\bm{\beta},\tau)  \equiv
\bm{z}_{it}
\left[
G_{h_n}\left(-\left(y_{it}-\alpha_i-\bm{x}_{it}'\bm{\beta}\right)\right)
-\tau
\right],
\\
\widehat{H}_{n}(\bm{\alpha},\bm{\beta},\tau)
&\equiv
\frac{1}{n}\sum_{i=1}^{n}
\widehat{H}_{ni}(\alpha_i,\bm{\beta},\tau)\equiv
\frac{1}{nT}\sum_{i=1}^{n}\sum_{t=1}^{T}
g_{it}(\alpha_i,\bm{\beta},\tau).
\end{split}
\label{eqn:def-M-hat}
\end{align}

The smoothing step is useful for both computation and asymptotic
analysis. Computationally, it produces moments that are differentiable
in $(\bm{\alpha},\bm{\beta})$ when $G$ is smooth, which makes numerical
optimization more stable. Theoretically, smoothness allows us to use
Taylor expansions of the profiled moment conditions and to control the
effect of estimating the individual effects. Related smoothing
ideas have been used in conditional quantile estimation by
\citet{Horowitz98}, \citet{GalvaoKato16}, and \citet{KaplanSun17}.

\subsection{Method of moments under exact identification}

\label{sec:est-MM}

Let $\bm{\theta}=(\bm{\alpha},\bm{\beta})$. In the exactly identified
case, \(d_z=d_\beta\), the smoothed moment equations determine the estimator directly. The individual effects are treated as nuisance parameters and are
profiled out by imposing the smoothed individual quantile restrictions.
Here and below, profiling (or concentrating) means solving the nuisance
restriction conditional on each candidate value of the common parameter:
\begin{equation}\label{eqn:def-est-MM-A}
\widehat{A}_{i}(\alpha_i,\bm{\beta})
\equiv
\frac{1}{T}\sum_{t=1}^{T}
\left[
G_{h_n}\left(
-\left(y_{it}-\alpha_i-\bm{x}_{it}'\bm{\beta}\right)
\right)
-\tau
\right]
=
0,
\qquad i=1,\ldots,n.
\end{equation}
For each \(\bm\beta\), the solution \(\widehat\alpha_i(\bm\beta)\) is
substituted into the common-parameter moment restriction, which we will further discuss in Section \ref{sec: profiled gmm}. 
The smoothed panel method-of-moments estimator
$(\widehat{\bm{\alpha}}_{\mathrm{MM}},
\widehat{\bm{\beta}}_{\mathrm{MM}})$ solves $\widehat{A}_{i}(\widehat{\alpha}_{\mathrm{MM},i},\widehat{\bm{\beta}}_\mathrm{MM})=0$ for $i=1,\ldots,n$, and 
\begin{equation}
\widehat{H}_{n}(\widehat{\bm{\theta}}_{\mathrm{MM}},\tau)
=
\frac{1}{nT}\sum_{i=1}^{n}\sum_{t=1}^{T}
\left[
G_{h_n}\left(
-\left(
y_{it}
-\widehat{\alpha}_{\mathrm{MM},i}
-\bm{x}_{it}'\widehat{\bm{\beta}}_{\mathrm{MM}}
\right)
\right)
-\tau
\right]\bm{z}_{it}
=
\bm{0}.
\label{eqn:def-est-MM}
\end{equation}
The right-hand side of \eqref{eqn:def-est-MM} may be replaced by
$o_p((nT)^{-1/2})$ without affecting the first-order asymptotic results.
We call $\widehat{\bm{\beta}}_{\mathrm{MM}}$ the smoothed panel MM
(SP-MM) estimator of $\bm{\beta}_{0\tau}$.

Although our main focus is GMM, the MM estimator provides a transparent
benchmark and a natural starting value for GMM optimization, whose
potentially non-convex criterion makes initialization important.
With a differentiable $G_{h_n}(\cdot)$ and bandwidth not too close to
zero, the smoothed MM equations are typically easier to solve and
check numerically.

Unsmoothed fixed-effects quantile regression may suffer from an
incidental-parameter problem in short panels.
For example, \citet{GrahamHahnPowell09} show that, with $T=2$ in a
linear location model, the estimated common slope does not depend
on $\tau$, although the population quantile slope generally does.
Thus, the standard fixed-effects QR estimator need not be consistent
when $T$ is small.

\begin{remark}
The estimating equations in \eqref{eqn:def-est-MM} are closely related
to the smoothed first-order conditions studied by \citet{GalvaoKato16}.
Their first-order conditions contain additional terms involving
$\widetilde K(u):=uK(u)$. \citet[Theorem 3.1]{GalvaoKato16} show that
these additional terms do not affect the first-order asymptotic
distribution of the common-parameter estimator. In this sense, the
present smoothed MM equations keep the first-order component of the
smoothed quantile score while avoiding terms that are asymptotically
negligible for the distribution of $\widehat{\bm{\beta}}$.
\end{remark}

\subsection{Two-step {GMM} estimator under over-identification}

\label{sec:est-GMM}

We now consider the possibly over-identified case, the main focus of
our theory. With more instruments than parameters, the equations in
\eqref{eqn:def-est-MM} generally have no exact solution, so GMM minimizes
a weighted quadratic norm of the smoothed sample moments.
Here, ``two-step GMM'' refers to updating the weighting matrix,
not a first-stage data transformation.

Given a symmetric positive definite weighting matrix $\widehat W$, the
one-step smoothed panel GMM estimator is defined by
\begin{align}
\widehat{\bm{\theta}}_{\mathrm{GMM,1s}}
=
\argmin_{\bm{\theta}\in\mathcal{A}^{n}\times\mathcal{B}}
\widehat{H}_{n}(\bm{\theta},\tau)'
\widehat{W}
\widehat{H}_{n}(\bm{\theta},\tau),
\
\text{s.t. } \ 
\widehat{A}_{i}(\alpha_i,\bm{\beta})=0,\quad i=1,\ldots,n.
\label{eqn:def-est-GMM}
\end{align}
The weighting matrix determines how the over-identifying restrictions
are combined. For efficiency, the relevant weighting matrix is the inverse
of the long-run variance of the profile moment, $\Omega$.
Let
$\bar{\Omega}$ be a consistent estimator of this long-run variance,
constructed from the initial estimator $\bar{\bm{\theta}}$. In practice, $\bar{\bm{\theta}}$ can be taken as the one-step
estimator in \eqref{eqn:def-est-GMM} computed with a fixed symmetric
positive definite matrix $W_0$. Alternatively, a prespecified
$d_\beta$-dimensional full-rank subset or linear combination of the
instruments may be used in the preliminary step. The efficient two-step smoothed panel GMM
estimator is then
\begin{equation}
\widehat{\bm{\theta}}_{\mathrm{GMM,2s}}
=
\argmin_{\bm{\theta}\in\mathcal{A}^{n}\times\mathcal{B}}
\widehat{H}_{n}(\bm{\theta},\tau)'
\bar{\Omega}^{-1}
\widehat{H}_{n}(\bm{\theta},\tau),
\
\text{s.t. } \ 
\widehat{A}_{i}(\alpha_i,\bm{\beta})=0,\quad i=1,\ldots,n .
\label{eqn:def-est-2s}
\end{equation}
Thus, for the two-step estimator, the weighting matrix denoted below by
$\widehat W$ is $\bar\Omega^{-1}$. The asymptotic theory focuses on its
common-parameter component.

The constrained GMM problem in equation \eqref{eqn:def-est-2s} may be non-convex. For this reason, the numerical implementation may use the simulated
annealing algorithm from the \texttt{GenSA} package in R \citep{R.GenSA},
which is designed for global optimization problems. Since simulated
annealing can still be sensitive to starting values in finite samples,
we initialize the GMM search using the smoothed MM estimator in
\eqref{eqn:def-est-MM}. This initialization uses the same smoothed
moment structure as the GMM estimator and therefore provides a coherent
starting point for the over-identified criterion.

After obtaining the smoothed GMM estimator, one could 
re-optimize the corresponding unsmoothed GMM criterion. We do not take
this as the baseline estimator. In linear iid IVQR models,
\citet{KaplanSun17} find that smoothing can improve pointwise mean
squared error, although it may reduce the estimated variation
across quantiles. Thus, the finite-sample benefit of an additional
unsmoothed optimization step is not clear. The statistical theory in this paper therefore focuses on the smoothed GMM estimator defined above.


\section{Asymptotic properties}

\label{sec:limiting}

This section establishes the large-sample properties of the panel smoothed
profiled GMM estimator introduced in Section \ref{sec:est-GMM}.  The dimensions \(d_\beta\) and \(d_z\)
are fixed as \(n,T\to\infty\).
The asymptotic argument treats the individual effects as nuisance
parameters. For this reason, the relevant common-parameter moment
is the moment obtained after profiling out the effect of estimating
$\alpha_{i}$.

\subsection{Profiled GMM representation}
\label{sec: profiled gmm}
For each \(i\) and \(\bm\beta\in\mathcal B\), define the sample root
correspondence
\begin{align}
\mathcal R_{ni}(\bm\beta)
=
\left\{
\alpha\in\mathcal A:
\widehat{A}_{i}(\alpha,\bm\beta)=0
\right\},\label{eq: sample root}
\end{align}
where \(\widehat{A}_{i}\) is defined in \eqref{eqn:def-est-MM-A}.
Because a higher-order kernel need not be nonnegative, the sample
equation need not be globally monotone in \(\alpha\), and
\(\mathcal R_{ni}(\bm\beta)\) may contain multiple roots in finite
samples. Lemma \ref{lem:profile_existence} below establishes that, with
probability approaching one, these sets are nonempty uniformly over
\(i\le n\) and \(\bm\beta\in\mathcal B\).

For each \(\bm\beta\in\mathcal B\), let
\(
\mathcal S_n(\bm\beta)
=
\argmin_{\bm\alpha\in\prod_{i=1}^n\mathcal R_{ni}(\bm\beta)}
\widehat H_n(\bm\alpha,\bm\beta,\tau)'
\widehat W
\widehat H_n(\bm\alpha,\bm\beta,\tau).
\)
Since \(\prod_{i=1}^n\mathcal R_{ni}(\bm\beta)\) is nonempty and compact
and the criterion is continuous, \(\mathcal S_n(\bm\beta)\) is nonempty
and compact with probability approaching one. If it contains more than
one vector, choose the one with the smallest first component. If several
such vectors remain, choose among them the one with the smallest second
component, and continue in this way through the remaining components.
This fixed rule defines a unique measurable choice
\(\widehat{\bm\alpha}(\bm\beta)\in\mathcal S_n(\bm\beta)\).

Define the profiled smoothed sample moment by
\[
\widehat m_n(\bm\beta)
=
\widehat H_n(\widehat{\bm\alpha}(\bm\beta),\bm\beta,\tau)
=
\frac1{nT}\sum_{i=1}^n\sum_{t=1}^T
\bm z_{it}
\left\{
G_{h_n}\!\left(
-y_{it}
+\widehat\alpha_i(\bm\beta)
+\bm x_{it}'\bm\beta
\right)
-\tau
\right\}.
\]
The common-parameter component of the joint constrained GMM estimator in \eqref{eqn:def-est-GMM}
then has the profiled representation
\begin{align}
\widehat{\bm\beta}
\in
\argmin_{\bm\beta\in\mathcal B}
\widehat m_n(\bm\beta)'
\widehat W
\widehat m_n(\bm\beta).
\label{eq:profile}
\end{align}
If the minimum is attained at more than one value of \(\bm\beta\), any
one of these values may be taken as \(\widehat{\bm\beta}\). The results
below hold for every such minimizer.

Similarly, for each \(i\) and \(\bm\beta\in\mathcal B\), let
\(\alpha_i(\bm\beta)\) denote the population constrained fixed-effect.
The collection
\(\{\alpha_i(\bm\beta):i\ge1,\bm\beta\in\mathcal B\}\)
is contained in a compact subset of
\(\operatorname{int}(\mathcal A)\) and satisfies
\(
\E\left[
1\left\{
y_{it}-\alpha_i(\bm\beta)-\bm x_{it}'\bm\beta\le0
\right\}
-\tau
\right]
=0.
\)
The profiled population moment is
\begin{align}
m_n(\bm\beta)
=
\frac1{nT}\sum_{i=1}^n\sum_{t=1}^T
\E\left[
\bm z_{it}
\left\{
1\left(
y_{it}-\alpha_i(\bm\beta)-\bm x_{it}'\bm\beta\le0
\right)-\tau
\right\}
\right],
\label{eq: global obj}
\end{align}
and the population GMM criterion is
\(\Q_n(\bm\beta)=m_n(\bm\beta)'Wm_n(\bm\beta)\).

Proposition \ref{prop1} below establishes
\(\widehat{\bm\beta}\stackrel{p}{\to}\bm\beta_0\). Lemma A.4
{} in Supplementary Appendix A then shows that, with probability
approaching one, the sample fixed-effect solutions
\(\widehat\alpha_i(\bm\beta)\) are unique and continuously
differentiable throughout a shrinking neighborhood of
\(\bm\beta_0\) containing \(\widehat{\bm\beta}\). In this neighborhood, the criterion-minimizing
profile therefore coincides with the unique sample profile, which
justifies the differentiations and mean-value expansions used in the proof.

Profiling reduces the outer optimization from \(n+d_\beta\) parameters
with \(n\) equality constraints to \(d_\beta\) parameters. For each
candidate \(\bm\beta\), the \(n\) scalar root equations are
individual-specific and can be solved separately or in parallel. When
multiple roots occur, however, the additional minimization over their
possible combinations may be computationally costly. In practice,
either the joint constrained or profiled formulation may be used
according to numerical stability.

Both formulations in \eqref{eqn:def-est-GMM} and \eqref{eq:profile} differ from unconstrained GMM, which stacks the \(n\) individual restrictions and the \(d_z\) common instrument moments in a single quadratic criterion. For example, with a block-diagonal weighting matrix, such an estimator minimizes
$$
\frac{1}{n}\sum_{i=1}^{n}\widehat{A}_{i}(\alpha_i,\bm\beta)^2
+
\widehat H_n(\bm\alpha,\bm\beta,\tau)'
\widehat W
\widehat H_n(\bm\alpha,\bm\beta,\tau)
$$
over \((\bm\alpha,\bm\beta)\). When \(d_z>d_\beta\), the \(n+d_z\) sample moment equations generally cannot all be satisfied by the \(n+d_\beta\) parameters. Minimizing the stacked criterion may therefore allow \(\widehat{A}_{i}(\alpha_i,\bm\beta)\neq0\) for some individuals in exchange for a better fit to the common instrument moments. Our estimator instead requires
$
\widehat{A}_{i}\bigl(\widehat\alpha_i(\bm\beta),\bm\beta\bigr)=0,
\  i=1,\ldots,n,
$
for every candidate \(\bm\beta\), and applies GMM weighting only to the resulting common instrument moments.

Under exact identification, whenever a joint solution exists, all sample moments can be set to zero, so profiling simply eliminates the nuisance parameters from the same estimating system, as in exactly identified smoothed score systems; see, for example, \citet{GalvaoKato16}. Under over-identification, however, imposing the individual restrictions exactly is a substantive part of the estimator's definition: these restrictions are enforced as constraints rather than balanced against the common instrument moments. Profiling implements this requirement while reducing the dimension of the optimization.

This nested structure is computationally related to the
inverse-quantile-regression procedure of
\citet{ChernozhukovHansen06}, which solves an auxiliary quantile
regression for each candidate value of the endogenous coefficient. Nevertheless, in
our setting, the inner step consists of \(n\) scalar smoothed equations,
and local uniqueness yields a differentiable concentrated GMM criterion
near \(\bm\beta_0\). Our estimator is also conceptually distinct from
the two-step transformation of \citet{Canay11}, who first estimates
quantile-invariant individual location effects from a preliminary
conditional-mean model and then applies quantile regression to the
transformed outcome. By contrast, we determine
\(\widehat\alpha_i(\bm\beta)\) from an individual smoothed quantile
restriction for each candidate \(\bm\beta\), within the GMM estimation
problem itself. Our approach therefore requires neither a preliminary
conditional-mean model nor a quantile-invariant location-shift
restriction.
\subsection{Assumptions}
\label{subsec:ass}
Throughout this section, let
$u_{it}=u_{it}(\alpha_{i0},\bm\beta_0)
=y_{it}-\alpha_{i0}-\bm x_{it}'\bm\beta_0$ denote the residual at the
true parameter. The moment restrictions
\eqref{eqn:alpha-moments} and \eqref{eqn:moments} are maintained. The following
conditions are stated for the asymptotic distribution.
\begin{description}
\item [{(A1)}] For each $i\geq1$, the process $\{(y_{it},\bm{x}_{it},\bm{z}_{it}):t\in\mathbb{Z}\}$
is strictly stationary and $\beta$-mixing. Let $\beta_{i}(j)$ denote
its $\beta$-mixing coefficient. There exist constants $a\in(0,1)$
and $B<\infty$ such that $\sup_{i\geq1}\beta_{i}(j)\leq Ba^{j}$.
The processes are independent across $i$.
\item [{(A2)}] There exist constants $M_1$ and $M_2$ such that, almost surely,
\(
\sup_{i\ge1}\|\bm x_{i1}\|\le M_1\) and \(
\sup_{i\ge1}\|\bm z_{i1}\|\le M_2.
\)
\item [{(A3)}] The parameter sets $\mathcal{A}\subset\mathbb{R}$ and $\mathcal{B}\subset\mathbb{R}^{d_{\beta}}$
are compact. For each $i$, $\alpha_{i0}$ is an interior point of
$\mathcal{A}$, and $\bm{\beta}_{0}$ is an interior point of $\mathcal{B}$.
\item [{(A4)}] Let $\bm{q}_{it}=(\bm{x}_{it}',\bm{z}_{it}')'$. The conditional
distribution of $u_{it}$ given $\bm{q}_{it}$ has a density $f_{i}(u\mid\bm{q})$.
The density $f_{i}(0)=\E[f_{i}(0\mid\bm{q}_{it})]$ is bounded away
from zero uniformly in $i$.
\item [{(A5)}] For every \(i\geq1\) and \(\bm\beta\in\mathcal B\),
the population fixed-effect restriction has a solution
\(\alpha_i(\bm\beta)\), and the collection
\(\{\alpha_i(\bm\beta):i\geq1,\bm\beta\in\mathcal B\}\) is contained
in a compact subset of \(\operatorname{int}(\mathcal A)\). Moreover,
\(\alpha_i(\bm\beta_0)=\alpha_{i0}\) and \(\Q_n(\bm{\beta}_0)=0\).
For every \(\varepsilon>0\),
there exist \(c_{\varepsilon,\alpha},c_{\varepsilon,\beta}>0\) such that \(
\liminf_{n\to\infty}
\inf_{\substack{\bm\beta\in\mathcal B:\\
\|\bm{\beta}-\bm{\beta}_0\|\ge\varepsilon}}
\left\Vert m_n(\bm{\beta})\right\Vert
\ge c_{\varepsilon,\beta},
\) and 
\[
\liminf_{n\to\infty}
\inf_{i\le n}\inf_{\bm{\beta}\in\mathcal B}
\inf_{\substack{\alpha\in\mathcal A:\\
|\alpha-\alpha_i(\bm{\beta})|\ge\varepsilon}}
\left|
\E\left[1\left\{ y_{it}-\alpha-\bm{x}_{it}'\bm{\beta}\leq0\right\}-\tau\right]
\right|
\ge c_{\varepsilon,\alpha}
\]

\item [{(A6)}]
For each \(i\), \(f_i(u\mid\bm q)\) is \(r\) times continuously
differentiable in \(u\in\mathbb R\), where \(r\ge4\) is an even
integer. Its derivatives up to order \(r\) are uniformly bounded over
\(i\), \(u\), and \(\bm q\). In addition, there exist constants
\(\varepsilon_0>0\), \(\vartheta>0\), and \(C<\infty\) such that, with
\(
\Theta_i(\varepsilon_0)
=
\{
(\alpha,\bm\beta)\in\mathcal A\times\mathcal B:
|\alpha-\alpha_{i0}|
+\|\bm\beta-\bm\beta_0\|
\le\varepsilon_0
\},
\)
for all \(0<a,b\le\vartheta\),
\[
\sup_{i\ge1}\sup_{j\ge1}
\sup_{\substack{
(\alpha,\bm\beta)\in\Theta_i(\varepsilon_0)\\
(\widetilde\alpha,\widetilde{\bm\beta})
\in\Theta_i(\varepsilon_0)
}}
P\left(
\left|u_{i1}(\alpha,\bm\beta)\right|\le a,\,
\left|u_{i,1+j}
(\widetilde\alpha,\widetilde{\bm\beta})\right|\le b
\right)
\le Cab.
\]

\item [{(A7)}] Let
\(
D_n=\frac1n\sum_{i=1}^n
\E\!\left[f_i(0\mid\bm q_{it})\bm z_{it}(\bm x_{it}-\bm\gamma_i)'\right],
\)
where
\(
s_i=\{f_i(0)\}^{-1}\), \(
\bm\gamma_i=s_i\E[f_i(0\mid\bm q_{it})\bm x_{it}],
\)
and suppose that $D_{n}\to D$. Assume $\widehat W\stackrel{p}{\to}W$,
where $W$ is symmetric positive definite. Finally,
suppose that $D'WD$ is nonsingular.
\item [{(A8)}] The kernel $K(\cdot)$ is symmetric about zero, has bounded
support, is continuously differentiable, and has a globally Lipschitz
first derivative. It is piecewise twice continuously differentiable, with
bounded derivatives on each piece.
$K(\cdot)$ is an $r$-th order kernel.
\item [{(A9)}] The bandwidth satisfies $h_{n}\asymp T^{-c}$ for some $1/r<c<1/3$.
\item [{(A10)}] Let $\Omega_{n}=\frac{1}{n}\sum_{i=1}^{n}\Omega_{ni}$
where 
\begin{align*}
\Omega_{ni}&=\Var\!\left(T^{-1/2}\sum_{t=1}^{T}\psi_{it}\right),
&\psi_{it}&=(\bm{z}_{it}-\bm{\delta}_{i})\{1(u_{it}\leq0)-\tau\},
&\bm{\delta}_{i}=s_{i}\E[f_{i}(0\mid\bm{q}_{it})\bm{z}_{it}].
\end{align*}
Assume that $\Omega_{n}\to\Omega$, where $\Omega$ is finite and
positive definite. In addition, let
\(
\omega_{it}=G_{h_n}(-u_{it})-\tau,\) \(
\bm\kappa_{it}
=
(\bm z_{it}-\bm\delta_i)K_{h_n}(-u_{it})
\), \(\bm\Lambda_{ni}
=
\sum_{|j|\le T-1}(1-|j|/T)
\E[\omega_{i1}\{\bm\kappa_{i,1+j}-\E[\bm\kappa_{it}]\}]\),
and assume
\(\sup_{i\le n}\|\bm\Lambda_{ni}\|=O(1)\).
\end{description}
\noindent
The assumptions above are standard in smoothed quantile estimation,
fixed-effects panel quantile regression, and GMM. Assumption (A1)
allows time-series dependence within individuals while maintaining
cross-individual independence. Its geometric $\beta$-mixing condition,
also used by \citet{GalvaoKato16}, facilitates uniform convergence
rates for kernel-type statistics and covers many weakly dependent
processes, including standard ARMA-type processes under regularity
conditions. Assumption (A2) imposes bounded covariates and instruments
to simplify the empirical-process arguments, as in
\citet{GalvaoKato16} and \citet{KaplanSun17}.
Assumption (A3) is the usual compactness and interior-point condition
for extremum estimation.

Assumption (A4) is the standard density condition for quantile
regression. The conditional density near zero permits local
linearization of the indicator moment, while bounding $f_i(0)$ away
from zero ensures local identification of the individual effect and
a well-behaved profile derivative. Assumption (A5) provides global
identification for consistency: its individual-effect restriction separates
$\alpha_i(\bm{\beta})$ from other individual-effect values for each
$\bm{\beta}$, and its profiled-moment restriction separates the profiled population
moment from zero away from $\bm{\beta}_0$. Related conditions appear
in smoothed IVQR and smoothed GMM quantile models
\citep{KaplanSun17,deCastroGalvaoKaplanLiu19}.

Assumption (A6) strengthens marginal density smoothness to control
higher-order smoothing bias. Its second part imposes uniform local
anti-concentration on pairs of residual indices, providing small-ball
bounds under within-individual serial dependence. This is
\emph{strictly weaker} than requiring a uniformly bounded joint
conditional density given regressors at both dates, as in
\citet{GalvaoKato16}. The weaker condition accommodates dynamic
specifications, where conditioning on adjacent regressor vectors may
reveal the earlier outcome and make the joint conditional distribution
singular. The requirement \(r\ge4\) is consistent with
\citet{Horowitz98} and \citet{KaplanSun17}.

Assumption (A7) imposes local full rank for asymptotic linearization
and normality. Nonsingularity of $D'WD$ is the usual GMM rank condition
for the common-parameter estimator. Thus, (A5) and (A7) serve distinct
purposes: global separation delivers consistency, whereas local rank
ensures an invertible GMM linearization.
Assumption (A8) specifies kernel regularity. The higher-order kernel
eliminates lower-order smoothing bias terms, leaving bias of order
$h_n^r$. As is standard for higher-order kernels, nonnegativity is
not required.

Assumption (A9) specifies the bandwidth rate. The lower bound
$c>1/r$ makes $h_n^r$ negligible at the rate required for the
asymptotic expansion, while the upper bound $c<1/3$ prevents excessive
variability in kernel terms and their derivatives. Such restrictions
are standard in smoothed quantile estimation
\citep{Horowitz98,GalvaoKato16}.
Finally, Assumption (A10) requires the average long-run variance of
the profiled quantile scores $\psi_{it}$ to converge to a finite
positive definite matrix. Here, $\bm{\delta}_i$ removes the
first-order effect of estimating the individual effect.
Its second part controls the long-run covariance term generating
the order \(T^{-1}\) incidental-parameter bias, providing the
profiled-GMM analogue of the long-run bias condition in
\citet{GalvaoKato16}.

\subsection{Profiling properties and limiting distribution}

\label{sec:asy-normality}

The following result gives the profiling properties and limiting
distribution of the proposed two-step smoothed GMM estimator. 
We first establish that the sample profile is well defined with
probability approaching one.

\begin{lemma}[Existence of the sample profile]
\label{lem:profile_existence}
Suppose that Assumptions (A1)-(A8) hold, \((\log n)/\sqrt{T}\to0\), and \(h_n\to0\). 
Then the smoothed fixed-effect equation $
\widehat{A}_{i}(\alpha,\bm\beta)=0
$ in \eqref{eq: sample root} admits a solution in \(\mathcal A\), uniformly over \(i\le n\) and \(\bm\beta\in\mathcal B\), with probability approaching one.
\end{lemma}

Lemma A.4
{} further strengthens the existence result locally.  We next show consistency of the estimator. As in \citet{KatoGalvaoMontes-Rojas12},
$(\widehat{\bm{\alpha}},\widehat{\bm{\beta}})$ is said to be weakly
consistent if $\max_{1\leq i\leq n}|\widehat{\alpha}_{i}-\alpha_{i0}|\stackrel{p}{\to}0$
and $\widehat{\bm{\beta}}\stackrel{p}{\to}\bm{\beta}_{0}$. 

\begin{proposition}[Consistency] \label{prop1} Suppose that the moment
restrictions \eqref{eqn:alpha-moments}-\eqref{eqn:moments} and
Assumptions (A1)-(A8) hold, $(\log n)/\sqrt{T}\to0$, and
$h_{n}\to0$. Then the joint constrained estimator, equivalently the profiled
estimator in \eqref{eq:profile}, is weakly consistent. Moreover, every measurable selection
\(\widehat\alpha_i(\bm\beta)\in\mathcal R_{ni}(\bm\beta)\) satisfies
\(
\max_{1\le i\le n}\sup_{\bm\beta\in\mathcal B}
\left|\widehat\alpha_i(\bm\beta)-\alpha_i(\bm\beta)\right|
\stackrel{p}{\longrightarrow}0.
\)
\end{proposition}

The following result establishes the asymptotic normality of the estimator. 

\begin{theorem}[Asymptotic normality] \label{thm1} Suppose that the moment
restrictions \eqref{eqn:alpha-moments}-\eqref{eqn:moments} and Assumptions
(A1)-(A10) hold, and $n/T\to0$. Then
\begin{align}
\sqrt{nT}(\widehat{\bm{\beta}}-\bm{\beta}_{0})=-(D'WD)^{-1}D'W\frac{1}{\sqrt{nT}}\sum_{i=1}^{n}\sum_{t=1}^{T}\psi_{it}+o_{p}(1).\label{eq: expansion betahat}
\end{align}
Consequently, 
\(
\sqrt{nT}(\widehat{\bm{\beta}}-\bm{\beta}_{0})\stackrel{d}{\to}N(0,\Sigma_{\beta}),
\)
where $\Sigma_{\beta}=(D'WD)^{-1}D'W\Omega WD(D'WD)^{-1}.$ If $W=\Omega^{-1}$,
then the asymptotic variance simplifies to $\Sigma_{\beta}=(D'\Omega^{-1}D)^{-1}.$
\end{theorem}

\begin{remark} The condition $n/T\to0$ is imposed to obtain a centered
limiting distribution without explicit bias correction. The
incidental-parameter bias from estimating the individual effects is of
order \(T^{-1}\), and its contribution after \(\sqrt{nT}\)
normalization is therefore of order \(\sqrt{n/T}\). If $n/T$
converges to a positive constant, this term generally remains in the
limit and would require explicit bias correction, which is outside the
scope of the present analysis. \end{remark}

To make the two bias sources and their rates transparent, it is useful to display the leading terms associated with each source. The following discussion is heuristic and is not used as a separate second-order result in the proof of Theorem \ref{thm1}. Define
\[
\Gamma_{ni}
=
\sum_{|j|\leq T-1}
\left(1-\frac{|j|}{T}\right)
\E\!\left[
\left\{\omega_{i1}-\E\omega_{i1}\right\}\left\{\omega_{i,1+j}-\E\omega_{i,1+j}\right\}
\right].
\]
Further, let
\(
\bm\eta_i
=
\E\!\left[
f_i^{(1)}(0\mid\bm q_{it})
(\bm z_{it}-\bm\delta_i)
\right],
\)
where \(f_i^{(1)}(0\mid\bm q)\) denotes the first derivative of
\(f_i(u\mid\bm q)\) with respect to \(u\), evaluated at zero. Define
the leading incidental-parameter bias by
\(
\mathcal B_{IP,n}
=
(D'WD)^{-1}D'W\sqrt{\frac nT}\,
\frac1n\sum_{i=1}^{n}
\left\{
s_i\bm\Lambda_{ni}
-\frac12s_i^2\Gamma_{ni}\bm\eta_i
\right\},
\)
and the smoothing bias by
\[
\mathcal B_{S,n}
=
-(D'WD)^{-1}D'W\sqrt{nT}\,
\frac1n\sum_{i=1}^{n}
\E\!\left[
(\bm z_{it}-\bm\delta_i)
\{G_{h_n}(-u_{it})-1(u_{it}\leq0)\}
\right].
\]
At the level of leading second-order terms, the decomposition can be written as
\[
\sqrt{nT}(\widehat{\bm\beta}-\bm\beta_0)
=
-(D'WD)^{-1}D'W\frac1{\sqrt{nT}}
\sum_{i=1}^{n}\sum_{t=1}^{T}\psi_{it}
+\mathcal B_{IP,n}
+\mathcal B_{S,n}+o_P(1).
\]
The first component of \(\mathcal B_{IP,n}\), involving
\(\bm\Lambda_{ni}\), is generated by the covariance between the
individual quantile score and the derivative of the profiled moment.
The second component, involving \(\Gamma_{ni}\bm\eta_i\), is the
curvature bias generated by the second-order effect of estimating the
individual fixed-effects. The corresponding rate bounds under Assumptions (A1)-(A10) are
\(
\|\mathcal B_{IP,n}\|
=
O\!\left(\sqrt{\frac nT}\right)\) and \(
\|\mathcal B_{S,n}\|
=
O\!\left(\sqrt{nT}\,h_n^r\right).
\)
These rates explain the restrictions used in Theorem \ref{thm1}:
\(n/T\to0\) removes the incidental-parameter contribution, while
Assumption (A9), together with \(n/T\to0\), implies
\(\sqrt{nT}h_n^r\to0\) and removes the smoothing contribution.

\section{Practical inference}\label{sec:inference}

In this section, we develop feasible inference based on the limiting
representation in Theorem \ref{thm1}. Estimating the individual effects
changes the influence function: the original instrument \(\bm z_{it}\)
is replaced by the residualized instrument
\(\bm z_{it}-\bm\delta_i\), which removes the first-order effect of
estimating \(\alpha_i\). This distinguishes the present covariance
formula from cross-sectional smoothed IVQR inference in
\citet{KaplanSun17}, which does not incorporate estimated panel fixed
effects or within-individual dependence. Relative to the smoothed panel
quantile results of \citet{GalvaoKato16}, our procedure additionally
allows endogenous regressors and over-identified GMM. Let
\(
\widehat{u}_{it}=y_{it}-\widehat{\alpha}_{i}-\bm{x}_{it}'\widehat{\bm{\beta}}\) and \(K_{h_{n}}(u)=h_{n}^{-1}K(u/h_{n}).
\)
Estimate the density and the two projection terms by 
\[
\widehat{\bm{\gamma}}_{i}=\widehat{s}_{i}\frac{1}{T}\sum_{t=1}^{T}K_{h_{n}}(\widehat{u}_{it})\bm{x}_{it},\qquad\widehat{\bm{\delta}}_{i}=\widehat{s}_{i}\frac{1}{T}\sum_{t=1}^{T}K_{h_{n}}(\widehat{u}_{it})\bm{z}_{it}.
\]
where $\widehat{f}_{i}=\frac{1}{T}\sum_{t=1}^{T}K_{h_{n}}(\widehat{u}_{it}),\ \widehat{s}_{i}=\widehat{f}_{i}^{-1}.$
The Jacobian matrix $D$ is estimated by $\widehat{D}=\frac{1}{nT}\sum_{i=1}^{n}\sum_{t=1}^{T}K_{h_{n}}(\widehat{u}_{it})\bm{z}_{it}(\bm{x}_{it}-\widehat{\bm{\gamma}}_{i})'.$
Define $\widehat{\psi}_{it}=(\bm{z}_{it}-\widehat{\bm{\delta}}_{i})\{1(\widehat{u}_{it}\leq0)-\tau\}.$
Let $\ell_n$ be a truncation lag such that $\ell_n\to\infty$ and
$\ell_n^{2}\log(nT)/(Th_{n}^{2})\to0$.
With the Bartlett kernel $k(a)=(1-|a|)_{+}$, define 
\[
\widehat{\Omega}=\frac{1}{n}\sum_{i=1}^{n}\widehat{\Omega}_{i},\qquad\widehat{\Omega}_{i}=\frac{1}{T}\sum_{t=1}^{T}\widehat{\psi}_{it}\widehat{\psi}_{it}'+\sum_{j=1}^{\ell_n}k(j/(\ell_n+1))\frac{1}{T}\sum_{t=1}^{T-j}\left(\widehat{\psi}_{it}\widehat{\psi}_{i,t+j}'+\widehat{\psi}_{i,t+j}\widehat{\psi}_{it}'\right).
\]
For a general weighting matrix $\widehat{W}$, the asymptotic covariance
matrix is estimated by 
\(
\widehat{\Sigma}_{\beta}=(\widehat{D}'\widehat{W}\widehat{D})^{-1}\widehat{D}'\widehat{W}\widehat{\Omega}\widehat{W}\widehat{D}(\widehat{D}'\widehat{W}\widehat{D})^{-1}.
\)
For the efficient two-step estimator, $\widehat W=\bar\Omega^{-1}$ in the
point-estimation criterion. Since both $\bar\Omega$ and the final
$\widehat\Omega$ consistently estimate $\Omega$, a convenient efficient
covariance estimator is
$\widehat{\Sigma}_{\beta}=(\widehat{D}'\widehat{\Omega}^{-1}\widehat{D})^{-1}$.

\begin{theorem}[Feasible covariance estimation] \label{thm2} Suppose that the conditions of Theorem
\ref{thm1} hold and that $\ell_{n}\to\infty$ with
$\ell_{n}^{2}\log(nT)/(Th_{n}^{2})\to0$.
Then, 
$\widehat{\Sigma}_{\beta}\stackrel{p}{\to}\Sigma_{\beta}$. 
Consequently,
for any fixed nonzero vector $a\in\mathbb{R}^{d_{\beta}}$, 
\(
\frac{\sqrt{nT}\,a'(\widehat{\bm{\beta}}-\bm{\beta}_{0})}{(a'\widehat{\Sigma}_{\beta}a)^{1/2}}\stackrel{d}{\to}N(0,1).
\)
\end{theorem}

For testing linear restrictions $H_{0}: R\bm{\beta}_{0}=r$, where $R$ is
a $q\times d_{\beta}$ matrix with full row rank, the Wald statistic
is 
\(
\mathcal{W}_{n}=nT(R\widehat{\bm{\beta}}-r)'(R\widehat{\Sigma}_{\beta}R')^{-1}(R\widehat{\bm{\beta}}-r).
\)
Under the null hypothesis, $\mathcal{W}_{n}\stackrel{d}{\to}\chi_{q}^{2}$. The result is immediate from Theorem \ref{thm2}.

\begin{remark} The covariance estimator above uses the unsmoothed
indicator $1(\widehat{u}_{it}\leq0)$ in the estimated influence function.
Replacing it by the smoothed version $G_{h_{n}}(-\widehat{u}_{it})$
gives the same first-order limit under the bandwidth conditions above.
\end{remark}

\begin{remark} The truncation lag $\ell_n$ controls the estimation
of the long-run variance within each individual time series. In applications with small $T$,
 a short truncation lag such as $\ell_n=1$ may be used, as discussed in \citet{HahnKuersteiner11}. For
larger $T$, standard data-dependent choices, e.g.,  \citet{Andrews91}'s procedure, may be used as informal
guidance, but the optimal lag choice for smoothed panel IV quantile
GMM remains a separate problem. \end{remark}

\section{Simulation study}
\label{sec:sim}

This section investigates the finite-sample performance of the proposed
smoothed GMM estimator for panel quantile models with endogenous regressors
and individual fixed-effects. The design evaluates both point estimation
and feasible inference. 

\subsection{Data-generating process}

We generate $y_{it}=\alpha_i+x_{it}+u_{it}$ with $\alpha_i\sim U(-1,1)$ and $x_{it}=z_{1,it}+z_{2,it}+v_{it}$. The two independent instrument processes have $U(-1,1)$ marginals and are obtained by transforming stationary Gaussian AR(1) processes with autoregressive coefficient $0.50$. The outcome disturbance $u_{it}$ is a stationary Gaussian AR(1) process with coefficient $0.30$, shifted so that $\Q_\tau(u_{it})=0$. The bounded disturbance $v_{it}$ is a transformation of a Gaussian AR(1) process with coefficient $0.30$, sharing innovations with $u_{it}$ when $\sigma_{uv}\ne0$. The detailed construction and initialization appear in Supplementary Appendix B.

We examine $\sigma_{uv}\in\{0,0.50\}$, $\tau\in\{0.25,0.50\}$, and all nine combinations $(n,T)\in\{20,50,100\}^2$, with 1000 replications per design. The balanced and larger-$n$ designs also assess finite-sample performance outside the asymptotic regime $n/T\to0$.

We compare two-step SP-GMM with conventional fixed-effects QR (FE-QR), smoothed fixed-effects QR (FE-SQR), and cross-sectional IVQR based on unit-level time averages (CS-IVQR). The fixed-effects benchmarks do not correct for endogeneity; CS-IVQR uses the two averaged instruments. SP-GMM uses identity weighting initially and Bartlett HAC weighting in the second step. Final inference uses the unsmoothed residualized score from Section \ref{sec:inference}, with $\ell_n=\max\{1,\lfloor T^{1/5}\rfloor\}$.

Following \citet{GalvaoKato16,KaplanSun17,deCastroGalvaoKaplanLiu19}, we use the fourth-order kernel
$K(u)=\frac{105}{64}(1-5u^2+7u^4-3u^6)1\{|u|\le1\}$. The  bandwidth is
\(
h_n=c_\tau T^{-2/7}\), where the constant \(c_\tau=\left[\frac{4.5\,\phi\{\Phi^{-1}(\tau)\}^{4}}
{\{2\Phi^{-1}(\tau)^2+1\}^{2}}\right]^{1/5}\) is the normal-reference constant frequently used in quantile
regression bandwidth selection. Here, $\phi$ and $\Phi$ are the standard normal density and distribution functions. We report bias, root mean squared error (RMSE), and empirical coverage probability (CP) of nominal 95\% normal confidence intervals.

\subsection{Simulation results}

The results for $\tau=0.50$ are reported in Table B.1 of Supplementary
Appendix B. Below, we present the results for $\tau=0.25$, which show
slightly weaker overall finite-sample performance.

\begin{table}[t!]
\centering
\caption{Monte Carlo Results, \(\tau=0.25\)}\label{tab:simulation_tau025_new}
\singlespacing
\resizebox{\textwidth}{!}{%
\begin{tabular}{lllcccccccc}
\toprule
 &  &  & \multicolumn{2}{c}{SP-GMM}
 & \multicolumn{2}{c}{FE-QR}
 & \multicolumn{2}{c}{FE-SQR}
 & \multicolumn{2}{c}{CS-IVQR}
 \tabularnewline
\cline{4-5}\cline{6-7}\cline{8-9}\cline{10-11}
Statistic & \(T\) & \(n\)
& \multicolumn{2}{c}{\(\sigma_{uv}\)}
& \multicolumn{2}{c}{\(\sigma_{uv}\)}
& \multicolumn{2}{c}{\(\sigma_{uv}\)}
& \multicolumn{2}{c}{\(\sigma_{uv}\)}
\tabularnewline
 &  &  & \(0.00\) & \(0.50\)
 & \(0.00\) & \(0.50\)
 & \(0.00\) & \(0.50\)
 & \(0.00\) & \(0.50\)
\tabularnewline
\hline
\multirow{9}{*}{Bias}
& 20 & 20 & -0.004 & -0.011 & -0.001 & 0.564 & -0.001 & 0.564 & 0.002 & -0.029 \tabularnewline
&    & 50 & 0.002 & -0.003 & -0.001 & 0.563 & -0.001 & 0.563 & -0.022 & -0.019 \tabularnewline
&    & 100 & -0.005 & -0.001 & -0.003 & 0.565 & -0.003 & 0.565 & -0.002 & -0.007 \tabularnewline
& 50 & 20 & 0.003 & -0.002 & 0.001 & 0.552 & 0.001 & 0.552 & 0.007 & -0.017 \tabularnewline
&    & 50 & 0.000 & -0.001 & -0.002 & 0.553 & -0.002 & 0.553 & -0.014 & 0.038 \tabularnewline
&    & 100 & 0.000 & -0.001 & 0.000 & 0.552 & 0.000 & 0.552 & 0.010 & 0.016 \tabularnewline
& 100 & 20 & 0.001 & -0.003 & 0.001 & 0.548 & 0.002 & 0.547 & -0.004 & -0.129 \tabularnewline
&     & 50 & -0.001 & 0.000 & 0.000 & 0.548 & 0.000 & 0.548 & 0.014 & -0.050 \tabularnewline
&     & 100 & 0.000 & 0.000 & 0.000 & 0.547 & 0.000 & 0.547 & -0.021 & -0.047 \tabularnewline
\hline
\multirow{9}{*}{RMSE}
& 20 & 20 & 0.105 & 0.107 & 0.081 & 0.568 & 0.080 & 0.568 & 0.709 & 0.674 \tabularnewline
&    & 50 & 0.067 & 0.069 & 0.054 & 0.564 & 0.052 & 0.564 & 0.436 & 0.457 \tabularnewline
&    & 100 & 0.048 & 0.046 & 0.036 & 0.566 & 0.035 & 0.566 & 0.316 & 0.307 \tabularnewline
& 50 & 20 & 0.063 & 0.061 & 0.049 & 0.553 & 0.049 & 0.553 & 1.050 & 1.002 \tabularnewline
&    & 50 & 0.038 & 0.038 & 0.031 & 0.553 & 0.031 & 0.553 & 0.591 & 0.601 \tabularnewline
&    & 100 & 0.028 & 0.027 & 0.023 & 0.552 & 0.023 & 0.552 & 0.426 & 0.432 \tabularnewline
& 100 & 20 & 0.044 & 0.043 & 0.036 & 0.548 & 0.036 & 0.548 & 1.424 & 1.354 \tabularnewline
&     & 50 & 0.028 & 0.027 & 0.022 & 0.548 & 0.022 & 0.549 & 0.829 & 0.854 \tabularnewline
&     & 100 & 0.019 & 0.019 & 0.015 & 0.547 & 0.015 & 0.547 & 0.575 & 0.567 \tabularnewline
\hline
\multirow{9}{*}{CP}
& 20 & 20 & 0.896 & 0.908 & 0.910 & 0.008 & 0.919 & 0.010 & 0.753 & 0.769 \tabularnewline
&    & 50 & 0.952 & 0.956 & 0.937 & 0.000 & 0.956 & 0.002 & 0.799 & 0.798 \tabularnewline
&    & 100 & 0.964 & 0.973 & 0.968 & 0.000 & 0.970 & 0.001 & 0.843 & 0.842 \tabularnewline
& 50 & 20 & 0.936 & 0.950 & 0.949 & 0.000 & 0.957 & 0.000 & 0.766 & 0.776 \tabularnewline
&    & 50 & 0.960 & 0.966 & 0.952 & 0.000 & 0.964 & 0.000 & 0.846 & 0.824 \tabularnewline
&    & 100 & 0.972 & 0.970 & 0.962 & 0.000 & 0.969 & 0.000 & 0.854 & 0.838 \tabularnewline
& 100 & 20 & 0.938 & 0.941 & 0.942 & 0.000 & 0.946 & 0.000 & 0.789 & 0.796 \tabularnewline
&     & 50 & 0.965 & 0.963 & 0.965 & 0.000 & 0.971 & 0.000 & 0.837 & 0.814 \tabularnewline
&     & 100 & 0.970 & 0.975 & 0.972 & 0.000 & 0.975 & 0.000 & 0.857 & 0.865 \tabularnewline
\bottomrule
\end{tabular}}
\par\smallskip {\small\textit{Notes:} The table reports bias, RMSE, and empirical coverage probability (CP)
for nominal \(95\%\) confidence intervals. The results are based on
\(1000\) Monte Carlo replications.}
\end{table}

Table \ref{tab:simulation_tau025_new} reports the results for $\tau=0.25$. Under exogeneity, all three panel data estimators have
negligible bias, and their RMSE decreases with the panel dimensions.
FE-QR and FE-SQR generally have smaller RMSE than SP-GMM because they
avoid the efficiency cost of instrumenting an exogenous regressor.
Coverage for FE-QR ranges from \(0.910\) to \(0.972\), while coverage
for FE-SQR ranges from \(0.919\) to \(0.975\). Their undercoverage is
concentrated in the smallest designs and diminishes as the panel
dimensions increase. SP-GMM also exhibits undercoverage in the smallest
design, with coverage equal to \(0.896\) when \((T,n)=(20,20)\).
Its coverage ranges from \(0.896\) to \(0.972\), with some overcoverage
in larger designs.

Under endogeneity, SP-GMM remains accurately centered across all nine
designs. Its absolute bias is at most \(0.011\), and its RMSE decreases
from \(0.107\) when \((T,n)=(20,20)\) to \(0.019\) when
\((T,n)=(100,100)\). Coverage is \(0.908\) when \((n,T)=(20,20)\) and ranges from \(0.941\)
to \(0.975\) in the remaining designs, where it is generally close to the nominal
\(95\%\) level and comparable to that in the exogenous case.

In contrast, the biases of both FE-QR and FE-SQR remain between
approximately \(0.547\) and \(0.565\) under endogeneity. Their RMSE
remains between approximately \(0.547\) and \(0.568\), and their coverage
probabilities are essentially zero. For example, when
\((T,n)=(50,50)\), SP-GMM has bias \(-0.001\), RMSE \(0.038\), and
coverage \(0.966\). The corresponding values are approximately \(0.553\),
\(0.553\), and \(0.000\) for both FE-QR and FE-SQR. The close agreement
between the two fixed-effects estimators confirms that smoothing does not
correct the inconsistency caused by endogeneity.

The cross-sectional IVQR estimator is highly unstable. Although its bias is relatively small in several designs, its
RMSE ranges from \(0.307\) to \(1.354\) under endogeneity, and its
coverage probability ranges from \(0.769\) to \(0.865\). Thus, its
apparent robustness in terms of average bias does not translate into
precise estimation or accurately calibrated inference.

Overall, the simulations yield three conclusions. First, exploiting
the panel time-series structure is important for precision. Under
endogeneity, the cross-sectional IVQR estimator has smaller bias and
substantially better coverage than the conventional fixed-effects
estimators, but its RMSE remains large and unstable. This comparison
highlights the value of retaining the within-unit information discarded
by time aggregation. Second, when the regressor is exogenous, FE-QR and
FE-SQR are nearly unbiased, provide generally accurate inference, and
have slightly smaller RMSE than SP-GMM as instrumenting is
unnecessary. SP-GMM remains nearly unbiased and achieves 
coverage close to FE-QR and
FE-SQR, indicating a modest efficiency cost under exogeneity.
Third, when endogeneity is present, FE-QR and FE-SQR become concentrated
around the wrong population value and exhibit severe coverage
distortions. In contrast, SP-GMM remains nearly unbiased, its RMSE
decreases with both \(n\) and \(T\), and its coverage remains comparable
to that under exogeneity.

\section{Empirical application}
\label{sec:application}

We apply the method to monthly inflation in OECD economies to examine
whether persistence varies across structural quantiles. Inflation dynamics
depend on expectations and adjustment
\citep{phelps1967phillips,friedman1968role}, while subsequent work emphasizes
the importance and instability of backward-looking dynamics
\citep{stock2007has,coibion2015phillips}.
Our framework allows heterogeneous inflation propagation while controlling
for country fixed-effects and treating lagged regressors as potentially
endogenous.

Let \(CPI_{it}\), \(IP_{it}\), and \(UR_{it}\) denote the consumer price
index, industrial production index, and unemployment rate for country
\(i\) in month \(t\). Define monthly inflation as
\(\pi_{it}=100\{\log(CPI_{it})-\log(CPI_{i,t-1})\}\),
industrial production growth as
\(gIP_{it}=100\{\log(IP_{it})-\log(IP_{i,t-1})\}\), and the change in
unemployment as \(\Delta UR_{it}=UR_{it}-UR_{i,t-1}\).
The data come from the OECD Main Economic Indicators through FRED.
The balanced baseline sample contains 27 countries and 225 monthly
observations per country.\footnote{The covariance theory assumes
cross-country independence; general cross-country dependence would
require an extension.}

We estimate
\[
\Q_{\tau}\!\left[
\pi_{it}-\alpha_{i0\tau}
-\beta_{0\tau,1}\pi_{i,t-1}
-\beta_{0\tau,2}\Delta UR_{i,t-1}
-\beta_{0\tau,3}gIP_{i,t-1}
\mid \bm z_{it},\alpha_{i0\tau}
\right]=0.
\]
Under the structural IVQR interpretation in Section \ref{sec:model},
\(\beta_{0\tau,1}\) measures the percentage-point response of the
\(\tau\)-th structural quantile of current inflation to a one-percentage-point
increase in lagged inflation, controlling for changes in labor-market
slack and real activity.

Serial dependence in inflation disturbances creates a natural
endogeneity concern. For example, independent innovations that each enter the outcome disturbance in two consecutive months, as in an MA(1) process, can make lagged inflation endogenous, while lags dated \(t-2\) or earlier can satisfy the sequential quantile restriction.
We therefore instrument the lagged regressors using earlier lags.
Identification maintains
\(P(u_{it}\leq0\mid\mathcal F_{i,t-2},\alpha_{i0\tau})=\tau\),
where \(u_{it}\) is the structural residual and \(\mathcal F_{i,t-2}\)
contains information through month \(t-2\).
This restriction permits dependence on information arriving at \(t-1\)
but requires a zero conditional \(\tau\)-quantile given earlier information.
Subject to relevance and rank conditions, the baseline instruments are:
$\bm z_{it} =(\pi_{i,t-2},\pi_{i,t-3},
\Delta UR_{i,t-2},\Delta UR_{i,t-3},
gIP_{i,t-2},gIP_{i,t-3})'$.
We assess sensitivity to instrument choice by extending the lag range
to \(t-2\) through \(t-5\).

Table \ref{tab:empirical_inflation_simple} reports the estimates.
All specifications use the fourth-order kernel and bandwidth rule in
Section \ref{sec:sim}, two-step weighting from Section \ref{sec:est-GMM},
and Bartlett HAC covariance estimation from Section \ref{sec:inference}.
Standard errors are the square roots of 
\(\widehat\Sigma_\beta/(nT)\).

\begin{table}[t!]
\centering
\caption{Inflation Persistence Across Structural Quantiles}
\label{tab:empirical_inflation_simple}
\begin{tabular}{lccc}
\toprule
Instrument sets & \(\tau=0.25\) & \(\tau=0.50\) & \(\tau=0.75\) \\
\hline
\(t-2,t-3\)
& \(0.426\) & \(0.427\) & \(0.112\) \\
& \((0.092)\) & \((0.178)\) & \((0.092)\) \\[0.3em]
\(t-2,t-3,t-4\)
& \(0.829\) & \(0.537\) & \(0.121\) \\
& \((0.084)\) & \((0.087)\) & \((0.128)\) \\[0.3em]
\(t-2,t-3,t-4,t-5\)
& \(0.829\) & \(0.457\) & \(0.126\) \\
& \((0.087)\) & \((0.111)\) & \((0.087)\) \\
\bottomrule
\end{tabular}

\medskip
\begin{minipage}{0.94\textwidth}
\small
\textit{Notes:} Smoothed GMM estimates of \(\beta_{0\tau,1}\), with
standard errors assuming cross-country independence in parentheses.
All specifications use \(n=27\) countries; \(T=225\), \(224\), and
\(223\) for the three instrument sets, respectively.
\end{minipage}
\end{table}

The point estimates suggest asymmetric inflation propagation.
At \(\tau=0.25\), persistence ranges from \(0.426\) to \(0.829\)
and is statistically significant in every specification, although its
magnitude increases substantially when lag \(t-4\) is added.
Median estimates range from \(0.427\) to \(0.537\) and are also
significant. With the shortest instrument set, lower-quantile and
median estimates are similar; with additional lags, persistence is
larger at the lower quantile.

At \(\tau=0.75\), estimates are much smaller, ranging from \(0.112\)
to \(0.126\), and none differs significantly from zero at conventional
levels. Across instrument sets, the point estimates therefore
suggest stronger backward-looking persistence at the lower and median
structural quantiles than at the upper quantile. This pattern may reflect more gradual adjustment at the lower
structural quantile and a greater influence of new shocks, such as
energy-price movements, tax changes, or supply disruptions, at the
upper structural quantile.

\section{Concluding remarks}

\label{sec:conclusion}

This paper develops a smoothed GMM framework for panel quantile models with individual fixed-effects and endogenous regressors. The framework exploits unconditional moment restrictions implied by conditional quantile restrictions and nests instrumental variables panel quantile regression as a leading special case. Smoothing makes the sample moments differentiable, while local
uniqueness of the fixed-effect profiles yields a differentiable
profiled GMM criterion near the true parameter with probability
approaching one.

We establish  consistency, and asymptotic normality for
the proposed estimator under weak time dependence. The framework allows for endogenous regressors and individual fixed-effects, thereby relaxing both the iid sampling
assumption and the exogeneity restrictions commonly imposed in conventional panel QR models. We also discuss feasible inference based on a consistent
estimator of the asymptotic covariance matrix.

The Monte Carlo evidence supports the usefulness of the proposed method.
Across both exogenous and endogenous designs, SP-GMM is nearly unbiased
and becomes increasingly precise as the panel dimensions grow.
Coverage remains comparable across the designs and generally close
to the nominal level, with some finite-sample undercoverage and
overcoverage.

The empirical application to OECD inflation dynamics illustrates the
practical value of the method. Under the structural IVQR interpretation,
the point estimates suggest asymmetric inflation propagation: lagged
inflation has a substantially larger effect at the lower structural
quantile of current inflation than at the upper structural quantile.

Several issues remain for future research. First, the choice of
bandwidth is important for finite-sample performance and deserves further
study. Second, instrument selection and weighting-matrix choice may affect
practical implementation, especially in dynamic panel applications with
many available lags. Third, extensions to richer forms of unobserved
heterogeneity, such as interactive fixed-effects or common time effects,
would further broaden the applicability of the approach. These topics are
left for future work.

\newpage

\begin{center}
{\Large\bfseries Supplementary Appendix \\  GMM for Panel Data Quantile Regression }\par
\bigskip
{\large Antonio F. Galvao, Ulrich Hounyo, and Jiahao Lin}
\end{center}
\noindent This supplement contains the proofs for the results in the main paper as well as auxiliary results (Appendix A), and the detailed Monte Carlo design, implementation, and additional median-quantile numerical results (Appendix B).
\appendix
\numberwithin{equation}{section}
\numberwithin{table}{section}
\numberwithin{figure}{section}
\section{Proofs}\label{appendix}

\subsection{Proof of Lemma 4.1
{}}
\begin{proof}
Define the unsmoothed population fixed-effect moment
\begin{align}\label{eq: A_iu}
A_i^u(\alpha,\bm\beta)
=
\E\left[
1\left\{
y_{it}-\alpha-\bm x_{it}'\bm\beta\le0
\right\}
-\tau
\right].
\end{align}
By Assumption (A5),
\(
A_i^u\bigl(\alpha_i(\bm\beta),\bm\beta\bigr)=0
\)
for every \(i\ge1\) and \(\bm\beta\in\mathcal B\). In addition, the
collection
\(
\left\{
\alpha_i(\bm\beta):
i\ge1,\ \bm\beta\in\mathcal B
\right\}
\)
is contained in a compact subset of
\(\operatorname{int}(\mathcal A)\). Consequently, there exists a fixed
\(\varepsilon_0>0\) such that
\[
\left[
\alpha_i(\bm\beta)-\varepsilon_0,\,
\alpha_i(\bm\beta)+\varepsilon_0
\right]
\subseteq\mathcal A
\]
uniformly over \(i\ge1\) and \(\bm\beta\in\mathcal B\).

Fix any \(\varepsilon\in(0,\varepsilon_0)\). The map
\(\alpha\mapsto A_i^u(\alpha,\bm\beta)\) is nondecreasing because
\[
A_i^u(\alpha,\bm\beta)
=
P\left(
y_{it}-\bm x_{it}'\bm\beta\le\alpha
\right)-\tau.
\]
Therefore,
\[
A_i^u\bigl(\alpha_i(\bm\beta)-\varepsilon,\bm\beta\bigr)
\le0,
\qquad
A_i^u\bigl(\alpha_i(\bm\beta)+\varepsilon,\bm\beta\bigr)
\ge0.
\]
The uniform separation condition in Assumption (A5) then implies that
there exists \(c_\varepsilon>0\) such that, for all sufficiently large
\(n\),
\[
\begin{aligned}
A_i^u\bigl(\alpha_i(\bm\beta)-\varepsilon,\bm\beta\bigr)
\le -c_\varepsilon,\quad A_i^u\bigl(\alpha_i(\bm\beta)+\varepsilon,\bm\beta\bigr)
\ge c_\varepsilon,
\end{aligned}
\]
uniformly over \(i\le n\) and \(\bm\beta\in\mathcal B\).

We next establish uniform convergence of the smoothed sample moment to
the unsmoothed population moment. By the triangle inequality,
\[
\begin{aligned}
\max_{i\le n}
\sup_{\alpha\in\mathcal A,\bm\beta\in\mathcal B}
\left|
\widehat{A}_{i}(\alpha,\bm\beta)-A_i^u(\alpha,\bm\beta)
\right|&\quad\le
\max_{i\le n}
\sup_{\alpha\in\mathcal A,\bm\beta\in\mathcal B}
\left|
\widehat{A}_{i}(\alpha,\bm\beta)-\E[\widehat{A}_{i}(\alpha,\bm\beta)]
\right|\\
&\qquad+
\max_{i\le n}
\sup_{\alpha\in\mathcal A,\bm\beta\in\mathcal B}
\left|
\E [\widehat{A}_{i}(\alpha,\bm\beta)]-A_i^u(\alpha,\bm\beta)
\right|.
\end{aligned}
\]
Applying the blocking and maximal-inequality argument for the bounded
VC-type class indexed by
\((\alpha,\bm\beta)\in\mathcal A\times\mathcal B\), as detailed below for \eqref{eq:hat-E}
in the proof of Proposition 4.1
{}, gives
\[
\max_{i\le n}
\sup_{\alpha\in\mathcal A,\bm\beta\in\mathcal B}
\left|
\widehat{A}_{i}(\alpha,\bm\beta)-\E[\widehat{A}_{i}(\alpha,\bm\beta)]
\right|
=o_p(1).
\]
Moreover, following the similar argument for \eqref{eq:E-u}, the boundedness of the conditional density and the compact
support of \(K\) imply that
\[
\max_{i\le n}
\sup_{\alpha\in\mathcal A,\bm\beta\in\mathcal B}
\left|
\E [\widehat{A}_{i}(\alpha,\bm\beta)]-A_i^u(\alpha,\bm\beta)
\right|
=O(h_n)
=o(1).
\]
Hence,
\[
\max_{i\le n}
\sup_{\alpha\in\mathcal A,\bm\beta\in\mathcal B}
\left|
\widehat{A}_{i}(\alpha,\bm\beta)-A_i^u(\alpha,\bm\beta)
\right|
=o_p(1).
\]

It follows that, with probability approaching one, simultaneously for
every \(i\le n\) and \(\bm\beta\in\mathcal B\),
\[
\begin{aligned}
\widehat{A}_{i}\bigl(\alpha_i(\bm\beta)-\varepsilon,\bm\beta\bigr)
<-\frac{c_\varepsilon}{2},\quad
\widehat{A}_{i}\bigl(\alpha_i(\bm\beta)+\varepsilon,\bm\beta\bigr)
>\frac{c_\varepsilon}{2}.
\end{aligned}
\]
Because \(K\) is continuous, \(G\) is continuous, and therefore
\(\alpha\mapsto \widehat{A}_{i}(\alpha,\bm\beta)\) is continuous. The
intermediate value theorem now implies that, for every \(i\le n\) and
\(\bm\beta\in\mathcal B\), there exists at least one
\[
\widehat\alpha_i(\bm\beta)
\in
\left(
\alpha_i(\bm\beta)-\varepsilon,\,
\alpha_i(\bm\beta)+\varepsilon
\right)
\subseteq\mathcal A
\]
such that
\(
\widehat{A}_{i}\bigl(\widehat\alpha_i(\bm\beta),\bm\beta\bigr)=0.
\)

Finally, \(\widehat{A}_{i}(\alpha,\bm\beta)\) is jointly continuous in
\((\alpha,\bm\beta)\). Thus its zero correspondence has a closed graph
and nonempty compact values on the preceding event, and a standard
measurable-selection result yields a measurable choice
\(\widehat\alpha_i(\bm\beta)\). On the event where a root does not
exist, whose probability converges to zero, the profile may be defined
arbitrarily at a fixed point of \(\mathcal A\). This does not affect any
of the asymptotic results.
\end{proof}

\subsection{Proof of Proposition 4.1
{}}
\label{B1}
\begin{proof}
Let $\|\cdot\|_{1}$ denote the $\ell_{1}$ norm, and for $\delta>0$
define
$N_{\delta}:=\{\bm\beta\in\mathcal{B}:\|\bm\beta-\bm\beta_0\|_{1}<\delta\}$
and
$N_{\delta}^{c}:=\mathbb{R}^{d_\beta}\setminus N_{\delta}$.
By Lemma 4.1
{}, with probability approaching one
the sample root sets are nonempty for every \(\bm\beta\in\mathcal B\).
On this event, let \(\widehat{\bm\alpha}(\bm\beta)\) be the
criterion-minimizing profile defined above. It satisfies
\[
\frac{1}{T}\sum_{t=1}^{T}\left\{
G_{h_{n}}\!\left(-y_{it}+\widehat\alpha_i(\bm\beta)
+\bm{x}_{it}'\bm{\beta}\right)-\tau\right\}=0,
\qquad i=1,\ldots,n.
\]

We now show that, for every $\delta>0$, with probability approaching one,
\[
\inf_{\bm\beta\in N_\delta^c\cap\mathcal B}
\widehat m_n(\bm\beta)'\widehat W\widehat m_n(\bm\beta)>0.
\]
The profiled smoothed sample moments are
\begin{align*}
\widehat m_n(\bm\beta)=\frac1n\sum_{i=1}^n\widehat m_{ni}(\bm\beta),
&&\widehat m_{ni}(\bm\beta)=\frac1T\sum_{t=1}^T\bm z_{it}
\Big\{G_{h_n}\!\big(-y_{it}+\widehat\alpha_i(\bm\beta)
+\bm x_{it}'\bm\beta\big)-\tau\Big\}.
\end{align*}
Write
\(
\widehat \Q_n(\bm\beta)=\widehat m_n(\bm\beta)'\widehat W
\widehat m_n(\bm\beta)
\)
and
\(
\Q_n(\bm\beta)=m_n(\bm\beta)'Wm_n(\bm\beta),
\)
where
\begin{align*}
m_n(\bm\beta)=\frac1n\sum_{i=1}^nm_{ni}(\bm\beta),&&
m_{ni}(\bm\beta)=\E\!\left[\bm z_{it}
\Bigl(1\{y_{it}-\alpha_i(\bm\beta)-\bm x_{it}'\bm\beta\le0\}-\tau\Bigr)\right].
\end{align*}
It suffices to prove 
\begin{equation}
\sup_{\bm\beta\in\mathcal B}
|\widehat \Q_n(\bm\beta)-\Q_n(\bm\beta)|\overset p\longrightarrow0,
\label{eq:uniform_obj}
\end{equation}

By Assumption (A7), $\widehat{W}=W+o_{p}(1)$, which does not depend
on $\bm{\theta}$. We now show $\sup_{\bm{\beta}\in\mathcal{B}}\|\widehat{m}_{n}(\bm{\beta})-m_n(\bm{\beta})\|\overset{p}{\to}0.$ Let $H_{ni}(\alpha_i,\bm\beta)=\E[\widehat H_{ni}(\alpha_i,\bm\beta)]$. 
By the triangular inequality, 
for every
\(\bm\beta\in\mathcal B\),
\begin{align}
&\|\widehat m_n(\bm\beta)-m_n(\bm\beta)\|
\le
\frac1n\sum_{i=1}^n
\|\widehat H_{ni}(\widehat\alpha_i(\bm\beta),\bm\beta)
        -H_{ni}(\widehat\alpha_i(\bm\beta),\bm\beta)\| \nonumber \\
&+
\frac1n\sum_{i=1}^n
\|H_{ni}(\widehat\alpha_i(\bm\beta),\bm\beta)
        -H_{ni}(\alpha_i(\bm\beta),\bm\beta)\|+\frac1n\sum_{i=1}^n
\|H_{ni}(\alpha_i(\bm\beta),\bm\beta)-H^u_{ni}(\alpha_i(\bm\beta),\bm\beta)\|,\label{eq: hn-h}
\end{align}
where $H^u_{ni}(\alpha_i,\bm\beta)$ is as defined in 
(4), with $\tau$ omitted for simplicity.

For the first stochastic term on the right hand side of \eqref{eq: hn-h}, let $\mathbb H_{ni}(\alpha_i,\bm{\beta}):=\widehat{H}_{ni}(\alpha_i,\bm{\beta})-{H}_{ni}(\alpha_i,\bm{\beta})$. In the preceding stochastic bound,
the argument should be understood for a generic \(\alpha\in\mathcal A\)
and \(\bm\beta\in\mathcal B\), and then applied at
\(\alpha=\widehat\alpha_i(\bm\beta)\). Focus on the $\ell$th coordinate,
define 
\[
f(\xi_{it}):=z_{it,\ell}\Bigl(G_{h_{n}}(-y_{it}+{\alpha}+\bm{x}_{it}'\bm{\beta})-\tau\Bigr)-\E\!\left[z_{it,\ell}\Bigl(G_{h_{n}}(-y_{it}+{\alpha}+\bm{x}_{it}
'\bm{\beta})-\tau\Bigr)\right],
\]
where $\xi_{it}=(y_{it},\bm{x}_{it},z_{it,\ell})$, so that $\mathbb H _{ni,\ell}(\alpha,\bm{\beta})=\frac{1}{T}\sum_{t=1}^{T}f(\xi_{it}).$
Since $|G_{h_{n}}(\cdot)|\le C_G$ with $C_G=\int |K(v)|dv$, and $|z_{it,\ell}|\le C_{Z}$, we
have $|f(\xi_{it})|\le2C_Z(C_G+1)$. Moreover, by Lemma 1 of \citet{Yoshihara76},
$|\Cov(f(\xi_{1}),f(\xi_{1+k}))|\le C{\beta}_{i}(k)^{1/2}.$ Hence
\[
\sigma_{q}^{2}(f):=\Var\!\left(q^{-1/2}\sum_{t=1}^{q}f(\xi_{it})\right)\le C\Bigl(1+2\sum_{k=1}^{\infty}{\beta}_{i}(k)^{1/2}\Bigr)<\infty
\]
uniformly in $i$, $\alpha$, $\bm{\beta}$, and $\ell$. Since $\beta_{i}(q)\le Ba^{q}$
for some $a\in(0,1)$, one can deduce that $r\beta_{i}(q)\le C\sqrt{T}\,a^{\sqrt{T}}=o(n^{-2}),$
using $\log n/\sqrt{T}\to0$. Therefore, applying Propositions C.1 and C.2 of
\citet{GalvaoKato16} to the VC-type class indexed by
\((\alpha,\bm{\beta})\in\mathcal A\times\mathcal B\), with
\(q=\lfloor\sqrt{T}\rfloor\), \(r=\lfloor T/(2q)\rfloor\), and
\(s=2\log n\), yields, for each coordinate \(\ell\),
\[
\max_{i\le n}
P\left(
\sup_{\alpha\in\mathcal A,\bm\beta\in\mathcal B}
|\mathbb H_{ni,\ell}(\alpha,\bm{\beta})|
\ge
C\frac{\log n+\log T}{\sqrt{T}}
\right)
\le
2e^{-s}+2r\sup_{i\le n}\beta_i(q).
\]
Since \(e^{-s}=n^{-2}\), \(r\le C\sqrt T\), and
\(\sup_{i\le n}\beta_i(q)\le Ba^{\sqrt T}\) for some \(a\in(0,1)\),
the condition \((\log n)/\sqrt T\to0\) implies
\(
r\sup_{i\le n}\beta_i(q)
\le C\sqrt T a^{\sqrt T}
=o(n^{-2}).
\)
Hence, by the union bound over \(i\) and the finite number of
coordinates,
\[
P\left(
\max_{i\le n}
\sup_{\alpha\in\mathcal A,\bm\beta\in\mathcal B}
|\mathbb H_{ni}(\alpha,\bm{\beta})|
\ge
C\frac{\log n+\log T}{\sqrt{T}}
\right)
\le
Cn\{n^{-2}+o(n^{-2})\}
=o(1).
\]
Finally, \((\log n)/\sqrt T\to0\), while
\((\log T)/\sqrt T\to0\) automatically. Therefore, we have that 
\begin{align}\label{eq:hat-E}
\max_{i\le n}
\sup_{\alpha\in\mathcal A,\bm\beta\in\mathcal B}
|\widehat H_{ni}(\alpha,\bm\beta)
-H_{ni}(\alpha,\bm\beta)|
=o_p(1).
\end{align}

Consider the last two terms on the right hand side of \eqref{eq: hn-h}. We first record the uniform
consistency of the constrained fixed effects. By the same uniform
law of large numbers used above, together with the smoothing bias
argument below applied with \(z_{it,\ell}=1\),
\[
\begin{aligned}
\max_{i\le n}\sup_{\alpha\in\mathcal A,\bm\beta\in\mathcal B}
\Bigg|{}&\frac1T\sum_{t=1}^T
\{G_{h_n}(-y_{it}+\alpha+\bm x_{it}'\bm\beta)-\tau\}-\E[1\{y_{it}-\alpha-\bm x_{it}'\bm\beta\le0\}-\tau]
\Bigg|=o_p(1).
\end{aligned}
\]
Since \(\widehat\alpha_i(\bm\beta)\) solves the sample constrained equation,
the first term inside the absolute value evaluated at \(\alpha=\widehat\alpha_i(\bm\beta)\)
is zero. Assumption (A5) implies that, for every \(\varepsilon>0\),
the population constrained equation is uniformly separated from zero
whenever \(|\alpha-\alpha_i(\bm\beta)|\ge\varepsilon\). Hence, if
\(\max_{i\le n}\sup_{\bm\beta\in\mathcal B}
|\widehat\alpha_i(\bm\beta)-\alpha_i(\bm\beta)|\ge\varepsilon\),
then for some \(i\) and \(\bm\beta\),
\[
        \left|
        A_i^u(\widehat\alpha_i(\bm\beta),\bm\beta)
        \right|
        \ge c_{\varepsilon,\alpha}/2
\]
with probability approaching one, where $A_i^u(\cdot)$ is defined in \eqref{eq: A_iu}. The preceding uniform convergence
and the sample equation imply that the same quantity is \(o_p(1)\), a
contradiction. Therefore,
\begin{align}\label{eq: alphahat}
        \max_{i\le n}\sup_{\bm\beta\in\mathcal B}
        |\widehat\alpha_i(\bm\beta)-\alpha_i(\bm\beta)|
        =o_p(1).
\end{align}

Next consider the smoothing bias. For the \(\ell\)th coordinate,
uniformly in \(i\), \(\alpha\), and \(\bm\beta\),
write
\(
c_{it}(\alpha,\bm\beta)
=(\alpha-\alpha_{i0})+\bm x_{it}'(\bm\beta-\bm\beta_0).
\) We have
\[
        \E[\widehat H_{ni,\ell}(\alpha,\bm\beta)]
        -
        H_{ni,\ell}^u(\alpha,\bm\beta)
        =
        \E\left[
        z_{it,\ell}
        \{G_{h_n}(-y_{it}+\alpha+\bm x_{it}'\bm\beta)
        -
        1(y_{it}-\alpha-\bm x_{it}'\bm\beta\le0)\}
        \right].
\]
Conditioning on \((\bm x_{it},\bm z_{it})\), and writing
\(u=u_{it}-c_{it}(\alpha,\bm\beta)\), we obtain
{\small\[
        \E[\widehat H_{ni,\ell}(\alpha,\bm\beta)]
        -
        H_{ni,\ell}^u(\alpha,\bm\beta)
        =
        \E\left[
        z_{it,\ell}
        \int
        \{G(-u/h_n)-1(u\le0)\}
        f_i(u+c_{it}(\alpha,\bm\beta)\mid \bm q_{it})
        du
        \right].
\]}
Changing variables \(u=h_nv\) gives
{\small\[
        \E[\widehat H_{ni,\ell}(\alpha,\bm\beta)]
        -
        H_{ni,\ell}^u(\alpha,\bm\beta)
        =
        h_n \E\left[
        z_{it,\ell}
        \int
        \{G(-v)-1(v\le0)\}
        f_i(h_nv+c_{it}(\alpha,\bm\beta)\mid \bm q_{it})
        dv
        \right].
\]}
By the boundedness of \(z_{it,\ell}\), the uniform boundedness of the
conditional density, and the bounded support of the kernel,
\(\int |G(-v)-1(v\le0)|dv<\infty\). Therefore, 
\begin{align}\label{eq:E-u}
        \max_{i\le n}
        \sup_{\alpha\in\mathcal A,\bm\beta\in\mathcal B}
        \|\E[\widehat H_{ni}(\alpha,\bm\beta)]
        -
        H_{ni}^u(\alpha,\bm\beta)\|
        =O(h_n)=o(1).
\end{align}

We now apply the above results to the profiled moment. 
The last term on the right-hand side of \eqref{eq: hn-h} is bounded by
\(
        \max_{i\le n}
        \sup_{\alpha\in\mathcal A,\bm\beta\in\mathcal B}
        \|\E[\widehat H_{ni}(\alpha,\bm\beta)]-H_{ni}^u(\alpha,\bm\beta)\|
        =o_p(1).
\)
For the second term, the map
\(\alpha\mapsto H_{ni}(\alpha,\bm\beta)\) is uniformly continuous,
because
\[
        \|H_{ni}(\alpha,\bm\beta)-H_{ni}(\alpha_i(\bm\beta),\bm\beta)\|
        \le
        C|\alpha-\alpha_i(\bm\beta)|
\]
uniformly in \(i\) and \(\bm\beta\), by the boundedness of
\(\bm z_{it}\) and the uniform boundedness of the density. Hence,
using
\(\max_{i\le n}\sup_{\bm\beta\in\mathcal B}
|\widehat\alpha_i(\bm\beta)-\alpha_i(\bm\beta)|=o_p(1)\), we obtain
\[
        \sup_{\bm\beta\in\mathcal B}
        \|\widehat m_n(\bm\beta)-m_n(\bm\beta)\|
        =o_p(1).
\]

Therefore,
\[
\begin{aligned}
&\sup_{\bm\beta\in\mathcal B}
\left|
\widehat m_n(\bm\beta)'\widehat W\widehat m_n(\bm\beta)
-
m_n(\bm\beta)'Wm_n(\bm\beta)
\right|  \\
&\quad\le
\|\widehat W-W\|
        \sup_{\bm\beta\in\mathcal B}\|\widehat m_n(\bm\beta)\|^2+
\|W\|
        \sup_{\bm\beta\in\mathcal B}
        \|\widehat m_n(\bm\beta)-m_n(\bm\beta)\|\times
        \sup_{\bm\beta\in\mathcal B}
        \{\|\widehat m_n(\bm\beta)\|+\|m_n(\bm\beta)\|\}.
\end{aligned}
\]
The moments are uniformly bounded by Assumption (A2), \(\widehat W-W=o_p(1)\),
and the preceding display gives
\(\sup_{\bm\beta\in\mathcal B}
\|\widehat m_n(\bm\beta)-m_n(\bm\beta)\|=o_p(1)\). Hence,
\[
        \sup_{\bm\beta\in\mathcal B}
        \left|
        \widehat m_n(\bm\beta)'\widehat W\widehat m_n(\bm\beta)
        -
        m_n(\bm\beta)'Wm_n(\bm\beta)
        \right|
        =o_p(1).
\]
Assumption (A5) and positive definiteness of $W$ imply that, for every
$\delta>0$, the population criterion is bounded away from zero on
$N_\delta^c\cap\mathcal B$. In addition, $\Q_n(\bm\beta_0)=0$, so the
preceding uniform convergence gives
\(\widehat \Q_n(\bm\beta_0)=o_p(1)\), whereas
\(\inf_{\bm\beta\in N_\delta^c\cap\mathcal B}\widehat \Q_n(\bm\beta)\)
is bounded away from zero with probability approaching one. Since
$\widehat{\bm\beta}$ minimizes $\widehat \Q_n$, it follows that
\(\widehat{\bm\beta}\stackrel{p}{\to}\bm\beta_0\).
Finally, since \(\widehat\alpha_i=\widehat\alpha_i(\widehat{\bm\beta})\)
and \(\alpha_i(\bm\beta_0)=\alpha_{i0}\),
\[
\begin{aligned}
\max_{i\le n}|\widehat\alpha_i-\alpha_{i0}|
&\le
\max_{i\le n}
|\widehat\alpha_i(\widehat{\bm\beta})-\alpha_i(\widehat{\bm\beta})|
+
\max_{i\le n}
|\alpha_i(\widehat{\bm\beta})-\alpha_i(\bm\beta_0)|  \\
&\le
\max_{i\le n}\sup_{\bm\beta\in\mathcal B}
|\widehat\alpha_i(\bm\beta)-\alpha_i(\bm\beta)|
+
\max_{i\le n}
|\alpha_i(\widehat{\bm\beta})-\alpha_i(\bm\beta_0)|.
\end{aligned}
\]
The first term is \(o_p(1)\). To control the second, boundedness of
\(\bm x_{it}\) and the conditional density implies that changing
\(\bm\beta\) by $d$ changes the population fixed-effect moment by at most
$C\|d\|$, uniformly in $i$. The uniform separation in (A5) then implies
uniform continuity of $\alpha_i(\bm\beta)$ at $\bm\beta_0$. Therefore,
\(\widehat{\bm\beta}\stackrel{p}{\to}\bm\beta_0\) gives that the second
term is \(o_p(1)\). Hence
\(\max_{i\le n}|\widehat\alpha_i-\alpha_{i0}|=o_p(1)\). Combining this
with \(\widehat{\bm\beta}\stackrel{p}{\to}\bm\beta_0\) proves weak
consistency. 
\end{proof}

\subsection{Proof of Theorem 4.1
{}}

\begin{proof}
\label{app:proof-thm1}

Recall that $\bm{q}_{it}=(\bm{x}_{it}',\bm{z}_{it}')',\ u_{it}=y_{it}-\alpha_{i0}-\bm{x}_{it}'\bm{\beta}_{0}.$
For notational simplicity, write 
\[
G_{h_{n}}(v)=G(v/h_{n}),\qquad K_{h_{n}}(v)=h_{n}^{-1}K(v/h_{n}),
\]
where $G_{h_{n}}(v)$ is the smoothed version of $1(v\geq0)$. Hence
$G_{h_{n}}(-u_{it})$ is the smoothed version of $1(u_{it}\leq0)$.
Define 
\[
a_{it}(\alpha_{i},\bm{\beta})=G_{h_{n}}(\alpha_{i}+\bm{x}_{it}'\bm{\beta}-y_{it})-\tau,
\]
\[
\widehat{A}_{i}(\alpha_{i},\bm{\beta})=\frac{1}{T}\sum_{t=1}^{T}a_{it}(\alpha_{i},\bm{\beta}),\qquad \widehat{H}_{n}(\bm{\alpha},\bm{\beta})=\frac{1}{nT}\sum_{i=1}^{n}\sum_{t=1}^{T}\bm{z}_{it}a_{it}(\alpha_{i},\bm{\beta}).
\]
Let $A_i(\alpha_i,\bm\beta)=\E[\widehat{A}_{i}(\alpha_i,\bm\beta)]$.
The profiled fixed effect $\widehat{\alpha}_{i}(\bm{\beta})$ satisfies
$\widehat{A}_{i}(\widehat{\alpha}_{i}(\bm{\beta}),\bm{\beta})=0$, and the
profiled moment is $\widehat{m}_{n}(\bm{\beta})=\widehat{H}_{n}(\widehat{\bm{\alpha}}(\bm{\beta}),\bm{\beta})$.
We first collect several elementary expansions used in the proofs.

\begin{lemma} \label{lem:smoothing} Under the assumptions of Theorem
4.1
{}, for any uniformly bounded vector-valued function $\bm{b}_{it}$
measurable with respect to $\bm{q}_{it}$, let
\(
\bm\xi_{it}^{(h)}
=\bm b_{it}\{G_{h_n}(-u_{it})-1(u_{it}\leq0)\}.
\)
Then
\[
\max_{1\leq i\leq n}\|\E[\bm\xi_{it}^{(h)}]\|=O(h_n^r).
\]
Moreover, 
\[
\frac{1}{\sqrt{nT}}\sum_{i=1}^{n}\sum_{t=1}^{T}\bm\xi_{it}^{(h)}
-\sqrt{nT}\,\frac{1}{n}\sum_{i=1}^{n}\E[\bm\xi_{it}^{(h)}]
=o_p(1).
\]
Consequently, if $\sqrt{nT}h_n^r\to0$, then
\[
\frac{1}{\sqrt{nT}}\sum_{i=1}^{n}\sum_{t=1}^{T}\bm\xi_{it}^{(h)}=o_p(1).
\]
\end{lemma}
\begin{proof}
Condition on $\bm{q}_{it}$. By a change of variables and integration
by parts, the conditional mean of $G_{h_{n}}(-u_{it})-1(u_{it}\leq0)$
is 
\[
\int\{G_{h_{n}}(-u)-1(u\leq0)\}f_{i}(u\mid\bm{q}_{it})du=\int K(v)\left(F_{i}(-h_nv\mid\bm{q}_{it})-F_{i}(0\mid\bm{q}_{it})\right)dv.
\]
Because $K$ is an $r$-th order kernel and $f_{i}(u\mid\bm{q})$
has uniformly bounded derivatives up to order $r$, the usual kernel
expansion and Taylor expansion give this conditional mean as $O(h_{n}^{r})$,
uniformly in $i$ and $\bm{q}$. Multiplication by the bounded vector
$\bm{b}_{it}$ and taking expectations gives the first claim.

 For the centered term, note that \(\bm\xi_{it}^{(h)}\) is uniformly
bounded and differs from zero only when \(u_{it}\) lies within a fixed
multiple of \(h_n\) around zero, by the bounded support of \(K\).
Consequently,
\(
\E\|\bm\xi_{it}^{(h)}\|^2=O(h_n)
\)
uniformly in \(i\). To control the covariance sum, Assumption (A6) and
the boundedness of \(\bm b_{it}\) give, uniformly in \(i\) and \(j\geq1\),
\[
\begin{aligned}
\E\!\left[
\left\|\bm\xi_{i1}^{(h)}
\bm\xi_{i,1+j}^{(h)\prime}\right\|
\right]
&\leq
C P\!\left(
|u_{i1}|\leq Ch_n,\,
|u_{i,1+j}|\leq Ch_n
\right) \leq Ch_n^2.
\end{aligned}
\]
Similarly, \(\E\|\bm\xi_{it}^{(h)}\|\leq Ch_n\), and hence
\(
\left\|
\operatorname{Cov}\!\left(
\bm\xi_{i1}^{(h)},\bm\xi_{i,1+j}^{(h)}
\right)
\right\|
\leq Ch_n^2.
\)
The beta-mixing covariance inequality also gives
\(
\left\|
\operatorname{Cov}\!\left(
\bm\xi_{i1}^{(h)},\bm\xi_{i,1+j}^{(h)}
\right)
\right\|
\leq Ca^{j/2}.
\)
Splitting at \(J_n=\lceil C\log(1/h_n)\rceil\), for a sufficiently
large fixed \(C\), therefore yields
\[
\begin{aligned}
\sum_{j=1}^{\infty}
\left\|
\operatorname{Cov}\!\left(
\bm\xi_{i1}^{(h)},\bm\xi_{i,1+j}^{(h)}
\right)
\right\|
&\leq
CJ_nh_n^2+C\sum_{j>J_n}a^{j/2} =
O\!\left(h_n^2\log(1/h_n)\right)+O(h_n)
=O(h_n),
\end{aligned}
\]
uniformly in \(i\). Together with
\(\E\|\bm\xi_{it}^{(h)}\|^2=O(h_n)\), this implies
\(
\left\|
\operatorname{Var}\left(
\frac{1}{\sqrt T}\sum_{t=1}^T\bm\xi_{it}^{(h)}
\right)
\right\|
=O(h_n)
\)
uniformly over \(i\). Since the processes are independent
across $i$, 
\[
\Var\left(\frac{1}{\sqrt{nT}}
\sum_{i=1}^{n}\sum_{t=1}^{T}\bm\xi_{it}^{(h)}\right)=O(h_n).
\]
Because $h_{n}\to0$, the centered term is $o_{p}(1)$. The final
statement follows by combining the centered bound with the bias bound. 
\end{proof}
\begin{lemma}[Local empirical expansion] \label{lem:local_empirical}
Under the assumptions of Theorem 4.1
{}, for any deterministic
sequence $\eta_n\downarrow0$, define
\[
\begin{aligned}
\mathcal R_{ni}(d,b)
={}&\widehat{A}_{i}(\alpha_{i0}+d,\bm\beta_0+b)-\widehat{A}_{i}(\alpha_{i0},\bm\beta_0)-A_i(\alpha_{i0}+d,\bm\beta_0+b)+A_i(\alpha_{i0},\bm\beta_0).
\end{aligned}
\]
Then
\(
\max_{1\le i\le n}\sup_{|d|+\|b\|\le\eta_n}
\frac{|\mathcal R_{ni}(d,b)|}{|d|+\|b\|}=o_p(1),
\)
where the ratio is understood as zero when $d=0$ and $b=0$. \end{lemma}
\begin{proof}
Let $\theta_{i}=(\alpha_{i},\bm{\beta}')'$, $\theta_{i0}=(\alpha_{i0},\bm{\beta}_{0}')'$,
and write $v_{it}=(1,\bm{x}_{it}')'$. For fixed $i$, define $\dot{\widehat{A}}_{i}(\theta_{i})=\partial \widehat{A}_{i}(\theta_{i})/\partial\theta_{i}$
and $\dot{A}_{i}(\theta_{i})=\partial A_{i}(\theta_{i})/\partial\theta_{i}$.
Since the score is smoothed, $\widehat{A}_{i}$ is continuously differentiable
in $\theta_{i}$. For $h_{i}=(d,b')'$, the mean-value identity gives
\[
|\mathcal R_{ni}(d,b)|
\le (|d|+\|b\|)
\sup_{\|\theta_i-\theta_{i0}\|\le\eta_n}
\|\dot{\widehat{A}}_{i}(\theta_i)-\dot A_i(\theta_i)\|.
\]
Hence it is enough to show $\max_{1\le i\le n}\sup_{\|\theta_{i}-\theta_{i0}\|\le\eta_{n}}\|\dot{\widehat{A}}_{i}(\theta_{i})-\dot{A}_{i}(\theta_{i})\|=o_{p}(1).$

We prove this componentwise. Each component of $\dot{\widehat{A}}_{i}(\theta_{i})$
is a kernel average of the form $\frac{1}{T}\sum_{t=1}^{T}\psi_{it}(\theta_{i}),\ \psi_{it}(\theta_{i})=K_{h_{n}}(y_{it}-\alpha_{i}-\bm{x}_{it}'\bm{\beta})c_{it},$
up to terms of the same order, where $K_{h_{n}}(u)=h_{n}^{-1}K(u/h_{n})$
and $c_{it}$ is one component of $v_{it}$ or a bounded product involving
$v_{it}$. Assumption (A6) replaces the bounded joint conditional density
condition in \citet{GalvaoKato16} with a two-time local concentration
bound. Together with compact kernel support and bounded multipliers,
this condition supplies the covariance bounds needed for the local
kernel averages. We can therefore follow the uniform kernel argument in their Lemma B.2 with \(j=0\) to obtain
\[
\max_{1\le i\le n}\sup_{\|\theta_{i}-\theta_{i0}\|\le\eta_{n}}\left\Vert \dot{\widehat{A}}_{i}(\theta_{i})-\dot{A}_{i}(\theta_{i})\right\Vert =O_{p}\left\{ \left(\frac{\log(nT)}{Th_{n}}\right)^{1/2}\right\}.
\]
Since $n/T\to0$, we have $n\le T$ eventually. Hence, with $h_{n}\asymp T^{-c}$,
$c<1/3<1/2$, and $\frac{\log T}{\sqrt{T}}\le1$, one can deduce that
$\frac{\log(nT)}{Th_{n}}\le C\frac{\log T}{\sqrt{T}}\frac{1}{T^{1/2-c}}\to0.$
Therefore the above bound is $o_{p}(1)$.

\end{proof}
\begin{lemma} \label{lem:alpha_max} Under the assumptions of Theorem
4.1
{}, uniformly for $\|\bm{\beta}-\bm{\beta}_{0}\|=o_p(1)$, 
\[
\max_{1\le i\le n}|\widehat{\alpha}_{i}(\bm{\beta})-\alpha_{i0}|=O_{p}\left(\sqrt{\frac{\log(nT)}{T}}\right)+O_{p}(\|\bm{\beta}-\bm{\beta}_{0}\|)+O(h_{n}^{r}).
\]
\end{lemma}
\begin{proof}
Let $\Delta_{\beta}=\bm{\beta}-\bm{\beta}_{0}$ and $\widehat{\Delta}_{i}=\widehat{\alpha}_{i}(\bm{\beta})-\alpha_{i0}$.
Since $\widehat{A}_{i}(\widehat{\alpha}_{i}(\bm{\beta}),\bm{\beta})=0$, we
have 
\begin{align}
0={}&\widehat{A}_{i}(\alpha_{i0},\bm\beta_0)
+A_i(\alpha_{i0}+\widehat\Delta_i,\bm\beta_0+\Delta_\beta)
-A_i(\alpha_{i0},\bm\beta_0)+R_{Ai}.
\label{eq: lemma A2-2}
\end{align}
where 
\begin{align*}
R_{Ai}={}&\widehat{A}_{i}(\alpha_{i0}+\widehat\Delta_i,\bm\beta_0+\Delta_\beta)
-\widehat{A}_{i}(\alpha_{i0},\bm\beta_0)-A_i(\alpha_{i0}+\widehat\Delta_i,\bm\beta_0+\Delta_\beta)
+A_i(\alpha_{i0},\bm\beta_0).
\end{align*}
Consider the second term on the right hand side of \eqref{eq: lemma A2-2},
by Lemma \ref{lem:smoothing} and the conditional density expansion,
uniformly in $i$, 
\begin{align}
A_i(\alpha_{i0}+d,\bm\beta_0+\Delta_\beta)-A_i(\alpha_{i0},\bm\beta_0)
={}&f_i(0)d+\E[f_i(0\mid\bm q_{it})\bm x_{it}']\Delta_\beta\nonumber\\
&+O\{(|d|+\|\Delta_\beta\|)^2\}+O(h_n^r).
\label{eq: lemmaA2 A_i}
\end{align}
Now for the first term on the right hand side of \eqref{eq: lemma A2-2}.
Also, Lemma \ref{lem:smoothing} gives
\[
\E[\widehat{A}_{i}(\alpha_{i0},\bm\beta_0)]=O(h_n^r).
\]
Since the summand is bounded and the time series is exponentially
mixing, a Bernstein inequality for beta-mixing sequences and a union
bound over $i=1,\ldots,n$ give 
\[
\max_{1\le i\le n}|\widehat{A}_{i}(\alpha_{i0},\bm{\beta}_{0})-\E[\widehat{A}_{i}(\alpha_{i0},\bm{\beta}_{0})]|=O_{p}\left(\sqrt{\frac{\log(nT)}{T}}\right).
\]
Hence, one can deduce that $\max_{1\le i\le n}|\widehat{A}_{i}(\alpha_{i0},\bm{\beta}_{0})|=O_{p}\left(\sqrt{\frac{\log(nT)}{T}}\right)+O(h_{n}^{r}).$

It remains to control the local empirical remainder $R_{Ai}$ on the
right hand side of \eqref{eq: lemma A2-2}. Applying Lemma \ref{lem:local_empirical}, established above, with $\max_i|\widehat{\Delta}_i|=o_p(1)$  
yields that $\max_{1\le i\le n}|R_{Ai}|=o_{p}\left(\max_{1\le i\le n}|\widehat{\Delta}_{i}|+\|\Delta_{\beta}\|\right).$

Combining the results above at $d=\widehat{\Delta}_i$ with the fact that $\inf_{i}f_{i}(0)>0$,
we obtain 
\begin{align}
\max_i|\widehat\Delta_i|
\le{}&C\max_i|\widehat{A}_{i}(\alpha_{i0},\bm\beta_0)|+C\|\Delta_\beta\|
+C\max_i|\widehat\Delta_i|^2\nonumber \\
&+C\|\Delta_\beta\|^2+Ch_n^r
+o_p\!\left(\max_i|\widehat\Delta_i|+\|\Delta_\beta\|\right).\label{eq: deltai}
\end{align}
We first verify that $\max_i|\widehat{\Delta}_i|=o_p(1)$. By the uniform profiling consistency established in \eqref{eq: alphahat}, we have
\(
\max_{i\le n}\sup_{\bm \beta\in\mathcal B}
|\widehat\alpha_i(\bm \beta)-\alpha_i(\bm \beta)|=o_p(1).
\)
It remains to control the population profile. Let
$d_i=\alpha_i(\bm\beta)-\alpha_{i0}$. Since
$\alpha_i(\bm\beta)$ solves the population fixed-effect equation and
$\alpha_i(\bm\beta_0)=\alpha_{i0}$, we have
\[
\E\left[
F_i(d_i+\bm x_{it}'\Delta_\beta\mid\bm q_{it})
-
F_i(0\mid\bm q_{it})
\right]=\tau-\tau=0,
\]
where $F_i(\cdot\mid\bm q)$ is the conditional distribution function of
$u_{it}$ given $\bm q_{it}$. Moreover, by Assumptions (A2), (A5), (A6), one can deduce that $\max_i|d_i|\to0$. Hence, uniformly
over $i$,
\[
\E\left[
F_i(d_i+\bm x_{it}'\Delta_\beta\mid\bm q_{it})
-
F_i(0\mid\bm q_{it})
\right]
=
f_i(0)d_i
+
\E[f_i(0\mid\bm q_{it})\bm x_{it}']\Delta_\beta
+
o(|d_i|+\|\Delta_\beta\|).
\]
Furthermore, the same expansion and the uniform lower bound on $f_i(0)$ imply
that 
\(
|d_i|
\le
C\|\Delta_\beta\|
+
o(|d_i|+\|\Delta_\beta\|),
\)
uniformly over $i$. One can deduce that $\max_i|\alpha_i(\bm\beta)-\alpha_{i0}|
\le C\|\Delta_\beta\|$, and therefore given $\|\bm{\beta}-\bm{\beta}_{0}\|=o_p(1)$,
\[
\max_i|\widehat{\Delta}_i|
\le
\max_{i\le n}\sup_{\bm \beta\in\mathcal B}
|\widehat\alpha_i(\bm \beta)-\alpha_i(\bm \beta)|
+
\max_i|\alpha_i(\bm\beta)-\alpha_{i0}|
=
o_p(1)+O(\|\Delta_\beta\|)
=o_p(1).
\]
Since $\max_i|\widehat{\Delta}_i|=o_p(1)$ and $\|\Delta_\beta\|=o_p(1)$,
the quadratic terms and the $o_p(\cdot)$ term in \eqref{eq: deltai} can be
absorbed into the left-hand side. Therefore, we have
\[
\max_{i}|\widehat{\Delta}_{i}|=O_{p}\left(\sqrt{\frac{\log(nT)}{T}}\right)+O_{p}(\|\Delta_{\beta}\|)+O(h_{n}^{r}).
\]
\end{proof}

\begin{lemma}[Local uniqueness and differentiability of the profile]
\label{lem:profile_unique}
Suppose that the moment restrictions
(2)
{}-(3)
{} and Assumptions
(A1)-(A10) hold, $n/T\to0$, and let \(\eta_n\to0\). Then, with
probability approaching one, uniformly over \(i\le n\) and
\(\|\bm\beta-\bm\beta_0\|\le\eta_n\), the sample fixed-effect equation
\(
\widehat{A}_{i}(\alpha,\bm\beta)=0
\)
has a unique solution in an \(o(1)\) neighborhood of \(\alpha_{i0}\).
Every solution in \(\mathcal A\) belongs to this neighborhood. Moreover,
the resulting profile \(\widehat\alpha_i(\bm\beta)\) is continuously
differentiable. Write
\(
\widehat v_{it}(\bm\beta)
=\widehat\alpha_i(\bm\beta)+\bm x_{it}'\bm\beta-y_{it}.
\)
Then
\[
\frac{\partial\widehat\alpha_i(\bm\beta)}
{\partial\bm\beta'}
=
-
\left\{
\frac1T\sum_{t=1}^T
K_{h_n}(\widehat v_{it}(\bm\beta))
\right\}^{-1}
\left\{
\frac1T\sum_{t=1}^T
K_{h_n}(\widehat v_{it}(\bm\beta))\bm x_{it}'
\right\}.
\]
In addition,
\[
\inf_{i\le n}
\inf_{\|\bm\beta-\bm\beta_0\|\le\eta_n}
\frac1T\sum_{t=1}^T
K_{h_n}
\bigl(
\widehat\alpha_i(\bm\beta)
+\bm x_{it}'\bm\beta-y_{it}
\bigr)
\ge \frac{\underline f}{2},
\]
with probability approaching one, where
\(\underline f=\inf_i f_i(0)>0\).
\end{lemma}

\begin{proof}[Proof of Lemma \ref{lem:profile_unique}]
For \(d=\alpha-\alpha_{i0}\) and
\(\Delta_\beta=\bm\beta-\bm\beta_0\),
we have \(
\frac{\partial \widehat{A}_{i}(\alpha,\bm\beta)}{\partial\alpha}
=
\frac1T\sum_{t=1}^T
K_{h_n}
\bigl(d+\bm x_{it}'\Delta_\beta-u_{it}\bigr).
\)
The population counterpart satisfies, uniformly in \(i\),
\[
\begin{aligned}
\E\left[
K_{h_n}
\bigl(d+\bm x_{it}'\Delta_\beta-u_{it}\bigr)
\right]
&=
\E\left[
\int K(v)
f_i
\bigl(
d+\bm x_{it}'\Delta_\beta-h_nv
\mid\bm q_{it}
\bigr)\,dv
\right] \\
&=
\E\left[
f_i
\bigl(
d+\bm x_{it}'\Delta_\beta
\mid\bm q_{it}
\bigr)
\right]
+O(h_n^r) \\
&=
f_i(0)
+
O\bigl(
|d|+\|\Delta_\beta\|+h_n^r
\bigr).
\end{aligned}
\]
Since \(\inf_i f_i(0)=\underline f>0\), there exists a neighborhood of
\((\alpha_{i0},\bm\beta_0)\) on which the population derivative is
bounded below by \(3\underline f/4\), uniformly in \(i\). The uniform
kernel law of large numbers gives
\[
\sup_{i\le n}
\sup_{\substack{|d|\le\delta_n\\
\|\Delta_\beta\|\le\eta_n}}
\left|
\frac{\partial \widehat{A}_{i}
(\alpha_{i0}+d,\bm\beta_0+\Delta_\beta)}
{\partial\alpha}
-
\E\left[
\frac{\partial \widehat{A}_{i}
(\alpha_{i0}+d,\bm\beta_0+\Delta_\beta)}
{\partial\alpha}
\right]
\right|
=o_p(1),
\]
for any sufficiently slowly vanishing sequence \(\delta_n\to0\).
Consequently,
\[
\inf_{i\le n}
\inf_{\substack{|d|\le\delta_n\\
\|\Delta_\beta\|\le\eta_n}}
\frac{\partial \widehat{A}_{i}
(\alpha_{i0}+d,\bm\beta_0+\Delta_\beta)}
{\partial\alpha}
\ge
\frac{\underline f}{2},
\]
with probability approaching one.

The proof of Lemma \ref{lem:alpha_max} applies to every solution of the
sample profiling equation and implies that, uniformly over
\(\|\bm\beta-\bm\beta_0\|\le\eta_n\), every such solution satisfies
\[
\max_{i\le n}
|\widehat\alpha_i(\bm\beta)-\alpha_{i0}|
=
O_p\left(\sqrt{\frac{\log(nT)}{T}}\right)
+
O(\eta_n)
+
O(h_n^r)
=o_p(1).
\]
Thus every solution lies, with probability approaching one, in the
neighborhood on which the derivative is strictly positive. The map
\(\alpha\mapsto \widehat{A}_{i}(\alpha,\bm\beta)\) is therefore strictly
increasing on a set containing every possible solution, which implies
uniqueness.

Finally, \(\widehat{A}_{i}(\alpha,\bm\beta)\) is continuously differentiable
because \(K\) is continuously differentiable, and its derivative with
respect to \(\alpha\) is bounded away from zero at the unique solution.
The implicit function theorem therefore yields continuous
differentiability of \(\widehat\alpha_i(\bm\beta)\) and the stated
derivative formula.
\end{proof}

\begin{lemma} \label{lem:alpha_profile} Under the assumptions of Theorem
4.1
{}, the following bounds hold uniformly for
\(\|\bm{\beta}-\bm{\beta}_{0}\|=O_{p}((nT)^{-1/2})\),
\begin{align*}
 & \frac{1}{n}\sum_{i=1}^{n}|\widehat{\alpha}_{i}(\bm{\beta})-\alpha_{i0}|
 =O_{p}(T^{-1/2})+O_{p}(\|\bm{\beta}-\bm{\beta}_{0}\|),\\
 & \frac{1}{n}\sum_{i=1}^{n}|\widehat{\alpha}_{i}(\bm{\beta})-\alpha_{i0}|^{2}
 =O_{p}(T^{-1})+O_{p}(\|\bm{\beta}-\bm{\beta}_{0}\|^{2}).
\end{align*}
\end{lemma}
\begin{proof}
Let $\widehat{\Delta}_{i}=\widehat{\alpha}_{i}(\bm{\beta})-\alpha_{i0}$
and $\Delta_{\beta}=\bm{\beta}-\bm{\beta}_{0}$.
Recall that $\widehat{A}_{i}(\widehat{\alpha}_{i}(\bm{\beta}),\bm{\beta})=0$,
\begin{align*}
0={}&\widehat{A}_{i}(\alpha_{i0},\bm\beta_0)
+A_i(\alpha_{i0}+\widehat\Delta_i,\bm\beta_0+\Delta_\beta)-A_i(\alpha_{i0},\bm\beta_0)+R_{Ai},
\end{align*}
where 
\begin{equation}
R_{Ai}=\widehat{A}_{i}(\alpha_{i0}+\widehat\Delta_i,\bm\beta_0+\Delta_\beta)
-\widehat{A}_{i}(\alpha_{i0},\bm\beta_0)-A_i(\alpha_{i0}+\widehat\Delta_i,\bm\beta_0+\Delta_\beta)
+A_i(\alpha_{i0},\bm\beta_0).
\label{eq: RAAi}
\end{equation}

We first expand the population term. Since $\E[F_{i}(0\mid\bm{q}_{it})]=\tau$,
a Taylor expansion gives, uniformly in $i$, 
\begin{align*}
\left.\E\{1(u_{it}\le d+\bm x_{it}'b)-\tau\}\right|_{(d,b)=(\widehat\Delta_i,\Delta_\beta)}
={}&f_i(0)\widehat\Delta_i
+\E[f_i(0\mid\bm q_{it})\bm x_{it}']\Delta_\beta+O\{(|\widehat\Delta_i|+\|\Delta_\beta\|)^2\}.
\end{align*}
By Lemma \ref{lem:smoothing}, the difference between the smoothed
and unsmoothed population moments is $O(h_{n}^{r})$, uniformly in
$i$. Hence 
\begin{equation}
\label{eq: Rpi}
A_i(\alpha_{i0}+\widehat\Delta_i,\bm\beta_0+\Delta_\beta)
-A_i(\alpha_{i0},\bm\beta_0)
=f_i(0)\widehat\Delta_i
+\E[f_i(0\mid\bm q_{it})\bm x_{it}']\Delta_\beta+R_{Pi}.
\end{equation}
where $|R_{Pi}|\le C\{|\widehat{\Delta}_{i}|+\|\Delta_{\beta}\|\}^{2}+Ch_{n}^{r}.$

We now bound $R_{Ai}$. By Lemma \ref{lem:alpha_max}, $\max_{i}|\widehat{\Delta}_{i}|=o_{p}(1)$.
On this event, the local empirical-process bound in Lemma \ref{lem:local_empirical},
applied now after averaging over $i$, yields
\[
\frac1n\sum_{i=1}^n|R_{Ai}|
\le o_p(\|\Delta_\beta\|)
+o_p\!\left(\frac1n\sum_{i=1}^n|\widehat\Delta_i|\right).
\]

Combining the population expansion with the score equation gives 
\[
\widehat{\Delta}_{i}=-s_{i}\widehat{A}_{i}(\alpha_{i0},\bm{\beta}_{0})-s_{i}\E[f_{i}(0\mid\bm{q}_{it})\bm{x}_{it}']\Delta_{\beta}-s_{i}R_{Pi}-s_{i}R_{Ai},
\]
where $s_{i}=\{f_{i}(0)\}^{-1}$. Since $s_{i}$ is uniformly bounded, 
\[
|\widehat{\Delta}_{i}|\le C|\widehat{A}_{i}(\alpha_{i0},\bm{\beta}_{0})|+C\|\Delta_{\beta}\|+C\widehat{\Delta}_{i}^{2}+C\|\Delta_{\beta}\|^{2}+Ch_{n}^{r}+C|R_{Ai}|.
\]
Averaging over $i$, using the preceding bound for $R_{Ai}$, and
absorbing $o_{p}(n^{-1}\sum_{i}|\widehat{\Delta}_{i}|)$ into the
left-hand side, gives 
\[
\frac{1}{n}\sum_{i=1}^{n}|\widehat{\Delta}_{i}|=O_{p}\left(\frac{1}{n}\sum_{i=1}^{n}|\widehat{A}_{i}(\alpha_{i0},\bm{\beta}_{0})|\right)+O_{p}(\|\Delta_{\beta}\|)+O_{p}\left(\frac{1}{n}\sum_{i=1}^{n}\widehat{\Delta}_{i}^{2}\right)+O(h_{n}^{r}).
\]
Next, by Lemma \ref{lem:smoothing}, $\E[\widehat{A}_{i}(\alpha_{i0},\bm{\beta}_{0})]=O(h_{n}^{r})$.
Boundedness, stationarity, and exponential beta-mixing imply $\E[\widehat{A}_{i}(\alpha_{i0},\bm{\beta}_{0})^{2}]=O(T^{-1})+O(h_{n}^{2r})$,
uniformly in $i$. Therefore 
\[
\frac{1}{n}\sum_{i=1}^{n}\widehat{A}_{i}(\alpha_{i0},\bm{\beta}_{0})^{2}=O_{p}(T^{-1})+O(h_{n}^{2r})=O_{p}(T^{-1}),
\]
because $Th_{n}^{r}\to0$. Moreover, $n^{-1}\sum_{i}|\widehat{A}_{i}(\alpha_{i0},\bm{\beta}_{0})|=O_{p}(T^{-1/2})$.

To obtain the second-moment bound, square the preceding expansion
for $\widehat{\Delta}_{i}$. Since $(a+b+c+d)^{2}\le C(a^{2}+b^{2}+c^{2}+d^{2})$,
\[
\frac{1}{n}\sum_{i=1}^{n}\widehat{\Delta}_{i}^{2}\le C\frac{1}{n}\sum_{i=1}^{n}\widehat{A}_{i}(\alpha_{i0},\bm{\beta}_{0})^{2}+C\|\Delta_{\beta}\|^{2}+C\frac{1}{n}\sum_{i=1}^{n}R_{Pi}^{2}+C\frac{1}{n}\sum_{i=1}^{n}R_{Ai}^{2}.
\]
Because $\max_{i}|\widehat{\Delta}_{i}|=o_{p}(1)$, 
\[
\frac{1}{n}\sum_{i=1}^{n}R_{Pi}^{2}\le\left(\max_{i}|R_{Pi}|\right)\frac{1}{n}\sum_{i=1}^{n}|R_{Pi}|=o_{p}\left(\frac{1}{n}\sum_{i}\widehat{\Delta}_{i}^{2}\right)+o_{p}(\|\Delta_{\beta}\|^{2})+o_{p}(T^{-1}).
\]
Similarly, by Lemma \ref{lem:local_empirical} and $\max_{i}|\widehat{\Delta}_{i}|+\|\Delta_{\beta}\|=o_{p}(1)$,
we have
\[
\frac1n\sum_{i=1}^nR_{Ai}^2
=o_p\!\left(\frac1n\sum_i\widehat\Delta_i^2\right)
+o_p(\|\Delta_\beta\|^2).
\]
Absorbing the two $o_{p}(n^{-1}\sum_{i}\widehat{\Delta}_{i}^{2})$
terms into the left-hand side yields 
\[
\frac{1}{n}\sum_{i=1}^{n}\widehat{\Delta}_{i}^{2}=O_{p}(T^{-1})+O_{p}(\|\Delta_{\beta}\|^{2})+o_{p}(T^{-1}).
\]
Substituting this second-moment bound into the previous first-moment
bound gives 
\[
\frac{1}{n}\sum_{i=1}^{n}|\widehat{\Delta}_{i}|=O_{p}(T^{-1/2})+O_{p}(\|\Delta_{\beta}\|).
\]
\end{proof}

\begin{lemma} \label{lem:profile_moment} Under the assumptions of Theorem
4.1
{}, 
\begin{align*}
\sqrt{nT}\,\widehat m_n(\bm\beta_0)
={}&\frac1{\sqrt{nT}}\sum_{i=1}^n\sum_{t=1}^T
(\bm z_{it}-\bm\delta_i)\{1(u_{it}\le0)-\tau\}+o_p(1),
\end{align*}
where $\bm{\delta}_{i}=s_{i}\E[f_{i}(0\mid\bm{q}_{it})\bm{z}_{it}].$
\end{lemma}
\begin{proof}

Recall that
\[
\widehat m_n(\bm\beta_0)=\frac1{nT}\sum_{i=1}^n\sum_{t=1}^T
\bm z_{it}a_{it}(\widehat\alpha_i(\bm\beta_0),\bm\beta_0),
\]
where
\(
a_{it}(\alpha_i,\bm\beta)
=G_{h_n}(\alpha_i+\bm x_{it}'\bm\beta-y_{it})-\tau.
\)
Let \(\widehat{\Delta}_{i}=\widehat{\alpha}_{i}(\bm{\beta}_{0})-\alpha_{i0}\).
Adding and subtracting the moment evaluated at the true fixed effects,
\begin{align*}
\widehat m_n(\bm\beta_0)
=\widehat{H}_{n}(\bm\alpha_0,\bm\beta_0)+\frac1n\sum_{i=1}^n\frac1T\sum_{t=1}^T\bm z_{it}
\{a_{it}(\alpha_{i0}+\widehat\Delta_i,\bm\beta_0)
-a_{it}(\alpha_{i0},\bm\beta_0)\}.
\end{align*}

\textbf{Rewrite $\widehat m_n(\bm\beta_0)$.}
We first expand the population component of the second term. Conditioning
on \(\bm q_{it}\) and then taking expectations gives
\[
\left.
\E[\bm z_{it}\{a_{it}(\alpha_{i0}+\Delta,\bm\beta_0)
-a_{it}(\alpha_{i0},\bm\beta_0)\}]
\right|_{\Delta=\widehat\Delta_i}
=
\E[f_i(0\mid\bm q_{it})\bm z_{it}]\widehat\Delta_i
+O(\widehat\Delta_i^2)
+O(h_n^r|\widehat\Delta_i|),
\]
uniformly in \(i\). The term \(O(h_n^r|\widehat\Delta_i|)\) is the
change in the smoothing bias when the threshold is shifted by
\(\widehat\Delta_i\). This population expansion follows from the
smoothness of the conditional density and the \(r\)-th order property
of the kernel.

We next control the centered sample fluctuation around the preceding
expectation,

\begin{align}
\begin{split}
\frac1n\sum_{i=1}^n\Bigg[&
\frac1T\sum_{t=1}^T\bm z_{it}
\{a_{it}(\alpha_{i0}+\widehat\Delta_i,\bm\beta_0)
-a_{it}(\alpha_{i0},\bm\beta_0)\}
\\[-0.5ex]
&{}-
\left.
\E[\bm z_{it}\{a_{it}(\alpha_{i0}+\Delta,\bm\beta_0)
-a_{it}(\alpha_{i0},\bm\beta_0)\}]
\right|_{\Delta=\widehat\Delta_i}
\Bigg].
\end{split}\label{eq:sample flucutation}
\end{align}
By the integral form of Taylor's theorem around \(\alpha_{i0}\),
\[
a_{it}(\alpha_{i0}+\widehat\Delta_i,\bm\beta_0)
-a_{it}(\alpha_{i0},\bm\beta_0)
=
\dot a_{it}(\alpha_{i0},\bm\beta_0)\widehat\Delta_i
+
\widehat\Delta_i^2\int_0^1(1-v)
\ddot a_{it}(\alpha_{i0}+v\widehat\Delta_i,\bm\beta_0)\,dv,
\]
where \(\dot a_{it}\)
and \(\ddot a_{it}\) denote the first and second derivatives of
\(a_{it}\) with respect to \(\alpha_i\). Let
\begin{align*}
L_{nT}
&=
\frac1n\sum_{i=1}^{n}\widehat\Delta_i
\left[
\frac1T\sum_{t=1}^{T}
\left\{
\bm z_{it}\dot a_{it}(\alpha_{i0},\bm\beta_0)
-
\E[\bm z_{it}\dot a_{it}(\alpha_{i0},\bm\beta_0)]
\right\}
\right]\\
&=\frac1n\sum_{i=1}^{n}\widehat\Delta_i
\left[
\frac1T\sum_{t=1}^{T}
\left\{
\bm z_{it}K_{h_n}(-u_{it})
-
\E[\bm z_{it}K_{h_n}(-u_{it})]
\right\}
\right].
\end{align*}
Combining the population expansion with \eqref{eq:sample flucutation} gives
\begin{align*}
&\frac1n\sum_{i=1}^{n}\frac1T\sum_{t=1}^{T}
\bm z_{it}
\{a_{it}(\alpha_{i0}+\widehat\Delta_i,\bm\beta_0)
-a_{it}(\alpha_{i0},\bm\beta_0)\}=
\frac1n\sum_{i=1}^{n}
\E[f_i(0\mid\bm q_{it})\bm z_{it}]\widehat\Delta_i
+L_{nT}+R_{mn},
\end{align*}
where, by the integral form of Taylor's theorem,
\[
\|R_{mn}\|
=
O_p\left(
\frac1n\sum_{i=1}^n\widehat\Delta_i^2
\right)
+
O_p\left(
h_n^r\frac1n\sum_{i=1}^n|\widehat\Delta_i|
\right)
+
O_p\left[
\left(\frac1n\sum_{i=1}^n\widehat\Delta_i^2\right)
\sqrt{\frac{\log(nT)}{Th_n^3}}
\right].
\]
The last rate follows by applying the uniform kernel bound to \(K_{h_n}^{(1)}\), with the bounded
multiplier \(\bm z_{it}\), which gives
\[
\max_{i\le n}\sup_{|a|\le C\sqrt{\log(nT)/T}}
\left\|
\frac1T\sum_{t=1}^T
\left[
\bm z_{it}K_{h_n}^{(1)}(a-u_{it})
-
\E\{\bm z_{it}K_{h_n}^{(1)}(a-u_{it})\}
\right]
\right\|=O_p\Biggl(\sqrt{\frac{\log(nT)}{Th_n^3}}\Biggr).
\]
By Lemma \ref{lem:alpha_profile},
\(
n^{-1}\sum_{i=1}^n\widehat\Delta_i^2=O_p(T^{-1})
\)
and
\(
n^{-1}\sum_{i=1}^n|\widehat\Delta_i|=O_p(T^{-1/2}).
\)
Therefore,
\[
\sqrt{nT}\|R_{mn}\|
=
O_p\left(\sqrt{\frac nT}\right)
+
O_p(\sqrt n\,h_n^r)
+
O_p\left[
\sqrt{\frac nT}
\sqrt{\frac{\log(nT)}{Th_n^3}}
\right]
=
o_p(1),
\]
where the last equality follows from \(n/T\to0\), \(Th_n^r\to0\), and
\(h_n\asymp T^{-c}\) with \(c<1/3\). Consequently, we can rewrite $\widehat m_n(\bm\beta_0)$ as
\begin{align}\label{eq:mnbeta}
\widehat m_n(\bm\beta_0)
=
\widehat{H}_{n}(\bm\alpha_0,\bm\beta_0)
+\frac1n\sum_{i=1}^n
\E[f_i(0\mid\bm q_{it})\bm z_{it}]\widehat\Delta_i
+L_{nT}
+o_p((nT)^{-1/2}).
\end{align}

\textbf{Decompose $\frac1n\sum_{i=1}^n
\E[f_i(0\mid\bm q_{it})\bm z_{it}]\widehat\Delta_i+L_{nT}$.}
Recall \(R_{Ai}\) and
\(R_{Pi}\) as defined in \eqref{eq: RAAi} and \eqref{eq: Rpi},
respectively. We have
\begin{align}
\widehat\Delta_i
={}&-s_i\widehat{A}_{i}(\alpha_{i0},\bm\beta_0)
-s_i\E[f_i(0\mid\bm q_{it})\bm x_{it}']\Delta_\beta
-s_iR_{Pi}-s_iR_{Ai}\nonumber\\
={}&-s_i\widehat{A}_{i}(\alpha_{i0},\bm\beta_0)
-\bm\gamma_i'\Delta_\beta-s_iR_{Pi}-s_iR_{Ai}.
\label{eq: hatdelta}
\end{align}
Evaluating \eqref{eq: hatdelta} at \(\bm\beta=\bm\beta_0\) gives
\(
\widehat\Delta_i
=-s_i\widehat{A}_{i}(\alpha_{i0},\bm\beta_0)
-s_iR_{Pi}-s_iR_{Ai}.
\) We have 
\begin{align}\frac1n\sum_{i=1}^n
\E[f_i(0\mid\bm q_{it})\bm z_{it}]\widehat\Delta_i+L_{nT}=-\frac1n\sum_{i=1}^n
\E[f_i(0\mid\bm q_{it})\bm z_{it}]s_i\widehat{A}_{i}(\alpha_{i0},\bm\beta_0)\nonumber\\+C_{nT}-\frac1n\sum_{i=1}^n
\E[f_i(0\mid\bm q_{it})\bm z_{it}]s_iR_{Pi},\label{eq:E+LnT}\end{align} where $C_{nT}
=
L_{nT}
-
\frac1n\sum_{i=1}^n
\E[f_i(0\mid\bm q_{it})\bm z_{it}]s_iR_{Ai}.$
We now  show that
\[
\sqrt{nT}\frac1n\sum_{i=1}^n
\E[f_i(0\mid\bm q_{it})\bm z_{it}]s_iR_{Pi}
=o_p(1)\quad\text{and}\quad
\sqrt{nT}\,C_{nT}=o_p(1).
\]
We treat these two remainders in turn.

Consider the population remainder,
by the smoothness of the conditional density and the order of the
smoothing bias,
\(
|R_{Pi}|\leq C\widehat\Delta_i^2+Ch_n^r
\)
uniformly over \(i\). Since
\(
n^{-1}\sum_{i=1}^n\widehat\Delta_i^2=O_p(T^{-1}),
\)
\[
\left\|
\sqrt{nT}\frac1n\sum_{i=1}^n
\E[f_i(0\mid\bm q_{it})\bm z_{it}]s_iR_{Pi}
\right\|
\leq
C\sqrt{nT}
\left(
\frac1n\sum_{i=1}^n\widehat\Delta_i^2+h_n^r
\right)
=o_p(1),
\]
where the last equality follows from \(n/T\to0\) and
\(\sqrt{nT}h_n^r\to0\).

\textbf{Empirical-process remainder.}
The term
\(
\sqrt{nT}\,n^{-1}\sum_{i=1}^n\bm\delta_iR_{Ai}
\)
can be written as
\[
\frac1{\sqrt{nT}}\sum_{i=1}^n\sum_{t=1}^T
\bm\delta_i
\left[
G_{h_n}(\widehat\Delta_i-u_{it})-G_{h_n}(-u_{it})
-
\left.\E\{G_{h_n}(d-u_{it})-G_{h_n}(-u_{it})\}
\right|_{d=\widehat\Delta_i}
\right].
\]
Here, the population expectation is calculated for deterministic \(d\)
and then evaluated at \(\widehat\Delta_i\).
For deterministic \(d\), the integral form of Taylor's theorem gives
\[
G_{h_n}(d-u_{it})-G_{h_n}(-u_{it})
=
d\cdot K_{h_n}(-u_{it})
+d^2\cdot \int_0^1(1-v)K_{h_n}^{(1)}(vd-u_{it})\,dv.
\]
Subtracting the expectation and averaging over \(t\), with
\(d=\widehat\Delta_i\), we have
\[
R_{Ai}
=
\widehat\Delta_i
\left[
\frac1T\sum_{t=1}^T
\{K_{h_n}(-u_{it})-\E[K_{h_n}(-u_{it})]\}
\right]
+\widetilde R_{Ai},
\]
where
\(
\widetilde R_{Ai}
=
\left.
d^2\int_0^1(1-v)\frac1T\sum_{t=1}^T
\left[
K_{h_n}^{(1)}(vd-u_{it})
-\E[K_{h_n}^{(1)}(vd-u_{it})]
\right]\,dv
\right|_{d=\widehat\Delta_i}.
\)
Since \(\int_0^1(1-v)\,dv=1/2\), this remainder satisfies
\[
|\widetilde R_{Ai}|
\leq
\frac{\widehat\Delta_i^2}{2}
\sup_{|d|\leq\max_{j\leq n}|\widehat\Delta_j|}
\left|
\frac1T\sum_{t=1}^T
\left[
K_{h_n}^{(1)}(d-u_{it})
-\E[K_{h_n}^{(1)}(d-u_{it})]
\right]
\right|.
\]
Consequently, the uniform kernel bound gives
\(
\frac1n\sum_{i=1}^n|\widetilde R_{Ai}|
=
o_p\left(
\frac{\log(nT)}{(Th_n^3)^{1/2}}
\frac1n\sum_{i=1}^n\widehat\Delta_i^2
\right).
\)
Since \(h_n=T^{-c}\) with \(c<1/3\), \(n/T\to0\), and
\(n^{-1}\sum_i\widehat\Delta_i^2=O_p(T^{-1})\),
\[
\frac{\log(nT)}{(Th_n^3)^{1/2}}
\frac1n\sum_{i=1}^n\widehat\Delta_i^2
=
O_p\left(
\frac{\log(nT)}{T^{(1-3c)/2}}\frac1T
\right)
=
o_p((nT)^{-1/2}).
\]
Therefore,
\begin{align*}
&\sqrt{nT}\frac1n\sum_{i=1}^n
\E[f_i(0\mid\bm q_{it})\bm z_{it}]s_iR_{Ai}\\
&\quad=
\sqrt{nT}\frac1n\sum_{i=1}^n
\E[f_i(0\mid\bm q_{it})\bm z_{it}]s_i\widehat\Delta_i
\left[
\frac1T\sum_{t=1}^T
\{K_{h_n}(-u_{it})-\E[K_{h_n}(-u_{it})]\}
\right]
+o_p(1).
\end{align*}

Using
\(
\widehat\Delta_i
=-s_i\widehat{A}_{i}(\alpha_{i0},\bm\beta_0)
+r_{\alpha i}(\bm\beta_0),
\)
the leading term becomes
\begin{align}
&-\sqrt{nT}\frac1n\sum_{i=1}^n
\E[f_i(0\mid\bm q_{it})\bm z_{it}]s_i^2
\widehat{A}_{i}(\alpha_{i0},\bm\beta_0)
\left[
\frac1T\sum_{t=1}^T
\{K_{h_n}(-u_{it})-\E[K_{h_n}(-u_{it})]\}
\right]\nonumber\\
&\quad+
\sqrt{nT}\frac1n\sum_{i=1}^n
\E[f_i(0\mid\bm q_{it})\bm z_{it}]s_i
r_{\alpha i}(\bm\beta_0)
\left[
\frac1T\sum_{t=1}^T
\{K_{h_n}(-u_{it})-\E[K_{h_n}(-u_{it})]\}
\right].
\label{eq: leading terms}
\end{align}
To control the second term, let
\[
\bar K_i
=
\frac1T\sum_{t=1}^T
\{K_{h_n}(-u_{it})-\E[K_{h_n}(-u_{it})]\}.
\]
Since
\(
\sup_i\|\E[f_i(0\mid\bm q_{it})\bm z_{it}]s_i\|<C,
\)
it suffices to show
\(
\sqrt{nT}n^{-1}\sum_i
|r_{\alpha i}(\bm\beta_0)\bar K_i|=o_p(1).
\)
A second-order Taylor expansion of the profile equation gives
\begin{align}
|r_{\alpha i}(\bm\beta_0)|
\leq
C\{|\bar K_i|+h_n^r\}|\widehat\Delta_i|
+
CM_n\widehat\Delta_i^2,
\label{eq:ralphai}
\end{align}
where
\(
M_n=
\max_{i\le n}
\sup_{|a|\le\max_{j\le n}|\widehat\Delta_j|}
\left|
\frac1T\sum_{t=1}^T
K_{h_n}^{(1)}(a-u_{it})
\right|.
\)
The uniform kernel laws of large numbers give
\begin{align}
\max_{i\le n}|\bar K_i|
=
O_p\left(\sqrt{\frac{\log(nT)}{Th_n}}\right),&&M_n
=
O_p\left(
1+\sqrt{\frac{\log(nT)}{Th_n^3}}
\right).
\label{eq: Mn property}
\end{align}
Moreover,
\(
n^{-1}\sum_i\bar K_i^2=O_p((Th_n)^{-1}),
\)
while Lemma \ref{lem:alpha_max} gives
\(
\max_i|\widehat\Delta_i|
=O_p(\sqrt{\log(nT)/T})
\)
and
\(
n^{-1}\sum_i\widehat\Delta_i^2=O_p(T^{-1}).
\)

Recall that
\(
\bm\kappa_{it}
=
(\bm z_{it}-\bm\delta_i)K_{h_n}(-u_{it}).
\)
The uniform kernel bound gives
\begin{align}
\max_{i\le n}
\left\|
\frac1T\sum_{t=1}^T
\{\bm\kappa_{it}-\E[\bm\kappa_{it}]\}
\right\|
&=
O_p\left(\sqrt{\frac{\log(nT)}{Th_n}}\right),\nonumber\\
\frac1n\sum_{i=1}^n
\left\|
\frac1T\sum_{t=1}^T
\{\bm\kappa_{it}-\E[\bm\kappa_{it}]\}
\right\|^2
&=
O_p\left(\frac1{Th_n}\right).
\label{eq: kappa property}
\end{align}
Since
\(
\bm\delta_i=s_i\E[f_i(0\mid\bm q_{it})\bm z_{it}],
\)
combining \(L_{nT}\) with the leading part of \(R_{Ai}\) gives
\begin{align*}
C_{nT}
&=
\frac1n\sum_{i=1}^n\widehat\Delta_i
\left[
\frac1T\sum_{t=1}^T
\{\bm z_{it}K_{h_n}(-u_{it})
-\E[\bm z_{it}K_{h_n}(-u_{it})]\}
\right]\\
&\quad-
\frac1n\sum_{i=1}^n\bm\delta_i\widehat\Delta_i
\left[
\frac1T\sum_{t=1}^T
\{K_{h_n}(-u_{it})-\E[K_{h_n}(-u_{it})]\}
\right]
-\frac1n\sum_{i=1}^n\bm\delta_i\widetilde R_{Ai}\\
&=
\frac1n\sum_{i=1}^n\widehat\Delta_i
\left[
\frac1T\sum_{t=1}^T
\{\bm\kappa_{it}-\E[\bm\kappa_{it}]\}
\right]
+o_p((nT)^{-1/2})\\
&=
-\frac1n\sum_{i=1}^ns_i
\widehat{A}_{i}(\alpha_{i0},\bm\beta_0)
\left[
\frac1T\sum_{t=1}^T
\{\bm\kappa_{it}-\E[\bm\kappa_{it}]\}
\right]\\
&\quad+
\frac1n\sum_{i=1}^nr_{\alpha i}(\bm\beta_0)
\left[
\frac1T\sum_{t=1}^T
\{\bm\kappa_{it}-\E[\bm\kappa_{it}]\}
\right]
+o_p((nT)^{-1/2}).
\end{align*}

Using the preceding bounds, the Cauchy-Schwarz inequality,
\eqref{eq:ralphai}, \eqref{eq: Mn property}, and
\eqref{eq: kappa property}, we obtain
{\small
\begin{align*}
&\sqrt{nT}\frac1n\sum_{i=1}^n
|r_{\alpha i}(\bm\beta_0)|
\left\|
\frac1T\sum_{t=1}^T
\{\bm\kappa_{it}-\E[\bm\kappa_{it}]\}
\right\|\\
&=
O_p\left(\frac{\sqrt{n\log(nT)}}{Th_n}\right)
+
O_p\left(\sqrt{\frac nT}h_n^{r-1/2}\right)
+
O_p\left(\frac{\sqrt{n\log(nT)}}{Th_n^{1/2}}\right)\\
&\quad+
O_p\left(\frac{\sqrt n\log(nT)}{T^{3/2}h_n^2}\right)
=o_p(1).
\end{align*}}
For example,
\(
\sqrt{n\log(nT)}/(Th_n)
\leq T^{c-1/2}\sqrt{\log T}=o(1),
\)
and the remaining terms follow similarly from
\(h_n\asymp T^{-c}\), \(c<1/3\), \(r\geq4\), and \(n/T\to0\).
Therefore,
\[
C_{nT}
=
-\frac1n\sum_{i=1}^ns_i
\widehat{A}_{i}(\alpha_{i0},\bm\beta_0)
\left[
\frac1T\sum_{t=1}^T
\{\bm\kappa_{it}-\E[\bm\kappa_{it}]\}
\right]
+o_p((nT)^{-1/2}).
\]

Since
\(
a_{it}(\alpha_{i0},\bm\beta_0)=\omega_{it},
\)
stationarity and the definition of \(\bm\Lambda_{ni}\) in Assumption
(A10) imply
\begin{align*}
&\E\left[
-\frac1n\sum_{i=1}^ns_i
\widehat{A}_{i}(\alpha_{i0},\bm\beta_0)
\left\{
\frac1T\sum_{t=1}^T
(\bm\kappa_{it}-\E[\bm\kappa_{it}])
\right\}
\right]\\
&\quad=
-\frac1{nT}\sum_{i=1}^ns_i
\sum_{|j|\le T-1}
\left(1-\frac{|j|}{T}\right)
\E\left[
\omega_{i1}
\{\bm\kappa_{i,1+j}-\E[\bm\kappa_{it}]\}
\right]\\
&\quad=
-\frac1{nT}\sum_{i=1}^ns_i\bm\Lambda_{ni}
=
O(T^{-1}).
\end{align*}

For the centered part, compact support of \(K\), boundedness of
\(\bm z_{it}\) and \(\bm\delta_i\), and boundedness of the conditional
density imply that, for every fixed \(p\geq2\),
\[
\E\left\|
(\bm z_{it}-\bm\delta_i)K(-u_{it}/h_n)
\right\|^p
=
O(h_n).
\]
Write
\(
\widetilde\omega_{it}=\omega_{it}-\E\omega_{it}
\)
and
\(
\widetilde{\bm\kappa}_{it}
=\bm\kappa_{it}-\E\bm\kappa_{it}.
\)
By bounded support of \(K\) and boundedness of
\(\bm z_{it}-\bm\delta_i\),
\[
\|\bm\kappa_{it}\|
\leq
Ch_n^{-1}1\{|u_{it}|\leq Ch_n\}.
\]
When \(t=t'\), Assumption (A6) gives
\[
\E\|\widetilde{\bm\kappa}_{it}\|^2
\leq
\E\|\bm\kappa_{it}\|^2
\leq
Ch_n^{-2}P(|u_{it}|\leq Ch_n)
=
O(h_n^{-1}).
\]
When \(t\neq t'\), stationarity and the two-time local concentration
condition in Assumption (A6) give
\[
\E\|\bm\kappa_{it}\bm\kappa_{it'}'\|
\leq
Ch_n^{-2}
P\left(
|u_{it}|\leq Ch_n,\,
|u_{it'}|\leq Ch_n
\right)
=
O(1),
\]
uniformly in \(i,t,t'\). Since
\(\E\|\bm\kappa_{it}\|=O(1)\), the same order holds after centering.

All bounds below hold uniformly in \(i\). The preceding moment bounds
and (A1) imply
\[
\|\Cov(\bm\kappa_{i1},\bm\kappa_{i,1+j})\|
\leq C\min\{1,h_n^{-2}a^j\},\qquad j\geq1.
\]
Summing the constant bound up to a sufficiently large multiple of
\(\log(1/h_n)\), and the geometric bound thereafter, gives
\[
\begin{aligned}
\sum_{j\in\mathbb Z}
\|\Cov(\bm\kappa_{i1},\bm\kappa_{i,1+j})\|
&\leq Ch_n^{-1}
+C\sum_{j=1}^{\infty}\min\{1,h_n^{-2}a^j\}=O\{h_n^{-1}+\log(1/h_n)\}
=O(h_n^{-1}).
\end{aligned}
\]

Choose \(q=\lceil C\log T\rceil\), with \(C\) sufficiently large that
\(a^q\leq T^{-4}\), and define
\(
\bm B_{ik}
=\sum_{t=(k-1)q+1}^{\min\{kq,T\}}
\widetilde{\bm\kappa}_{it}
\),
\(k=1,\ldots,\lceil T/q\rceil\).
The covariance bound and
\(\E\|\widetilde{\bm\kappa}_{it}\|^4=O(h_n^{-3})\) imply
\[
\E\|\bm B_{ik}\|^2\leq Cqh_n^{-1},
\qquad
\E\|\bm B_{ik}\|^4
\leq q^3
\sum_{t=(k-1)q+1}^{\min\{kq,T\}}
\E\|\widetilde{\bm\kappa}_{it}\|^4
\leq Cq^4h_n^{-3}.
\]

Separate the blocks into the groups \(\{1,3,5,\ldots\}\) and
\(\{2,4,6,\ldots\}\). Successive blocks within each group have an
intervening block of length \(q\). Berbee's lemma and (A1) therefore
allow us to construct \(\bm B_{ik}^*\), independent within each group,
such that
\[
\bm B_{ik}^*\stackrel{d}{=}\bm B_{ik},
\qquad
P(\bm B_{ik}\neq\bm B_{ik}^*)
\leq\beta_i(q)\leq Ba^q.
\]
Each \(\bm B_{ik}^*\) may depend on its corresponding original block
\(\bm B_{ik}\). 
By the union bound,
\[
P\left(\bm B_{ik}\neq\bm B_{ik}^*
       \text{ for some }k\right)
\leq
\sum_{k=1}^{\lceil T/q\rceil}
P(\bm B_{ik}\neq\bm B_{ik}^*)
\leq C(T/q)a^q.
\]
The two sums $\frac1T\sum_k\bm B_{ik}$ and $\frac1T\sum_k\bm B_{ik}^*$ are equal on the event that every
\(\bm B_{ik}=\bm B_{ik}^*\). Using
\(\|T^{-1}\sum_k\bm B_{ik}\|\leq C/h_n\) and
splitting the expectation over these two events gives
\[
\begin{aligned}
\E\left\|\frac1T\sum_k\bm B_{ik}\right\|^4
&\leq
\E\left\|\frac1T\sum_k\bm B_{ik}^*\right\|^4
+\frac{C}{h_n^4}
P\left(\bm B_{ik}\neq\bm B_{ik}^*
       \text{ for some }k\right)\\
&\leq
\E\left\|\frac1T\sum_k\bm B_{ik}^*\right\|^4
+\frac{CTa^q}{q h_n^4}.
\end{aligned}
\]

We next bound the fourth moment of the constructed sum.
Within either group, independence and centering eliminate all
terms in the fourth-moment expansion containing a block only once.
Bounding the remaining inner products by Cauchy-Schwarz gives,
with all indices in the following display restricted to that group,
\[
\begin{aligned}
\E\left\|\sum_k\bm B_{ik}^*\right\|^4
&\leq
\sum_k\E\|\bm B_{ik}^*\|^4
+6\sum_{k<\ell}
\E\|\bm B_{ik}^*\|^2\,
\E\|\bm B_{i\ell}^*\|^2\\
&\leq
\sum_k\E\|\bm B_{ik}^*\|^4
+3\left\{\sum_k\E\|\bm B_{ik}^*\|^2\right\}^2.
\end{aligned}
\]
Combining the two groups, and using equality of the individual
block moments, therefore yields
\[
\begin{aligned}
\E\left\|\sum_k\bm B_{ik}^*\right\|^4
&\leq
8\E\left\|\sum_{k=1,3,\ldots}\bm B_{ik}^*\right\|^4
+8\E\left\|\sum_{k=2,4,\ldots}\bm B_{ik}^*\right\|^4\\
&\leq C\left[
\sum_k\E\|\bm B_{ik}\|^4
+\left\{\sum_k\E\|\bm B_{ik}\|^2\right\}^2
\right],
\end{aligned}
\]
where the sums in the last line range over all blocks.

Since there are \(O(T/q)\) blocks and
\(\sum_k\bm B_{ik}=\sum_{t=1}^T\widetilde{\bm\kappa}_{it}\),
the preceding inequalities give
\[
\begin{aligned}
\E\left\|\frac1T\sum_{t=1}^T
\widetilde{\bm\kappa}_{it}\right\|^4
&\leq
\frac{C}{T^4}
\left[
\frac{T}{q}\frac{q^4}{h_n^3}
+\left\{\frac{T}{q}\frac{q}{h_n}\right\}^2
\right]
+\frac{CTa^q}{q h_n^4}\\
&\leq C\left(
\frac{q^3}{T^3h_n^3}
+\frac1{T^2h_n^2}
+\frac{Ta^q}{h_n^4}
\right)
=O(T^{-2}h_n^{-2}).
\end{aligned}
\]
Here, the first and third terms
are negligible because
\[
\frac{q^3}{Th_n}\to0,
\qquad
\frac{T^3a^q}{h_n^2}
\leq\frac1{Th_n^2}\to0,
\]
by the choice of \(q\) and (A9).

For the bounded centered scores \(\widetilde\omega_{it}\), the
corresponding block second and fourth moments are \(O(q)\) and
\(O(q^4)\). Repeating the calculation above, and using
\(
\widehat{A}_{i}(\alpha_{i0},\bm\beta_0)
=T^{-1}\sum_{t=1}^T\widetilde\omega_{it}+\E\omega_{i1}
\)
with \(\E\omega_{i1}=O(h_n^r)\), gives
\[
\begin{aligned}
\E|\widehat{A}_{i}(\alpha_{i0},\bm\beta_0)|^4
&\leq
8\E\left|\frac1T\sum_{t=1}^T\widetilde\omega_{it}\right|^4
+8|\E\omega_{i1}|^4\leq C\left(T^{-2}+q^3T^{-3}+Ta^q+h_n^{4r}\right)
=O(T^{-2}),
\end{aligned}
\]
where the last equality follows from (A9) and the choice of \(q\).
Consequently, Cauchy-Schwarz yields
\[
\begin{aligned}
&\E\left[
\widehat{A}_{i}(\alpha_{i0},\bm\beta_0)^2
\left\|\frac1T\sum_{t=1}^T
\widetilde{\bm\kappa}_{it}\right\|^2
\right]\leq
\left\{\E|\widehat{A}_{i}(\alpha_{i0},\bm\beta_0)|^4\right\}^{1/2}
\left\{
\E\left\|\frac1T\sum_{t=1}^T
\widetilde{\bm\kappa}_{it}\right\|^4
\right\}^{1/2}
=O\left(\frac1{T^2h_n}\right),
\end{aligned}
\]
uniformly in \(i\).
Using independence across \(i\),
\[
\Var\left(
\frac1n\sum_{i=1}^ns_i
\widehat{A}_{i}(\alpha_{i0},\bm\beta_0)
\left[
\frac1T\sum_{t=1}^T
\{\bm\kappa_{it}-\E[\bm\kappa_{it}]\}
\right]
\right)
=
O\left(\frac1{nT^2h_n}\right).
\]
Consequently,
\(
C_{nT}
=
O(T^{-1})
+
O_p((T\sqrt{nh_n})^{-1})
+
o_p((nT)^{-1/2}),
\)
and hence
\[
\sqrt{nT}\,C_{nT}
=
O\left(\sqrt{\frac nT}\right)
+
O_p\left(\frac1{\sqrt{Th_n}}\right)
+
o_p(1)
=
o_p(1).
\]

\textbf{Conclusion.}
Combining the bounds for \(R_{Pi}\) and \(C_{nT}\) with the profile
expansion, \eqref{eq:mnbeta}, and \eqref{eq:E+LnT} yields
\[
\widehat m_n(\bm\beta_0)
=
\widehat{H}_{n}(\bm\alpha_0,\bm\beta_0)
-
\frac1n\sum_{i=1}^n
\bm\delta_i\widehat{A}_{i}(\alpha_{i0},\bm\beta_0)
+
o_p((nT)^{-1/2}).
\]
Since
\(
\widehat{H}_{n}(\bm\alpha_0,\bm\beta_0)
-
\frac1n\sum_{i=1}^n
\bm\delta_i\widehat{A}_{i}(\alpha_{i0},\bm\beta_0)
=
\frac1{nT}\sum_{i=1}^n\sum_{t=1}^T
(\bm z_{it}-\bm\delta_i)
\{G_{h_n}(-u_{it})-\tau\},
\)
Lemma \ref{lem:smoothing} gives
\[
\sqrt{nT}
\left[
\widehat{H}_{n}(\bm\alpha_0,\bm\beta_0)
-
\frac1n\sum_{i=1}^n
\bm\delta_i\widehat{A}_{i}(\alpha_{i0},\bm\beta_0)
\right]
=
\frac1{\sqrt{nT}}\sum_{i=1}^n\sum_{t=1}^T
\psi_{it}
+
o_p(1).
\]
Combining the last two displays proves the result.

\end{proof}

\begin{lemma} \label{lem:derivative} Under the assumptions of Theorem
4.1
{}, uniformly for $\|\bm{\beta}-\bm{\beta}_{0}\|=o_{p}(1)$,
\[
\frac{\partial\widehat{m}_{n}(\bm{\beta})}{\partial\bm{\beta}'}=D_{n}+o_{p}(1).
\]
In particular,
\(
\partial\widehat m_n(\bm\beta_0)/\partial\bm\beta'=D_n+o_p(1),
\)
and $D_n\to D$.
\end{lemma}
\begin{proof}
 By Lemma \ref{lem:profile_unique}, the sample profile is uniquely
defined and continuously differentiable in a shrinking neighborhood
of \(\bm\beta_0\). By the implicit function theorem applied to $\widehat{A}_{i}(\widehat{\alpha}_{i}(\bm{\beta}),\bm{\beta})=0$,
\begin{align*}
\frac{\partial\widehat\alpha_i(\bm\beta)}{\partial\bm\beta'}
={}&-\left\{\frac1T\sum_{t=1}^T
K_{h_n}(\widehat\alpha_i(\bm\beta)+\bm x_{it}'\bm\beta-y_{it})\right\}^{-1}\left\{\frac1T\sum_{t=1}^T
K_{h_n}(\widehat\alpha_i(\bm\beta)+\bm x_{it}'\bm\beta-y_{it})\bm x_{it}'\right\}.
\end{align*}
At $\bm{\beta}=\bm{\beta}_{0}$, the denominator and numerator in
the implicit-function expression have the desired uniform limits.
To see this, write
$\widehat\Delta_i(\bm\beta)=\widehat\alpha_i(\bm\beta)-\alpha_{i0}$.
For the denominator, uniformly over $i$,
\[
\frac1T\sum_{t=1}^T
K_{h_n}(\widehat\alpha_i(\bm\beta_0)+\bm x_{it}'\bm\beta_0-y_{it})
=\frac1T\sum_{t=1}^T K_{h_n}(\widehat\Delta_i(\bm\beta_0)-u_{it}).
\]
By Lemma \ref{lem:alpha_max}, $\max_{i}|\widehat{\Delta}_{i}(\bm{\beta}_{0})|=o_{p}(1)$.
Hence the uniform kernel law of large numbers as above,
applied over the local set $\max_{i}|\widehat{\Delta}_{i}|=o(1)$,
gives 
\[
\max_{i}\left|\frac{1}{T}\sum_{t=1}^{T}K_{h_{n}}(\widehat{\Delta}_{i}(\bm{\beta}_{0})-u_{it})-\left.\E[K_{h_n}(d-u_{it})]\right|_{d=\widehat\Delta_i(\bm\beta_0)}\right|=o_{p}(1).
\]
Moreover, the population kernel expansion in Lemma \ref{lem:smoothing} gives
\[
\left.\E[K_{h_n}(d-u_{it})]\right|_{d=\widehat\Delta_i}=f_{i}(\widehat{\Delta}_{i})+O(h_{n}^{r})
\]
uniformly in $i$ and $|\widehat{\Delta}_{i}|=o(1)$. Since the densities
are uniformly continuous and $\max_{i}|\widehat{\Delta}_{i}(\bm{\beta}_{0})|=o_{p}(1)$,
it follows that 
\[
\max_{i}\left|\frac{1}{T}\sum_{t=1}^{T}K_{h_{n}}(\widehat{\alpha}_{i}(\bm{\beta}_{0})+\bm{x}_{it}'\bm{\beta}_{0}-y_{it})-f_{i}(0)\right|=o_{p}(1).
\]

The numerator is handled in the same way. The uniform kernel LLN with
the bounded multiplier $\bm x_{it}$ gives
\begin{align*}
\max_i\Bigg\|{}&\frac1T\sum_{t=1}^T
K_{h_n}(\widehat\Delta_i(\bm\beta_0)-u_{it})\bm x_{it}'-\left.\E\{K_{h_n}(d-u_{it})\bm x_{it}'\}\right|_{d=\widehat\Delta_i(\bm\beta_0)}\Bigg\|
=o_p(1).
\end{align*}
Conditioning on \(\bm q_{it}\) and changing variables gives
\[
\E\{K_{h_n}(d-u_{it})\bm x_{it}'\}
=
\E\left[
\bm x_{it}'\int K(v)f_i(d-h_nv\mid\bm q_{it})\,dv
\right]
=
\E[f_i(d\mid\bm q_{it})\bm x_{it}']+O(h_n^r),
\]
uniformly in \(i\) and deterministic \(|d|=o(1)\). The last equality
follows by expanding the density around \(d\), using the vanishing
kernel moments in (A8), the uniformly bounded \(r\)-th derivative
in (A6), and the boundedness of \(\bm x_{it}\) in (A2).
Evaluating at $d=\widehat\Delta_i$ and using the uniform continuity of $f_{i}(\cdot\mid\bm{q}_{it})$, the
boundedness of $\bm{x}_{it}$, and $\max_{i}|\widehat{\Delta}_{i}(\bm{\beta}_{0})|=o_{p}(1)$,
we obtain 
\[
\max_{i}\left\Vert \frac{1}{T}\sum_{t=1}^{T}K_{h_{n}}(\widehat{\alpha}_{i}(\bm{\beta}_{0})+\bm{x}_{it}'\bm{\beta}_{0}-y_{it})\bm{x}_{it}'-\E[f_{i}(0\mid\bm{q}_{it})\bm{x}_{it}']\right\Vert =o_{p}(1).
\]

The same argument holds uniformly for $\bm{\beta}$ in a shrinking
neighborhood of $\bm{\beta}_{0}$. Indeed, Lemma \ref{lem:alpha_max}
implies $\max_{i}\sup_{\|\bm{\beta}-\bm{\beta}_{0}\|\le\eta_{n}}|\widehat{\alpha}_{i}(\bm{\beta})-\alpha_{i0}|=o_{p}(1)$,
and the additional shift $\bm{x}_{it}'(\bm{\beta}-\bm{\beta}_{0})$
is uniformly $o_{p}(1)$ because $\bm{x}_{it}$ is bounded and $\eta_{n}\to0$.
The preceding stochastic and population kernel approximations therefore
apply uniformly over this neighborhood. Hence 
\begin{align}
\max_{1\leq i\leq n}\left\Vert \frac{\partial\widehat{\alpha}_{i}(\bm{\beta})}{\partial\bm{\beta}'}+\bm{\gamma}_{i}'\right\Vert =o_{p}(1).\label{eq: gamma consistency}
\end{align}
Differentiating the profiled moment gives 
\begin{align*}
\frac{\partial\widehat m_n(\bm\beta)}{\partial\bm\beta'}
={}&\frac1{nT}\sum_{i=1}^n\sum_{t=1}^T
K_{h_n}(\widehat\alpha_i(\bm\beta)+\bm x_{it}'\bm\beta-y_{it})\bm z_{it}\left\{\bm x_{it}'
+\frac{\partial\widehat\alpha_i(\bm\beta)}{\partial\bm\beta'}\right\}.
\end{align*}
Using the previous display, boundedness of $\bm{x}_{it}$ and $\bm{z}_{it}$,
and the kernel law of large numbers under exponential $\beta$-mixing,
\[
\frac{\partial\widehat{m}_{n}(\bm{\beta})}{\partial\bm{\beta}'}=\frac{1}{nT}\sum_{i=1}^{n}\sum_{t=1}^{T}K_{h_{n}}(u_{it})\bm{z}_{it}(\bm{x}_{it}-\bm{\gamma}_{i})'+o_{p}(1).
\]
The expectation of the first term is 
\[
\frac{1}{n}\sum_{i=1}^{n}\E[f_{i}(0\mid\bm{q}_{it})\bm{z}_{it}(\bm{x}_{it}-\bm{\gamma}_{i})']+O(h_{n}^{r})=D_{n}+O(h_{n}^{r}),
\]
and the difference is of order $o_{p}(1)$ by the same kernel LLN.
Therefore $\partial\widehat{m}_{n}(\bm{\beta})/\partial\bm{\beta}'=D_{n}+o_{p}(1)$.
The convergence $D_{n}\to D$ is Assumption (A7). 
\end{proof}

By consistency, with probability approaching one, $\widehat{\bm{\beta}}$
lies in an interior neighborhood of $\bm{\beta}_{0}$. Hence the first-order
condition for the profiled GMM problem is 
\[
\frac{\partial\widehat{m}_{n}(\widehat{\bm{\beta}})'}{\partial\bm{\beta}}\widehat{W}\widehat{m}_{n}(\widehat{\bm{\beta}})=\bm{0}.
\]
By the mean-value expansion of $\widehat{m}_{n}(\widehat{\bm{\beta}})$
around $\bm{\beta}_{0}$, $\widehat{m}_{n}(\widehat{\bm{\beta}})=\widehat{m}_{n}(\bm{\beta}_{0})+\bar{D}_{n}(\widehat{\bm{\beta}}-\bm{\beta}_{0}),$
where $\bar{D}_{n}=\int_0^1\frac{\partial \widehat{m}_{n}(\bm \beta_0+s(\widehat{\bm\beta}-\bm\beta_0))}{\partial \bm \beta'}ds$.
By Lemma \ref{lem:derivative}, $\bar{D}_{n}=D_{n}+o_{p}(1)$.
The derivative in the first-order condition is also $D_{n}+o_{p}(1)$,
and $\widehat{W}=W+o_{p}(1)$. Therefore, 
\[
\bm{0}=\{D_{n}+o_{p}(1)\}'\{W+o_{p}(1)\}\left[\widehat{m}_{n}(\bm{\beta}_{0})+\{D_{n}+o_{p}(1)\}(\widehat{\bm{\beta}}-\bm{\beta}_{0})\right].
\]
Because $D_{n}\to D$, $W$ is positive definite, and $D'WD$ is nonsingular
by Assumption (A7), the matrix $\{D_{n}+o_{p}(1)\}'\{W+o_{p}(1)\}\{D_{n}+o_{p}(1)\}$
is nonsingular with probability approaching one. Solving the previous
display gives 
\[
\widehat{\bm{\beta}}-\bm{\beta}_{0}=-(D_{n}'\widehat{W}D_{n})^{-1}D_{n}'\widehat{W}\widehat{m}_{n}(\bm{\beta}_{0})+r_{\beta n},
\]
where
\(
\|r_{\beta n}\|=o_p(\|\widehat m_n(\bm\beta_0)\|).
\)
By Lemma \ref{lem:profile_moment},
\[
\|\widehat m_n(\bm\beta_0)\|=O_p((nT)^{-1/2}),
\qquad \sqrt{nT}\,r_{\beta n}=o_p(1).
\]
Multiplying by $\sqrt{nT}$
and using Lemma \ref{lem:profile_moment}, 
\[
\sqrt{nT}(\widehat{\bm{\beta}}-\bm{\beta}_{0})=-(D_{n}'\widehat{W}D_{n})^{-1}D_{n}'\widehat{W}\frac{1}{\sqrt{nT}}\sum_{i=1}^{n}\sum_{t=1}^{T}\psi_{it}+o_{p}(1),
\]
where $\psi_{it}=(\bm{z}_{it}-\bm{\delta}_{i})\{1(u_{it}\leq0)-\tau\}.$
Since $D_{n}\to D$ and $\widehat{W}\stackrel{p}{\to}W$, this becomes 
\[
\sqrt{nT}(\widehat{\bm{\beta}}-\bm{\beta}_{0})=-(D'WD)^{-1}D'W\frac{1}{\sqrt{nT}}\sum_{i=1}^{n}\sum_{t=1}^{T}\psi_{it}+o_{p}(1).
\]
This proves the linear representation.

It remains to prove the central limit theorem for the leading term.
For any nonzero fixed vector $\bm{a}$ with the same dimension as $\psi_{it}$,
define $S_{ni}(\bm{a})=T^{-1/2}\sum_{t=1}^{T}\bm{a}'\psi_{it}.$ The
 quantile restriction implies  $\E[\psi_{it}]=\bm{0}$ and $\E[S_{ni}(\bm{a})]=0$. By
Assumption (A1), the vectors $S_{n1}(\bm{a}),\ldots,S_{nn}(\bm{a})$
are independent across $i$. By Assumption (A2), $\bm{a}'\psi_{it}$
is uniformly bounded. The exponential $\beta$-mixing condition implies
that, for some $\delta>0$,
\(
\max_{1\leq i\leq n}\E|S_{ni}(\bm{a})|^{2+\delta}\leq C.
\)
Therefore,
\[
\sum_{i=1}^{n}\E\left[\left|\frac{S_{ni}(\bm{a})}{\sqrt{n}}\right|^{2+\delta}\right]\leq Cn^{-\delta/2}\to0.
\]
The Lyapunov condition holds for the independent triangular array
$\{S_{ni}(\bm{a}):1\leq i\leq n\}$. Moreover, 
\[
\Var\left(\frac{1}{\sqrt{n}}\sum_{i=1}^{n}S_{ni}(\bm{a})\right)=\bm{a}'\left(\frac{1}{n}\sum_{i=1}^{n}\Omega_{ni}\right)\bm{a}\to\bm{a}'\Omega\bm{a}
\]
by Assumption (A10). Thus $\frac{1}{\sqrt{n}}\sum_{i=1}^{n}S_{ni}(\bm{a})\stackrel{d}{\to}N(0,\bm{a}'\Omega\bm{a}).$
By the Cramer-Wold device, $\frac{1}{\sqrt{nT}}\sum_{i=1}^{n}\sum_{t=1}^{T}\psi_{it}\stackrel{d}{\to}N(\bm{0},\Omega).$
Combining this convergence with the linear representation and Slutsky's
theorem gives 
\[
\sqrt{nT}(\widehat{\bm{\beta}}-\bm{\beta}_{0})\stackrel{d}{\to}N(\bm{0},\Sigma_{\beta}),
\]
where $\Sigma_{\beta}=(D'WD)^{-1}D'W\Omega WD(D'WD)^{-1}.$ If $W=\Omega^{-1}$,
then the usual simplification gives $\Sigma_{\beta}=(D'\Omega^{-1}D)^{-1}$.
This proves Theorem 4.1
{}. \end{proof}

\subsection{Proof of Theorem 5.1
{}}

\begin{proof}
\label{app:proof-thm2}
We prove the consistency of the feasible covariance matrix. For a generic
weighting matrix, it is enough to show that
\(\widehat{D} \stackrel{p}{\to} D\) and
\(\widehat{\Omega}\stackrel{p}{\to}\Omega\), because
\(\widehat W\stackrel{p}{\to}W\) is imposed in (A7). For the efficient
two-step estimator, we also verify below that the preliminary covariance
matrix $\bar\Omega$ converges to $\Omega$. We first focus on
\(\widehat{D}\) and \(\widehat{\Omega}\).
Throughout this proof, the uniform kernel and lag-uniform HAC bounds are
obtained by applying Corollary C.1 of \citet{GalvaoKato16} to the relevant
beta-mixing triangular arrays, with \(q_T=\lfloor(T/\log(nT))^{1/2}\rfloor\),
and then taking a union bound over the growing index set. The single condition
\(
\ell_n^2\log(nT)/(Th_n^2)\to0
\)
implies each auxiliary rate used below. In particular,
\[
\frac{\ell_n}{h_n}\sqrt{\frac{\log(nT)}{T}}\to0,
\qquad
\ell_n\sqrt{\frac{\log(nT)}{T}}\to0,
\qquad
\frac{\ell_n^2}{T}\to0.
\]
Moreover, Assumption (A9) gives \(h_n^r=O(T^{-1})\), so
\(\ell_nh_n^r=o(1)\). We invoke these implications explicitly in the following argument.

First consider \(\widehat D\). Recall that
\(
\widehat D=\frac{1}{nT}\sum_{i=1}^{n}\sum_{t=1}^{T}K_{h_n}(\widehat u_{it})\bm z_{it}(\bm x_{it}-\widehat{\bm\gamma}_i)'.
\)
For deterministic \((\alpha,\bm\beta)\), define the population maps
\[
\begin{aligned}
M_{zx,i}(\alpha,\bm\beta)
&=\E\left[K_{h_n}(y_{it}-\alpha-\bm x_{it}'\bm\beta)\bm z_{it}\bm x_{it}'\right],\\
M_{z,i}(\alpha,\bm\beta)
&=\E\left[K_{h_n}(y_{it}-\alpha-\bm x_{it}'\bm\beta)\bm z_{it}\right],\\
M_{1,i}(\alpha,\bm\beta)
&=\E\left[K_{h_n}(y_{it}-\alpha-\bm x_{it}'\bm\beta)\right].
\end{aligned}
\]
By the kernel uniform law of large numbers, uniformly over \(i\) and over
\((\alpha,\bm\beta)\in\mathcal A\times\mathcal B\), the corresponding
sample averages differ from \(M_{zx,i}(\alpha,\bm\beta)\),
\(M_{z,i}(\alpha,\bm\beta)\), and \(M_{1,i}(\alpha,\bm\beta)\) by
\(o_p(1)\). Hence
\[
\widehat D
=
\frac1n\sum_{i=1}^{n}
\left\{
M_{zx,i}(\widehat\alpha_i,\widehat{\bm\beta})
-M_{z,i}(\widehat\alpha_i,\widehat{\bm\beta})
\widehat{\bm\gamma}_i'
\right\}
+o_p(1).
\]
Moreover, Theorem 4.1
{} and Lemma \ref{lem:alpha_profile} imply that
there exists a deterministic sequence \(\rho_n\downarrow0\) such that
\[
P\left(
\max_{i\le n}|\widehat\alpha_i-\alpha_{i0}|
+\|\widehat{\bm\beta}-\bm\beta_0\|
\le\rho_n
\right)\to1.
\]
For deterministic \((\alpha,\bm\beta)\) satisfying
\(|\alpha-\alpha_{i0}|+\|\bm\beta-\bm\beta_0\|\le\rho_n\), conditional
integration and the smoothness of \(f_i(\cdot\mid\bm q)\), together with
boundedness of \(\bm x_{it}\) and \(\bm z_{it}\), give uniformly in \(i\)
\[
\begin{aligned}
&\|M_{zx,i}(\alpha,\bm\beta)-M_{zx,i}(\alpha_{i0},\bm\beta_0)\|
+\|M_{z,i}(\alpha,\bm\beta)-M_{z,i}(\alpha_{i0},\bm\beta_0)\|\\
&\qquad
+|M_{1,i}(\alpha,\bm\beta)-M_{1,i}(\alpha_{i0},\bm\beta_0)|
\le C\rho_n.
\end{aligned}
\]
Evaluating this uniform deterministic bound on the preceding
high-probability event yields
\[
\widehat D
=
\frac1n\sum_{i=1}^{n}
\left\{
M_{zx,i}(\alpha_{i0},\bm\beta_0)
-M_{z,i}(\alpha_{i0},\bm\beta_0)\widehat{\bm\gamma}_i'
\right\}
+o_p(1).
\]
By the same argument as in \eqref{eq: gamma consistency},
\(
\max_{i\le n}\|\widehat{\bm\gamma}_i-\bm\gamma_i\|=o_p(1).
\)
Since \(M_{z,i}(\alpha_{i0},\bm\beta_0)=O(1)\) uniformly in \(i\), it follows that
\[
\frac1n\sum_{i=1}^{n}
\left\{
M_{zx,i}(\alpha_{i0},\bm\beta_0)
-M_{z,i}(\alpha_{i0},\bm\beta_0)\widehat{\bm\gamma}_i'
\right\}
=
\frac1n\sum_{i=1}^{n}
\E\left[
K_{h_n}(u_{it})\bm z_{it}(\bm x_{it}-\bm\gamma_i)'
\right]
+o_p(1).
\]
Finally, the kernel bias expansion in Lemma \ref{lem:smoothing} gives, uniformly in \(i\),
\begin{align*}
&\E\left[K_{h_n}(u_{it})\bm z_{it}(\bm x_{it}-\bm\gamma_i)'\right]=
\E\left[f_i(0\mid\bm q_{it})\bm z_{it}
(\bm x_{it}-\bm\gamma_i)'\right]+O(h_n^r).
\end{align*}
Therefore,
\(
\widehat D
=
\frac1n\sum_{i=1}^{n}
\E\left[f_i(0\mid\bm q_{it})\bm z_{it}(\bm x_{it}-\bm\gamma_i)'\right]
+o_p(1)
=
D+o_p(1).
\)

Before comparing the feasible and infeasible HAC matrices, we establish
the uniform consistency and boundedness of the estimated residualization
terms. We use the same uniform kernel law and deterministic-neighborhood
argument for \(M_{1,i}\) and \(M_{z,i}\), together with the kernel bias
expansion. Define
\[
\widehat{\bm g}_i=\frac1T\sum_{t=1}^T
K_{h_n}(\widehat u_{it})\bm z_{it},
\qquad
\bm g_i=\E[f_i(0\mid\bm q_{it})\bm z_{it}].
\]
Lemma \ref{lem:alpha_max}, Theorem 4.1
{}, and boundedness of
\(\bm x_{it}\) imply that the random evaluation point
\((\widehat\alpha_i,\widehat{\bm\beta})\) is within
\(O_p(\sqrt{\log(nT)/T})\) of \((\alpha_{i0},\bm\beta_0)\), uniformly in
\(i\). The resulting population-map error is of the same order and is
absorbed by the kernel stochastic rate because \(h_n\to0\). Thus, we obtain
\[
\max_{i\le n}|\widehat f_i-f_i(0)|
+\max_{i\le n}\|\widehat{\bm g}_i-\bm g_i\|
=
O_p\left(\frac1{h_n}\sqrt{\frac{\log(nT)}{T}}\right)
+O(h_n^r).
\]
Because \(\inf_i f_i(0)>0\), it follows that, with probability
approaching one, \(\inf_i\widehat f_i>\underline f/2\) and
\[
\max_{i\le n}\|\widehat{\bm\delta}_i-\bm\delta_i\|
=
O_p\left(\frac1{h_n}\sqrt{\frac{\log(nT)}{T}}\right)+O(h_n^r).
\]
In particular, \(\max_i\|\widehat{\bm\delta}_i\|=O_p(1)\) and the
condition of the theorem implies
\begin{align}\label{eq:ellndelta}\ell_n\max_i\|\widehat{\bm\delta}_i-\bm\delta_i\|=o_p(1).\end{align}

We now prove \(\widehat{\Omega}\stackrel{p}{\to}\Omega\). Recall that
\(\widehat{\psi}_{it}=(\bm{z}_{it}-\widehat{\bm{\delta}}_{i})\{1(\widehat{u}_{it}\leq0)-\tau\}\)
and \(\psi_{it}=(\bm{z}_{it}-\bm{\delta}_{i})\{1(u_{it}\leq0)-\tau\}\).
We compare the feasible HAC estimator directly with its population
truncated target. Let
\[
\Omega_{n}^{(m)}
=
\frac{1}{n}\sum_{i=1}^{n}
\left\{
\E[\psi_{i1}\psi_{i1}']
+
\sum_{j=1}^{\ell_n}
k(j/(\ell_n+1))
\left(1-\frac{j}{T}\right)
\bigl(
\E[\psi_{i1}\psi_{i,1+j}']
+
\E[\psi_{i,1+j}\psi_{i1}']
\bigr)
\right\}.
\]
It is enough to prove
\(\widehat{\Omega}-\Omega_n^{(m)}=o_p(1)\) and
\(\Omega_n^{(m)}-\Omega_n=o(1)\).
The second statement follows from the uniform absolute summability of
the autocovariances. Indeed, since \(\psi_{it}\) is bounded and the
time-series process is exponentially beta-mixing, there exists a
summable sequence \(\{a_j\}_{j\geq1}\) such that, uniformly in \(i\),
\[
\|\E[\psi_{i1}\psi_{i,1+j}']\|
+
\|\E[\psi_{i,1+j}\psi_{i1}']\|
\leq a_j,
\qquad
\sum_{j=1}^{\infty}a_j<\infty.
\]
Therefore,
\[
\begin{aligned}
\|\Omega_n^{(m)}-\Omega_n\|
&\leq
C\sum_{j>\ell_n}a_j
+
\sum_{j=1}^{\ell_n}
|k(j/(\ell_n+1))-1|a_j
\to0,
\end{aligned}
\]
where the conclusion follows from dominated convergence.

It remains to prove \(\widehat{\Omega}-\Omega_{n}^{(m)}=o_p(1)\).
Let \(\widetilde{\Omega}\) be the same HAC estimator as \(\widehat{\Omega}\),
but with \(\widehat{\psi}_{it}\) replaced by
\(\psi_{it}=(\bm z_{it}-\bm\delta_i)\{1(u_{it}\le0)-\tau\}\). Then
\(\widehat{\Omega}-\Omega_{n}^{(m)}
=(\widehat{\Omega}-\widetilde{\Omega})
+
(\widetilde{\Omega}-\Omega_{n}^{(m)})\). We first show that
\(\widetilde{\Omega}-\Omega_n^{(m)}=o_p(1)\).
Write
\(
\widetilde{\Omega}
=
\frac1n\sum_{i=1}^n\widetilde{\Omega}_i,
\)
where
\[
\widetilde{\Omega}_i
=
\frac1T\sum_{t=1}^T\psi_{it}\psi_{it}'
+
\sum_{j=1}^{\ell_n}
k(j/(\ell_n+1))
\frac1T\sum_{t=1}^{T-j}
\left(
\psi_{it}\psi_{i,t+j}'
+
\psi_{i,t+j}\psi_{it}'
\right).
\]
By stationarity and the definition of \(\Omega_n^{(m)}\),
\(
\E[\widetilde{\Omega}]
=
\Omega_n^{(m)}.
\)

Since \(\psi_{it}\) is uniformly bounded and exponentially
beta-mixing, a standard four-index covariance calculation gives
\[
\sup_{i\le n}
\E\left[
\left\|
\widetilde{\Omega}_i-\E[\widetilde{\Omega}_i]
\right\|_F^2
\right]
\leq
C\frac{\ell_n^2}{T}.
\]
Indeed, after expanding the square, each covariance involves two
products indexed by
\((t,t+j)\) and \((s,s+k)\), with \(j,k\leq\ell_n\).
The contribution from overlapping or neighboring index sets is
\(O(T\ell_n^2)\), while the contribution from separated index sets
is also \(O(T\ell_n^2)\) by the covariance inequality and the
geometric decay of the beta-mixing coefficients. Division by \(T^2\)
gives the displayed bound.

Because the processes are independent across \(i\),
\[
\begin{aligned}
\E\left[
\left\|
\widetilde{\Omega}-\Omega_n^{(m)}
\right\|_F^2
\right]
&=
\frac1{n^2}\sum_{i=1}^n
\E\left[
\left\|
\widetilde{\Omega}_i-\E[\widetilde{\Omega}_i]
\right\|_F^2
\right] \leq
C\frac{\ell_n^2}{nT}
=o(1).
\end{aligned}
\]
Here,
\(
\frac{\ell_n^2}{T}
=
\frac{\ell_n^2\log(nT)}{Th_n^2}
\frac{h_n^2}{\log(nT)}
\to0
\)
under the condition of the theorem. Hence, by Markov's inequality,
\(
\widetilde{\Omega}-\Omega_n^{(m)}=o_p(1).
\)

Now, we want to show \(\widehat{\Omega}-\widetilde{\Omega}=o_{p}(1)\).
Write
\[
\begin{aligned}
\widehat{\psi}_{it}-\psi_{it}
&=(\bm{z}_{it}-\widehat{\bm{\delta}}_{i})\{1(\widehat{u}_{it}\le0)-1(u_{it}\le0)\}-(\widehat{\bm{\delta}}_{i}-\bm{\delta}_{i})\{1(u_{it}\le0)-\tau\}.
\end{aligned}
\]
We consider the two terms separately. For the first term, note that
\[
|1(\widehat{u}_{it}\le0)-1(u_{it}\le0)|
\le
1\{|u_{it}|\le|\widehat{u}_{it}-u_{it}|\}.
\]
By Lemma \ref{lem:alpha_max}, Theorem 4.1
{}, and boundedness of \(\bm x_{it}\),
\(\max_{i,t}|\widehat{u}_{it}-u_{it}|=O_{p}((\log(nT)/T)^{1/2})\).
Since the density of \(u_{it}\) is uniformly bounded near zero, a similar
maximal argument and a uniform empirical-process bound for the indicator class gives
\[
\frac{1}{nT}\sum_{i=1}^{n}\sum_{t=1}^{T}|1(\widehat{u}_{it}\le0)-1(u_{it}\le0)|
=
O_{p}\left((\log(nT)/T)^{1/2}\right)+O_{p}((nT)^{-1/2}).
\]
Therefore, by boundedness of \(\bm{z}_{it}\), \(\widehat{\bm{\delta}}_{i}\),
and the scores,
\[
\begin{aligned}
&\left\|\sum_{j=0}^{\ell_n}k(j/(\ell_n+1))\frac{1}{nT}\sum_{i=1}^{n}\sum_{t=1}^{T-j}\Big[(\bm{z}_{it}-\widehat{\bm{\delta}}_{i})\{1(\widehat{u}_{it}\le0)-1(u_{it}\le0)\}\widehat{\psi}_{i,t+j}'\right.\\
&\left.\qquad\qquad\qquad\qquad\qquad+\widehat{\psi}_{it}\{1(\widehat{u}_{i,t+j}\le0)-1(u_{i,t+j}\le0)\}(\bm{z}_{i,t+j}-\widehat{\bm{\delta}}_{i})'\Big]\right\|\\
&\qquad\le
C\ell_n\left\{
O_{p}\left((\log(nT)/T)^{1/2}\right)+O_{p}((nT)^{-1/2})
\right\}
=o_{p}(1),
\end{aligned}
\]
provided \(\ell_n\sqrt{\log(nT)/T}\to0\). The terms containing products of
two indicator differences are bounded in the same way and are also
\(o_p(1)\).

It remains to handle the replacement of \(\bm{\delta}_{i}\) by
\(\widehat{\bm{\delta}}_{i}\) after the indicators have been fixed at
\(1(u_{it}\le0)\). Expanding the product gives, for each lag \(j\),
\[
\begin{aligned}
&(\bm{z}_{it}-\widehat{\bm{\delta}}_{i})\{1(u_{it}\le0)-\tau\}\{1(u_{i,t+j}\le0)-\tau\}(\bm{z}_{i,t+j}-\widehat{\bm{\delta}}_{i})'-\psi_{it}\psi_{i,t+j}'\\
&= -(\widehat{\bm{\delta}}_{i}-\bm{\delta}_{i})\{1(u_{it}\le0)-\tau\}\psi_{i,t+j}'-\psi_{it}\{1(u_{i,t+j}\le0)-\tau\}(\widehat{\bm{\delta}}_{i}-\bm{\delta}_{i})'\\
&\quad+(\widehat{\bm{\delta}}_{i}-\bm{\delta}_{i})(\widehat{\bm{\delta}}_{i}-\bm{\delta}_{i})'\{1(u_{it}\le0)-\tau\}\{1(u_{i,t+j}\le0)-\tau\}.
\end{aligned}
\]
Thus the corresponding HAC contribution is bounded by
\[
\begin{aligned}
& C\max_{i}\|\widehat{\bm{\delta}}_{i}-\bm{\delta}_{i}\|\frac{1}{n}\sum_{i=1}^{n}\left\Vert \frac{1}{T}\sum_{t=1}^{T}\{1(u_{it}\le0)-\tau\}\psi_{it}'\right.\\
& \qquad\left.+\sum_{j=1}^{\ell_n}k(j/(\ell_n+1))\frac{1}{T}\sum_{t=1}^{T-j}\Big[\{1(u_{it}\le0)-\tau\}\psi_{i,t+j}'+\{1(u_{i,t+j}\le0)-\tau\}\psi_{it}'\Big]\right\Vert \\
& \quad+C\max_{i}\|\widehat{\bm{\delta}}_{i}-\bm{\delta}_{i}\|^{2}\frac{1}{n}\sum_{i=1}^{n}\left|\frac{1}{T}\sum_{t=1}^{T}\{1(u_{it}\le0)-\tau\}^{2}\right.\\
& \qquad\left.+\sum_{j=1}^{\ell_n}k(j/(\ell_n+1))\frac{1}{T}\sum_{t=1}^{T-j}\{1(u_{it}\le0)-\tau\}\{1(u_{i,t+j}\le0)-\tau\}\right|.
\end{aligned}
\]
For the present purpose it is enough to use the crude bound that the two
HAC sums in braces are \(O_p(\ell_n)\), since all summands and HAC weights
are uniformly bounded. Hence
\[
\|\widehat{\Omega}-\widetilde{\Omega}\|
\le
O_p(\ell_n)\max_{i}\|\widehat{\bm{\delta}}_{i}-\bm{\delta}_{i}\|
+
O_p(\ell_n)\max_{i}\|\widehat{\bm{\delta}}_{i}-\bm{\delta}_{i}\|^{2}
+
o_p(1).
\]
By the residualization rate established above in \eqref{eq:ellndelta}, we have
\(\widehat{\Omega}-\widetilde{\Omega}=o_p(1)\).
Combining this with
\(\widetilde{\Omega}-\Omega_n^{(m)}=o_p(1)\),
\(\Omega_n^{(m)}-\Omega_n=o(1)\), and Assumption (A10)
gives \(\widehat{\Omega}\stackrel{p}{\to}\Omega\).
The same argument applies to the preliminary one-step estimator computed with the fixed matrix $W_0$:
Theorem 4.1
{} first applies with $\widehat W=W_0$, and all residual
and kernel rates used above hold for the resulting preliminary residuals.
Consequently, $\bar\Omega\stackrel{p}{\to}\Omega$ and hence
$\bar\Omega^{-1}\stackrel{p}{\to}\Omega^{-1}$. This verifies (A7) for
the efficient second step without assuming consistency of a covariance
matrix constructed from an arbitrary preliminary estimator.

Finally, combining \(\widehat{D}\stackrel{p}{\to}D\),
\(\widehat{W}\stackrel{p}{\to}W\), and \(\widehat{\Omega}\stackrel{p}{\to}\Omega\), we obtain
\[
\widehat{\Sigma}_{\beta}=(\widehat{D}'\widehat{W}\widehat{D})^{-1}\widehat{D}'\widehat{W}\widehat{\Omega}\widehat{W}\widehat{D}(\widehat{D}'\widehat{W}\widehat{D})^{-1}\stackrel{p}{\to}\Sigma_{\beta}.
\]
For the efficient two-step estimator, the alternative estimator
\(\widehat{\Sigma}_{\beta}=(\widehat{D}'\widehat{\Omega}^{-1}\widehat{D})^{-1}\)
converges to \((D'\Omega^{-1}D)^{-1}\).
The studentized statistic follows directly from Theorem 4.1
{},
the consistency of \(\widehat{\Sigma}_{\beta}\), and Slutsky's theorem.
\end{proof}

\clearpage
\section{Monte Carlo Settings and Additional Results}
We consider the panel quantile model
\[
y_{it}
=
\alpha_i+\beta_\tau x_{it}+u_{it},
\qquad
i=1,\ldots,n,\quad t=1,\ldots,T,
\]
where \(\alpha_i\) is an individual fixed effect and \(\beta_\tau\) is
the common quantile slope parameter. The endogenous regressor is generated
according to
\begin{align}
x_{it}
=
\pi_z(z_{1,it}+z_{2,it})+v_{it},
\qquad
\pi_z=1,\label{eq: x dgp}
\end{align}
where
\(
\bm z_{it}=(z_{1,it},z_{2,it})'
\)
is the two-dimensional instrument vector. For \(k=1,2\), let
\[
\widetilde z_{k,it}=\rho_z \widetilde z_{k,i,t-1}+\sqrt{1-\rho_z^2}\,\xi_{k,it},
\qquad
z_{k,it}=2\Phi(\widetilde z_{k,it})-1,
\]
where \(\rho_z=0.50\). The initial values \(\widetilde z_{k,i0}\) and innovations
\(\xi_{k,it}\) are mutually independent standard normal random variables.
Thus, the two instrument processes are independent of each other and
across individuals, have \(U(-1,1)\) marginals, and exhibit serial
dependence.

The outcome disturbance follows the stationary AR(1) process
\(
u_{it}
=
\rho u_{i,t-1}+\widetilde\varepsilon_{it},\ 
\rho=0.30,
\)
where
\(
\widetilde\varepsilon_{it}
=
\varepsilon_{it}-q_\tau
\)
and
\(
q_\tau
=
\frac{(1-\rho)\Phi^{-1}(\tau)}
{\sqrt{1-\rho^2}}.
\)
The initial value is drawn from the corresponding stationary distribution.
Since the stationary distribution of \(u_{it}\) has mean
\(-q_\tau/(1-\rho)\) and variance \(1/(1-\rho^2)\), this choice ensures
that
\(
\Q_\tau(u_{it})=0.
\)
For each value of \(\tau\), a separate sample is generated using the
corresponding value of \(q_\tau\).

To generate a bounded, serially dependent regressor disturbance in \eqref{eq: x dgp}, let
\[
\widetilde v_{it}
=\rho_v \widetilde v_{i,t-1}
+\sqrt{\pi}\,\sigma_{uv}\varepsilon_{it}
+\sqrt{1-\rho_v^2-\pi\sigma_{uv}^2}\,\eta_{it},
\qquad \rho_v=0.30,
\]
and set \(v_{it}=2\Phi(\widetilde v_{it})-1\). The pairs
\((\varepsilon_{it},\eta_{it})\) are standard bivariate normal with
independent components, independently across individuals and time.
The initial pair \((u_{i0},\widetilde{v}_{i0})\) is drawn jointly from the stationary
Gaussian distribution of the two recursions, independently of subsequent
innovations. Since \(\rho_v^2+\pi\sigma_{uv}^2<1\) for the parameter
values considered, we have \(\widetilde v_{it}\sim N(0,1)\). Hence, by Stein's lemma,
\[
v_{it}\sim U(-1,1),
\qquad
\operatorname{Cov}(\varepsilon_{it},v_{it})=\sigma_{uv},
\qquad
\operatorname{Cov}(u_{it},v_{it})
=\frac{\sigma_{uv}}{1-\rho\rho_v}.
\]
The sequence \(\{\bm z_{it}\}_{t\in\mathbb Z}\) is independent of the
entire process \(\{(u_{is},\widetilde v_{is}):s\in\mathbb Z\}\) and of \(\alpha_i\).
Consequently, \(\bm z_{it}\) is independent of \(u_{it}\), whereas
\(x_{it}\) is correlated with \(u_{it}\) whenever
\(\sigma_{uv}\neq0\). In particular,
\(
P(u_{it}\leq0\mid\bm z_{it},\alpha_i)=\tau
\)
and therefore
\(
\E\left[
\bm z_{it}\{1(u_{it}\leq0)-\tau\}
\right]
=
\bm 0.
\)
Thus, both instruments satisfy the quantile moment restriction, while the
regressor is endogenous when \(\sigma_{uv}\neq0\).

We consider
\(
\sigma_{uv}\in\{0,0.50\},
\)
corresponding to exogenous and endogenous designs. The individual effects
are generated as
\(\alpha_i\overset{\mathrm{iid}}{\sim}U(-1,1)\), and the true slope is
\(\beta_\tau=1\). We examine
\[
\tau\in\{0.25,0.50\},
\qquad
(n,T)\in\{20,50,100\}\times\{20,50,100\}.
\]
Although the centered asymptotic normality result requires \(n/T\to0\),
we report all nine combinations of \(n\) and \(T\) to examine performance
outside the most favorable asymptotic configurations. For each design, the
results are based on \(1000\) Monte Carlo replications.

For each replication, we compute four estimators. The proposed
smoothed panel GMM estimator, denoted SP-GMM, profiles the individual
effects by solving the \(n\) smoothed individual quantile restrictions
for each candidate value of \(\beta\). It then forms the two-dimensional
profiled instrument moment. The first GMM step uses the identity weighting
matrix, while the second step uses the inverse of a Bartlett HAC estimator
constructed from the residualized profiled score. The final covariance
estimator uses the unsmoothed indicator score and the residualized
instrument \(\bm z_{it}-\widehat{\bm\delta}_i\), as described in
Section 5 of the main paper. For the panel sizes considered, the Bartlett
truncation lag is \(\ell_n=\max\{1,\lfloor T^{1/5}\rfloor\}\).\footnote{The first step is
initialized by the smoothed MM estimator, and the second by the first-step
estimate.}

The three comparison estimators are the conventional fixed-effects
quantile regression estimator (FE-QR), its smoothed counterpart (FE-SQR),
and a cross-sectional IVQR estimator based on the unit-level time averages.
The two fixed-effects estimators do not correct for endogeneity, whereas
the cross-sectional IVQR estimator uses both time-averaged instruments but
does not exploit the original panel quantile restrictions.

For the smoothed estimators, we use
\[
h_n=c_\tau T^{-2/7},
\qquad
c_\tau
=
\left[
\frac{4.5\,\phi\{\Phi^{-1}(\tau)\}^{4}}
{\left\{2\Phi^{-1}(\tau)^{2}+1\right\}^{2}}
\right]^{1/5},
\]
where \(\phi\) and \(\Phi\) denote the standard normal density and
distribution functions, respectively. The constant \(c_\tau\) is the normal-reference constant frequently used in quantile
regression bandwidth selection. Moreover, following the simulation practice in the smoothed panel quantile and
smoothed GMM literature, see, e.g., \citet{GalvaoKato16},
\citet{KaplanSun17}, and \citet{deCastroGalvaoKaplanLiu19}, we apply the
fourth-order kernel
\[
K(u)
=
\frac{105}{64}
(1-5u^2+7u^4-3u^6)1\{|u|\leq1\}.
\]
We report bias, root mean squared error (RMSE), and empirical coverage
probability (CP). Coverage is calculated from the nominal \(95\%\) normal
confidence interval.

Table \ref{tab:simulation_tau050_new} reports the median-quantile results.
Under exogeneity, all three panel estimators have negligible bias, and
their RMSE decreases as \(n\) and \(T\) increase. FE-QR and FE-SQR have
slightly smaller RMSE than SP-GMM because instrumenting is unnecessary
when \(\sigma_{uv}=0\). The inference procedures for FE-QR and FE-SQR
also produce coverage probabilities close to the nominal level: coverage
ranges from \(0.933\) to \(0.974\) for FE-QR and from \(0.950\) to
\(0.981\) for FE-SQR. The corresponding range for SP-GMM is \(0.928\)
to \(0.973\). Thus, when the regressor is exogenous, all three panel
methods perform well, with FE-QR and FE-SQR enjoying a modest efficiency
advantage.

\begin{table}[t!]
\centering
\caption{Monte Carlo Results, \(\tau=0.50\)}\label{tab:simulation_tau050_new}
\singlespacing
\resizebox{\textwidth}{!}{%
\begin{tabular}{lllcccccccc}
\toprule
 &  &  & \multicolumn{2}{c}{SP-GMM}
 & \multicolumn{2}{c}{FE-QR}
 & \multicolumn{2}{c}{FE-SQR}
 & \multicolumn{2}{c}{CS-IVQR}
 \tabularnewline
\cline{4-5}\cline{6-7}\cline{8-9}\cline{10-11}
Statistic & \(T\) & \(n\)
& \multicolumn{2}{c}{\(\sigma_{uv}\)}
& \multicolumn{2}{c}{\(\sigma_{uv}\)}
& \multicolumn{2}{c}{\(\sigma_{uv}\)}
& \multicolumn{2}{c}{\(\sigma_{uv}\)}
\tabularnewline
 &  &  & \(0.00\) & \(0.50\)
 & \(0.00\) & \(0.50\)
 & \(0.00\) & \(0.50\)
 & \(0.00\) & \(0.50\)
\tabularnewline
\hline
\multirow{9}{*}{Bias}
& 20 & 20 & 0.000 & -0.005 & 0.002 & 0.550 & 0.003 & 0.551 & 0.042 & -0.037 \tabularnewline
&    & 50 & 0.000 & -0.004 & 0.000 & 0.554 & 0.001 & 0.554 & -0.003 & -0.006 \tabularnewline
&    & 100 & -0.001 & -0.003 & 0.000 & 0.553 & -0.001 & 0.553 & 0.009 & 0.008 \tabularnewline
& 50 & 20 & -0.002 & 0.000 & 0.000 & 0.541 & 0.000 & 0.541 & 0.006 & -0.013 \tabularnewline
&    & 50 & 0.003 & -0.001 & 0.003 & 0.542 & 0.003 & 0.542 & -0.001 & 0.001 \tabularnewline
&    & 100 & 0.000 & 0.000 & 0.000 & 0.541 & 0.000 & 0.541 & -0.026 & -0.006 \tabularnewline
& 100 & 20 & 0.001 & 0.000 & 0.000 & 0.537 & 0.000 & 0.537 & 0.008 & -0.118 \tabularnewline
&     & 50 & -0.001 & 0.000 & -0.001 & 0.538 & -0.001 & 0.538 & -0.019 & -0.010 \tabularnewline
&     & 100 & -0.001 & 0.000 & 0.000 & 0.536 & 0.000 & 0.536 & -0.005 & -0.016 \tabularnewline
\hline
\multirow{9}{*}{RMSE}
& 20 & 20 & 0.094 & 0.094 & 0.073 & 0.553 & 0.071 & 0.553 & 0.734 & 0.681 \tabularnewline
&    & 50 & 0.058 & 0.060 & 0.044 & 0.555 & 0.043 & 0.555 & 0.443 & 0.442 \tabularnewline
&    & 100 & 0.041 & 0.041 & 0.032 & 0.553 & 0.031 & 0.553 & 0.323 & 0.321 \tabularnewline
& 50 & 20 & 0.057 & 0.058 & 0.047 & 0.542 & 0.046 & 0.542 & 0.985 & 1.039 \tabularnewline
&    & 50 & 0.036 & 0.036 & 0.029 & 0.543 & 0.029 & 0.542 & 0.682 & 0.672 \tabularnewline
&    & 100 & 0.025 & 0.025 & 0.021 & 0.542 & 0.021 & 0.542 & 0.508 & 0.512 \tabularnewline
& 100 & 20 & 0.039 & 0.040 & 0.032 & 0.537 & 0.031 & 0.537 & 1.358 & 1.413 \tabularnewline
&     & 50 & 0.025 & 0.025 & 0.021 & 0.538 & 0.020 & 0.538 & 0.956 & 0.931 \tabularnewline
&     & 100 & 0.018 & 0.019 & 0.015 & 0.536 & 0.014 & 0.536 & 0.676 & 0.715 \tabularnewline
\hline
\multirow{9}{*}{CP}
& 20 & 20 & 0.928 & 0.940 & 0.933 & 0.001 & 0.950 & 0.001 & 0.758 & 0.800 \tabularnewline
&    & 50 & 0.970 & 0.964 & 0.974 & 0.000 & 0.981 & 0.000 & 0.828 & 0.804 \tabularnewline
&    & 100 & 0.973 & 0.976 & 0.974 & 0.000 & 0.977 & 0.001 & 0.857 & 0.858 \tabularnewline
& 50 & 20 & 0.942 & 0.942 & 0.944 & 0.000 & 0.956 & 0.000 & 0.765 & 0.776 \tabularnewline
&    & 50 & 0.959 & 0.967 & 0.948 & 0.001 & 0.952 & 0.000 & 0.818 & 0.810 \tabularnewline
&    & 100 & 0.967 & 0.964 & 0.959 & 0.000 & 0.967 & 0.000 & 0.854 & 0.842 \tabularnewline
& 100 & 20 & 0.952 & 0.948 & 0.950 & 0.001 & 0.961 & 0.000 & 0.791 & 0.760 \tabularnewline
&     & 50 & 0.965 & 0.952 & 0.958 & 0.000 & 0.964 & 0.000 & 0.808 & 0.809 \tabularnewline
&     & 100 & 0.956 & 0.943 & 0.950 & 0.000 & 0.963 & 0.000 & 0.833 & 0.851 \tabularnewline
\bottomrule
\end{tabular}}
\par\smallskip {\small\textit{Notes:} The table reports bias, RMSE, and empirical coverage probability (CP)
for nominal \(95\%\) confidence intervals. The results are based on
\(1000\) Monte Carlo replications.}
\end{table}

The results change sharply under endogeneity. The biases of both FE-QR
and FE-SQR remain between approximately \(0.536\) and \(0.554\), even as
the panel dimensions increase. Their RMSE therefore remains between
approximately \(0.536\) and \(0.555\), and their coverage probabilities
are essentially zero. The similar biases of FE-QR and FE-SQR show that smoothing does not correct the inconsistency caused by 
regressor endogeneity. By contrast, the absolute bias of SP-GMM is at
most \(0.005\). Its RMSE decreases from \(0.094\) when
\((T,n)=(20,20)\) to \(0.019\) when \((T,n)=(100,100)\), while its
coverage probability ranges from \(0.940\) to \(0.976\), with some
finite-sample undercoverage and overcoverage.

For example, when \((T,n)=(50,50)\) and \(\sigma_{uv}=0.50\), the bias,
RMSE, and coverage probability of SP-GMM are \(-0.001\), \(0.036\), and
\(0.967\), respectively. The corresponding values are \(0.542\),
\(0.543\), and \(0.001\) for FE-QR, and \(0.542\), \(0.542\), and
\(0.000\) for FE-SQR. Thus, the conventional
fixed-effects quantile estimators become increasingly concentrated around
the wrong population value when the regressor is endogenous.

The cross-sectional IVQR estimator is approximately centered in several
designs but has much larger and more variable RMSE and exhibits persistent
undercoverage. Under endogeneity, its RMSE ranges from \(0.321\) to
\(1.413\). Collapsing the
panel into unit-level time averages discards the within-unit time-series
information exploited by SP-GMM, resulting in substantially less precise
estimation.

\bibliographystyle{chicago}
\bibliography{spqriv}

\end{document}